\documentclass[11pt]{article}
\usepackage{graphicx, psfrag,epsf}
\usepackage{setspace}
\usepackage{amsthm,amsmath,bbm,amsfonts,amssymb, mathrsfs}
\usepackage{enumerate}
\usepackage{cancel}
\usepackage{hyperref}
\usepackage{natbib}
\usepackage{soul}
\usepackage{subcaption} 
\usepackage{booktabs}
\usepackage{adjustbox} 
\usepackage{etoolbox}  
\usepackage{accents}
\usepackage{apptools}
\usepackage{float}
\usepackage{scalerel,stackengine}
\usepackage{mathtools}
\usepackage{cleveref}
\usepackage{caption}
\usepackage{rotating}

\usepackage[status=final,inline,index]{fixme}
\fxsetup{theme=color,margin=false,mode=multiuser}
\FXRegisterAuthor{bw}{bwe}{BW}
\FXRegisterAuthor{dn}{dne}{\color{blue}DN}

\hypersetup{hidelinks,
	colorlinks=true,
	linkcolor=black,
	citecolor=black}

\newtheorem{assumption}{Assumption}
\AfterEndEnvironment{assumption}{\noindent\ignorespaces}
\newtheorem{proposition}{Proposition}
\AfterEndEnvironment{proposition}{\noindent\ignorespaces}
\newtheorem{theorem}{Theorem}
\AfterEndEnvironment{theorem}{\noindent\ignorespaces}
\newtheorem{corollary}{Corollary}
\AfterEndEnvironment{corollary}{\noindent\ignorespaces}

\AfterEndEnvironment{condition}{\noindent\ignorespaces}

\AfterEndEnvironment{lemma}{\noindent\ignorespaces}

\AtAppendix{\counterwithin{lemma}{section}}

\theoremstyle{definition}
\newtheorem{example}{Example}
\AfterEndEnvironment{example}{\noindent\ignorespaces}
\newtheorem{definition}{Definition}
\AfterEndEnvironment{definition}{\noindent\ignorespaces}

\AfterEndEnvironment{remark}{\noindent\ignorespaces}

\def \bz {\boldsymbol{z}}
\def \bZ {\boldsymbol{Z}}
\def \bA {\boldsymbol{A}}
\def \bY {\boldsymbol{Y}}

\def \bX {\boldsymbol{X}}

\newcommand{\defeq}{\vcentcolon=}

\stackMath
\newcommand\reallywidehat[1]{%
\savestack{\tmpbox}{\stretchto{%
  \scaleto{%
    \scalerel*[\widthof{\ensuremath{#1}}]{\kern-.6pt\bigwedge\kern-.6pt}%
    {\rule[-\textheight/2]{1ex}{\textheight}}
  }{\textheight}%
}{0.5ex}}%
\stackon[1pt]{#1}{\tmpbox}%
}
\newcommand{\blind}{1}

\begin{document}

\def\spacingset#1{\renewcommand{\baselinestretch}%
{#1}\small\normalsize} \spacingset{1}


\if1\blind
{
    \title{\bf Causal Inference with Misspecified Network Interference Structure}
\author{Bar Weinstein\thanks{barwein@mail.tau.ac.il}
    \;and Daniel Nevo\thanks{
    The authors gratefully acknowledge support from the Israel Science Foundation (ISF grant No. 827/21)}
\hspace{.2cm}\\
   Department of Statistics and Operations Research, Tel Aviv University}
  \maketitle
 \fi

\if0\blind
{
  \bigskip
  \bigskip
  \bigskip
  \begin{center}
    {\LARGE\bf Causal inference with misspecified network interference structure}
\end{center}
  \medskip
} \fi
\bigskip
\begin{abstract}
Under interference, the treatment of one unit may affect the outcomes of other units. Such interference patterns between units are typically represented by a network. Correctly specifying this network requires identifying which units can affect others -- an inherently challenging task. Nevertheless, most existing approaches assume that a known and accurate network specification is given. In this paper, we study the consequences of such misspecification.

We derive bounds on the bias arising from estimating causal effects using a misspecified network, showing that the estimation bias grows with the divergence between the assumed and true networks, quantified through their induced exposure probabilities. To address this challenge,
we propose a novel estimator that leverages multiple networks simultaneously and remains unbiased if at least one of the networks is correct, even when we do not know which one. Therefore, the proposed estimator provides robustness to network specification.  We illustrate key properties and demonstrate the utility of our proposed estimator through simulations and analysis of a social network field experiment.
\end{abstract}
\noindent%
\vfill

\newpage
\spacingset{1.2} 

\section{Introduction}
A common assumption in causal inference is that there is \emph{no interference}.
However, interference between units is present in many settings where units interact, resulting in the spread of treatment effects.
When relaxing the no-interference assumption, researchers typically represent the interference structure as a network, where nodes represent units and edges indicate pairwise interference. Researchers have to specify the network to estimate causal effects in such settings.
However, correctly specifying the interference network is often challenging due to the complex interactions between units that characterize interference scenarios. 

Consider two examples that illustrate this challenge. 
\citet{Paluck2016} studied the effects of an educational intervention within a student social network.
They constructed the network from questionnaires asking students to list up to ten friends they spend time with. This approach could misrepresent actual social interactions if students' responses were inaccurate or important relationships existed beyond the ten-friend limit.
In another study, \cite{Hayek2022} examined the indirect protective effect of parental vaccination on children's SARS-CoV-2 infection, assuming interference occurred only within households. Since infections can spread between households, this assumption overlooked potentially important community-level effects from other vaccinated individuals \citep{Halloran1991}. Despite such challenges in accurately specifying network interference structures, researchers typically treat these structures as unique and correctly specified \citep[e.g.,][]{ aronow_estimating_2017, Forastiere2020, Tchetgen2020, gao2023,ogburn2022}. 

We extend the exposure mapping framework \citep{Manski2013,ugander2013, aronow_estimating_2017} to explicitly address misspecified network interference structures. Network misspecification can be viewed as a distinct type of exposure mapping misspecification with its own unique consequences and implications. While previous work examined exposure mapping misspecification \citep{aronow_estimating_2017, Saevje2023}, it did not explicitly distinguish between misspecification of the mapping itself and misspecification of the underlying network.
We develop a formal framework highlighting that the correctly specified network interference structure may not be unique, that is, different networks can represent the same effective interference structure. We show that uniqueness emerges under specific constraints on exposure mapping and potential outcomes. 
Using this framework, we consider the settings of randomized experiments under networked interference and derive bounds on the estimation bias that occurs when an incorrect network is assumed.

To address the challenge of network misspecification, we propose a novel estimator that simultaneously incorporates multiple networks. We prove this estimator is robust to misspecification, namely, it remains unbiased if at least one of the networks correctly specifies the interference structure, even when we do not know which network is correct. 
We illustrate that this unbiasedness may come with a price of increased variance, where the magnitude of the increase depends on the number of networks used and their relative (dis)similarity. Additionally, we establish the estimator's theoretical properties under large-sample conditions, showing that it is both consistent and asymptotically normal under standard assumptions. 

The rest of the paper is organized as follows. Section~\ref{sec:lit.rev} reviews relevant literature. Section~\ref{sec:notation.ass.estimand} introduces notations and formalizes the problem. Section~\ref{sec:bias} reviews practical examples of misspecified networks and shows that commonly used estimators are biased when the network is misspecified. Section~\ref{sec:NMR} presents the novel network-misspecification-robust estimator. Section~\ref{sec:simulations} presents simulation studies that show the bias from network misspecification and the proposed estimator's bias-variance tradeoff.
Section~\ref{sec:data} analyzes a social network field experiment. Finally, Section~\ref{sec:disc} discusses the findings and potential areas for future research.

\section{Related literature}
\label{sec:lit.rev}
Previous research has proposed various methods for estimating causal effects when the interference network is uncertain or only partially measured. These methods typically either impute missing edges or assume a specific measurement error model.
\citet{bhattacharya2020causal} developed a causal discovery method for partial interference settings, focusing on networks with well-separated clusters but unknown within-cluster structures. \citet{Tortu2021} proposed imputing missing edges using a network model trained on observed edges. \citet{Egami2020} introduced a sensitivity analysis for settings with both online and offline networks, examining how unobserved offline networks affect causal estimates. \citet{Leung2022} extended the traditional neighborhood interference assumption by allowing interference effects to decay with network distance.
Under the linear-in-means model, 
\citet{boucher2021estimating} considered estimation when only a distribution of the network is known, and \citet{Griffith2021} analyzed the impact of edge censoring (see Example \ref{exmp:network_censored}). 

Building on the exposure mapping framework, 
\citet{Li2021} developed unbiased estimators for networks measured with random error, requiring specific measurement error models and at least three noisy network measurements.
\citet{Hardy2019} assumed a parametric model for the exposure mapping and
proposed an EM algorithm.
In comparison to both \citet{Li2021} and \citet{Hardy2019}, which assumed a specific network measurement error model and implicitly regarded the true network as unique, our approach acknowledges the possibility that the correct network is not unique and does not view the network specification problem as a measurement error problem. This perspective complements, rather than contradicts, previous approaches.

Notably, some causal effects under interference can be estimated without network data. \citet{Saevje2021} showed that the Expected Average Treatment Effect (EATE), which is an effect marginalized by other units' treatments, can be consistently estimated with the common design-based estimators, under limiting interference dependence between units. \citet{Yu2022} showed that the Total Treatment Effect (TTE) -- treating all units versus none -- can be unbiasedly estimated, under restrictions on the potential outcomes and the experimental design. TTE and EATE are closely related \citep{Saevje2021}. 
However, analyzing other causal estimands requires correct network measurements.


\section{Notations, assumptions and causal estimands}
\label{sec:notation.ass.estimand}

\subsection{Setup}
\label{subsec:setup}

Consider a population of $n$ units, indexed by $i=1,\dots,n$. Let $\bZ$ be the treatment assignment vector of the entire population and let $\mathcal{Z}$ denote the treatments' space which is assumed to be finite. Each unit has a function $Y_i: \mathcal{Z} \to \mathbb{R}$ denoting the \textit{potential outcomes}, that is,  $Y_i(\bz)$ is the outcome of $i$ when, possibly contrary to the fact, the population treatment is set to $\bz \in \mathcal{Z}$.  
In our framework,  $Y_i(\bz)$ are fixed, hence randomness arises solely from the assignment of $\bZ$. 

We focus on network interference that, for simplicity, is assumed to be represented by an undirected and unweighted network. Extensions to directed and weighted networks are possible with appropriate modifications. 
In the network, each node represents a unit and the edges indicate possible pairwise interference, as we define below.
We represent the network by its symmetric $n \times n$ adjacency matrix $\bA$, with $A_{ij} = 1$ only if an edge exists between units $i$ and $j$, and by convention $A_{ii} = 0$. 
Let $\mathcal{N}_i(\bA)=\{j: A_{ij}=1\}$ be the set of \textit{neighbors} of unit $i$. Let $\mathscr{A} \subseteq \{0,1\}^{n\times n}$ denote the space of all undirected and unweighted networks of size $n$.
We further assume that the treatments affect the outcomes only through values of an exposure mapping $f: \mathcal{Z} \times \mathscr{A} \to \mathcal{C} = \{c_1,\ldots, c_L\}$ which maps from the treatments and networks space into $L=\lvert \mathcal{C} \rvert$ different discrete exposure levels.
We take the common neighborhood network interference assumption \citep{Forastiere2020, ogburn2022}, which states that interference occurs only between neighbors. Specifically, we assume that for any unit $i$ the values of $f$ depend only on the treatments assigned to its neighbors. Let $\bA_i$ be the $i$-th row of $\bA$. 
We denote the exposures by $f(\bz,\bA_i)$.

Turning to the treatments' assignment, we assume that the experimental design $\Pr(\bZ = \bz)$ is known. Let $\mathbb{I}\{\cdot\}$ denote the indicator function. Define the probability that unit $i$ has exposure $c_\ell \in \mathcal{C}$ under $\bA \in \mathscr{A}$ by
$p_i^{(\bA)}(c_\ell) = \mathbb{E}_{\bZ}[\mathbb{I}\{f(\bZ,\bA_i) = c_\ell \}]$.
Calculating $p_i^{(\bA)}(c_\ell)$ is computationally intensive, but can be approximated (Web Appendix~F). 
The following definition is the exposure mapping analog of the standard positivity assumption. 
\begin{definition}[Positivity]
    \label{defi:positivity} We say that $\bA \in \mathscr{A}$ satisfies positivity if
    $p_i^{(\bA)}(c_\ell)>0$
    for all units $i=1,...,n$ and exposure values $c_\ell \in \mathcal{C}$. 
\end{definition}
 Given the experimental design and the exposure mapping, positivity is a property of the network.
Positivity may not hold for some networks. For instance, if $f$ indicates whether a unit and at least one of its neighbors are treated \citep{aronow_estimating_2017}, then if a unit is isolated ($\mathcal{N}_i(\bA) = \emptyset$), there will be a structural violation of positivity for some exposures.

\subsection{Correctly specified network}
\label{subsec:correctly.spec}

Assume that for each unit there exists a function $\widetilde{Y}_i: \mathcal{C} \to \mathbb{R}$ such that $\widetilde{Y}_i(c_\ell)$ is the outcome of unit $i$ when its exposure value is $c_\ell$. We denote $\widetilde{Y}_i(c_1),\dots,\widetilde{Y}_i(c_L)$ as the induced potential outcomes expressed in terms of exposure values.
To connect $\widetilde{Y}(\cdot)$ to $Y(\cdot)$, the researcher must specify a network that accurately represents the interference structure, as expressed in the following definition.
\begin{definition}[Correctly specified interference structure]
\label{defi:correct_spec_interference}
    For an exposure mapping $f$, we say that the interference structure is correctly specified by $\bA \in \mathscr{A}$, if $\bA$ satisfies Definition~\ref{defi:positivity}, and for all $\bz \in \mathcal{Z}$,
   \begin{equation*}
        \text{if}\; f(\bz,\bA_i) = c_\ell,\; \text{then}\; Y_i(\bz) = \widetilde{Y}_i(c_\ell), \quad i = 1,\ldots,n.
    \end{equation*}
\end{definition}
If some $\bA \in \mathscr{A}$ satisfies Definition~\ref{defi:correct_spec_interference}, then for any $\bz,\bz'$, if $f(\bz,\bA_i)=f(\bz',\bA_i)$ then $Y_i(\bz) = Y_i(\bz')$. The latter property is often called an \textit{exclusion restriction} condition \citep{Puelz2022}. Therefore, Definition~\ref{defi:correct_spec_interference} formalizes the role of the exposure mapping as a bridge between the network $\bA$ and treatments $\bz$ on one side and the potential outcomes on the other side. 
We assume there exists at least one network that satisfies Definition~\ref{defi:correct_spec_interference}.

\paragraph{Exposure Mapping Misspecification} A misspecified interference network represents a specific type of exposure mapping misspecification. 
Our framework explicitly separates between two components -- the assumed network $\bA$ and the mapping $f(\bz,\bA_i)$. Through this separation, we can see that exposure mapping misspecification can arise from two distinct sources: an incorrect mapping $f$ or a network $\bA$. 
Previous work \citep{aronow_estimating_2017, Saevje2023} studied exposure mapping misspecification without distinguishing between these sources. In contrast, we focus specifically on network misspecification while assuming the mapping $f$ is correct, as expressed in
Definition \ref{defi:correct_spec_interference}.

\paragraph{Network Uniqueness} Typically, it is explicitly or implicitly assumed that a unique network correctly specifies the interference structure \citep[e.g.,][]{aronow_estimating_2017, Li2021}. We show that uniqueness holds under further strong constraints on the exposure mapping and the potential outcomes (see Web Appendix~A for a formal statement and proof). 
Let $\bA^\ast$ denote any network that correctly specifies the interference structure. This $\bA^\ast$ can be unique or belong to an equivalence class $\mathscr{A}^\ast \subseteq \mathscr{A}$ of networks that yield equivalent interference structures, where $\mathscr{A}^\ast$ contains all networks satisfying Definition~\ref{defi:correct_spec_interference}.
Furthermore, while one might consider a \emph{minimal} class of correctly specified networks, which includes networks from $\mathscr{A}^\ast$ with the fewest edges, this minimal class is not necessarily a singleton (Web Appendix~A).
Under the sharp null
$\big(\widetilde{Y}_i(c_k)=\widetilde{Y}_i(c_\ell),\; \forall i,k,\ell \big)$, given any exposure value, all other potential outcomes $\widetilde{Y}(\cdot)$ are imputable \citep{Athey2018, Basse2019_rand} and any network that satisfies positivity (Definition~\ref{defi:positivity}) will correctly specify the interference structure.
\paragraph{Exposure Mapping Implications for Uniqueness}
Without the additional assumption of exposure mapping, any superset of a correctly specified network (i.e., networks with additional edges) would also correctly specify the interference structure. In the extreme, the fully connected network is always correct, implying that the network that correctly specifies the interference cannot be unique.
The exposure mapping framework fundamentally changes this property.
By Definition~\ref{defi:correct_spec_interference}, a superset of a correctly specified network may no longer be correct, and notably, even the fully connected network is not guaranteed to be correctly specified. 
Thus, while the exposure mapping framework reduces the number of effective potential outcomes through its summarizing property, it
also implies restrictions on the class $\mathscr{A}^\ast$ of correctly specified networks.

We denote by $\bY=\left(Y_1,\dots, Y_n\right)$ the observed outcomes vector, which we assume are related to the potential outcomes in the following manner.
\begin{assumption}[Consistency]
    \label{ass:consistency}
   The observed outcomes are generated from one of the potential outcomes by
   $
       Y_i = \sum_{j = 1}^{L}\mathbb{I}\{f(\bZ, \bA^{\ast}_i) = c_j\} \widetilde{Y}_i(c_j),\; i=1,\dots,n, \bA^\ast \in \mathscr{A}^\ast.
   $
\end{assumption}
 Even if $\mathscr{A}^\ast$ is not a singleton, all networks in it will result in the same observed outcomes. That is, the sum $ \sum_{j = 1}^{L}\mathbb{I}\{f(\bZ, \bA^{\ast}_i) = c_j\} \widetilde{Y}_i(c_j)$ is constant for any $\bA^{\ast} \in \mathscr{A}^\ast$. 

\subsection{Causal estimands}
\label{subsec:estimands}

To define causal effects under the above-described framework, we first define the mean potential outcomes $\mu(c_\ell) = \frac{1}{n}\sum_{i=1}^{n}\widetilde{Y}_i(c_\ell),\; c_\ell \in\mathcal{C}$.
Causal effects are defined as the difference in the mean potential outcomes, 
    $
    \tau(c_\ell, c_k) = \mu(c_\ell) - \mu(c_k).
    $
This definition is common in the literature \citep[e.g.,][]{ugander2013,aronow_estimating_2017, Forastiere2020}.

\section{Bias from using a misspecified network}
\label{sec:bias}
Let $\bA^{sp}$ be the network specified by the researchers. In this section, we study the bias resulting from using a misspecified network, i.e., when $\bA^{sp} \notin \mathscr{A}^\ast$. We first review common sources and types of network misspecification that can lead to incorrect interference structures.

\begin{figure}[h!]
    \centering
   \includegraphics[scale=0.4]{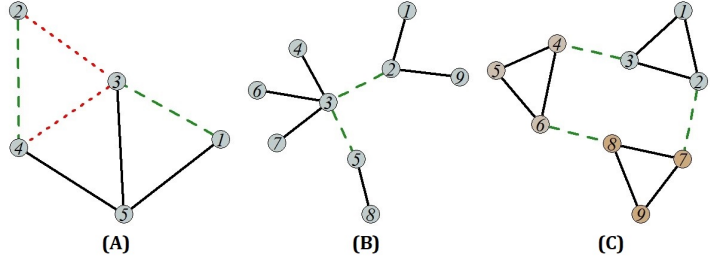}
    \caption{Schematic view of network misspecification. Edges in dashed lines are missing whereas edges in dotted lines are assumed to be present but should be removed. (A) Network with an incorrect list of edges. 
    (B) Network with edges censored at $K=3$. Node $3$ has five edges but two are censored ($2-3,3-5$).
    (C) Cross-clusters contamination with three clusters.}
    \label{fig:network_misspec_example}
\end{figure}

\begin{example}[Incorrect reporting of social connections]
\label{exmp:network_noise}
Networks are often measured from participant self-reported surveys listing frequently interacted friends \citep{Paluck2016, Cai2015} or through epidemiological contact tracing \citep{Nagarajan2020}. 
However, determining the interference structure through surveys can be susceptible to inaccuracies.
For instance, if participants omit friends they interact with or report non-relevant friends, the specified network may fail to reflect the actual interference structure.
A misspecified network due to incorrect reporting of social interactions is illustrated in  Figure~\hyperref[fig:network_misspec_example]{1(A)}.
\end{example}

\begin{example}[Censoring]
\label{exmp:network_censored}
Questionnaires often request participants to list their top $K>0$ friends, but this limitation can result in neglected social connections,
known as \textit{censoring} of edges \citep{Griffith2021}.
For example, \citet{Cai2015} and \citet{Paluck2016} asked participants to list five and ten friends, respectively. 
To assess the extent of censoring present, one can look at the percentage of participants that listed the maximum number of friends, which were $91\%$ in \citet{Cai2015} and $46\%$ in \citet{Paluck2016}.
An illustration of censoring can be seen in Figure~\hyperref[fig:network_misspec_example]{1(B)}.
\end{example}

\begin{example}[Reciprocity]
\label{exmp:reciprocity}
Undirected network edges are mutual, meaning if unit $j$'s treatment affects unit $i$, then $i$ also affects $j$.
When constructing undirected networks from questionnaires, researchers may define an edge if either participant names the other as a friend, or only if both do.
These two options will likely result in different network structures.  
\end{example}

\begin{example}[Temporality]
\label{exmp:repeated_measures}
Social interactions evolve over time, so observed networks often reflect only a ``snapshot". Networks are typically defined using data collected before treatment assignment, but using post-treatment data can yield different structures. \citet{Paluck2016} found that only $42.2\%$ of pre-intervention edges persisted a year later. Nonetheless, using a network that is measured post-treatment necessitates the assumption that treatment did not affect the network structure and further assumptions required by the dynamic nature of the problem.
\end{example}

\begin{example}[Cross-clusters contamination]
\label{exmp:contamination}
In partial interference settings, interference is assumed to occur only between units within the same cluster.
The resulting network consists of well-separated clusters, but contamination can occur between clusters, leading to unaccounted-for interference.
For example, 
\citet{Hayek2022} estimated the indirect effect of vaccination against SARS-CoV-2 while implicitly assuming that the protective effect was limited to households. However, if infection can occur outside the household, then the vaccination status of individuals from different households may affect household members, resulting in contamination between clusters.
The network structure of clusters with possible contamination is illustrated in  Figure~\hyperref[fig:network_misspec_example]{1(C)}. 
\end{example}

\subsection{Estimation bias}
Given the specified network $\bA^{sp}$, the mean potential outcomes $\mu(c_\ell)$ are often estimated by the Horvitz-Thompson (HT) estimator \citep{ugander2013, aronow_estimating_2017}
\begin{equation}
    \label{eq:HT_estimator}
    \hat{\mu}_{\bA^{sp}}(c_\ell) = \frac{1}{n}\sum_{i=1}^{n}\frac{\mathbb{I}\{f(\bZ,\bA^{sp}_i)=c_\ell\}}{p_i^{(\bA^{sp})}(c_\ell)} Y_i.
\end{equation}
Let $\widetilde{n}(\bA,c_\ell) \defeq \sum_{i=1}^{n}\frac{\mathbb{I}\{f(\bZ,\bA_i)=c_\ell\}}{p_i^{(\bA)}(c_\ell)}$. Alternatively, the Hajek estimator,
\begin{equation}
    \label{eq:hajek_estimator}
    \hat{\mu}^{H}_{\bA^{sp}}(c_\ell) = 
    \frac{1}{\widetilde{n}(\bA^{sp},c_\ell)}\sum_{i=1}^{n}\frac{ 
    \mathbb{I}\{f(\bZ,\bA^{sp}_i)=c_\ell\}}{p_i^{(\bA^{sp})}(c_\ell)} Y_i,
\end{equation}
is known to have better finite-sample accuracy \citep{Saerndal2003}. Subsequently, $\tau(c_\ell,c_k)$ is estimated by the plug-in HT estimator
$\hat{\tau}_{\bA^{sp}}(c_\ell,c_k) = \hat{\mu}_{\bA^{sp}}(c_\ell) - \hat{\mu}_{\bA^{sp}}(c_k)$, and similarly for the Hajek estimator $\hat{\tau}^H_{\bA^{sp}}$.
The researcher estimates the causal effects with $\bA^{sp}$, which, as previously indicated, may or may not be in $\mathscr{A}^\ast$. Namely, $\bA^{sp}$ might not correctly represent the interference structure. By replacing $Y_i$ in \eqref{eq:HT_estimator} with its definition under consistency (Assumption~\ref{ass:consistency}), we obtain that 
\begin{align}
\begin{split}
     \hat{\mu}_{\bA^{sp}}(c_\ell) &=
            \frac{1}{n}
	\sum_{i=1}^{n} \bigg[
	\underbrace{\frac{\mathbb{I}\{f(\bZ,\bA^{sp}_i)=c_\ell\}}{p_i^{(\bA^{sp})}(c_\ell)}}_{\text{Selection and weighting}} 
 \underbrace{\sum_{j=1}^{L}\mathbb{I}\{f(\bZ,\bA^\ast_i)=c_j\}\widetilde{Y}_i(c_j)}_{\text{Observation}}\bigg].
\end{split}
 \label{eq:HT_decomp}
 \end{align}
Eq.~\eqref{eq:HT_decomp} highlights that unit selection and weighting are based on $\bA^{sp}$, while the observed outcomes are generated according to a network in $\mathscr{A}^\ast$. Consequently, if $\bA^{sp} \notin \mathscr{A}^{\ast}$, estimation using $\bA^{sp}$ may lead to erroneous results --- either by selecting incorrect units or by applying incorrect weights to the observed outcomes. 

For any two networks $\bA, \bA' \in \mathscr{A}$, define the joint probability that unit $i$ is exposed to $c_\ell$ under $\bA$ and to $c_k$ under $\bA'$ by
\begin{equation}
\label{eq:joint_exposure_prob}
    \begin{split}
     p_i^{(\bA, \bA')}(c_\ell, c_k)
      &=
    \mathbb{E}_{\bZ}\Big[\mathbb{I}\Big\{\big(f(\bZ,\bA_i)=c_\ell \big) \cap \big(f(\bZ,\bA'_i)=c_k\big)\Big\}\Big].
    \end{split}
\end{equation}
\begin{assumption}[Bounded potential outcomes]
    \label{ass:bounded.po}
    There exists a constant $\kappa > 0$ such that $\big\lvert \widetilde{Y}_i(c_\ell) \big \rvert \leq \kappa$ for all $i=1,\ldots,n$ and $c_\ell \in \mathcal{C}$.
\end{assumption}
The following theorem derives bounds on the absolute bias of $\hat{\mu}_{\bA^{sp}}$.
\begin{theorem}
    \label{thm:ht_bias}
    Let $\bA^\ast$ be an arbitrarily chosen network from $\mathscr{A}^\ast$,
    and let $\bA^{sp} \in \mathscr{A}$ be a network  satisfying  Definition~\ref{defi:positivity}. Under Assumptions~\ref{ass:consistency}-\ref{ass:bounded.po}, for any $c_\ell \in \mathcal{C}$,
    \begin{equation*}\Big \lvert \mathbb{E}_{\bZ}\left[\hat{\mu}_{\bA^{sp}}(c_\ell)\right] - \mu(c_\ell) \Big \rvert 
    \leq 
        \frac{2\kappa}{n}
        \sum_{i=1}^{n}
        \Big[1 - p_i(c_\ell;\bA^\ast \mid c_\ell; \bA^{sp})\Big],
     \end{equation*}
  where $p_i(c_\ell; \bA^\ast \mid c_\ell; \bA^{sp}) = \frac{p_i^{(\bA^\ast, \bA^{sp})}(c_\ell, c_\ell)}{ p_i^{(\bA^{sp})}(c_\ell)}$ is the conditional probability that unit $i$ is exposed to $c_\ell$ under $\bA^\ast$ given it is exposed to $c_\ell$ under $\bA^{sp}$. Furthermore, this bound is sharp.
\end{theorem}
Theorem \ref{thm:ht_bias} shows that the bounds on the absolute bias of $\hat{\mu}_{\bA^{sp}}$ increase with the divergence of $\bA^{sp}$ from $\bA^\ast$, in terms of resulting exposure levels. Namely, the conditional probabilities $p_i(c_\ell; \bA^\ast \mid c_\ell; \bA^{sp})$ quantify how the extent of misspecification of $\bA^{sp}$ impacts the maximal bias. The difference between $\bA^{sp}$ and $\bA^\ast$ affects the bias only through their disagreement on the set of exposures. The absolute bias also increases with $\kappa$, the assumed bound of the potential outcomes. The maximal bias of the plug-in causal effects estimator $\hat{\tau}_{\bA^{sp}}$ follows from Theorem \ref{thm:ht_bias} and is given in Web Appendix~A. 

We also derive the exact bias of $\hat{\mu}_{\bA^{sp}}$ and $\hat{\tau}_{\bA^{sp}}$, which are found to be linear combinations of all potential outcomes with weights relating to the aforementioned conditional probabilities (Web Appendix~A). 
The following corollary states that the bias is zero when $\bA^{sp} \in \mathscr{A}^{\ast}$.
\begin{corollary}
\label{corollary:unbiased_HT}
Under the conditions stated in Theorem~\ref{thm:ht_bias}, if $\bA^{sp} \in \mathscr{A}^{\ast}$,   
$\mathbb{E}_{\bZ}\left[\hat{\mu}_{\bA^{sp}}(c_\ell)\right] =
        \mu(c_\ell)$ for all $c_\ell \in \mathcal{C}$.
\end{corollary}
 The corollary follows from the fact that if $\bA^{sp} \in \mathscr{A}^{\ast}$, then in Theorem~\ref{thm:ht_bias} we can choose  $\bA^\ast=\bA^{sp}$. Thus, the conditional probabilities are all equal to one, and the bound is equal to zero.
\cite{ugander2013} and  \cite{aronow_estimating_2017} proved a similar version of Corollary~\ref{corollary:unbiased_HT} without considering the class $\mathscr{A}^\ast$ nor the bounds shown in Theorem \ref{thm:ht_bias}. 
The Hajek estimator \eqref{eq:hajek_estimator} is biased even if $\bA^{sp}\in\mathscr{A}^\ast$, but the bias can be bounded (Web Appendix~B).


\section{Network-misspecification-robust estimator}
\label{sec:NMR}
As established in Section~\ref{sec:bias}, using a misspecified network may lead to biased estimation. 
We propose a solution for a common scenario where researchers observe several possible networks but are uncertain which, if any, correctly specifies the interference structure. Our proposed Network Misspecification Robust (NMR) estimator leverages multiple networks simultaneously, remaining unbiased if at least one network is correct.

Assume that researchers observe a collection  $\mathcal{A} =\big\{\bA^1,\ldots\,\bA^M\big\}$ of $M$ networks.
Define 
    $
    I_i^{(\mathcal{A})}(\bZ,c_{\ell}) =\prod_{\bA \in \mathcal{A}}\mathbb{I}\{f(\bZ,\bA_i)=c_{\ell}\},    
    $
to be the indicator that equals one only if the exposure value equals $c_\ell$ under each of the networks in $\mathcal{A}$. 
Extending \eqref{eq:joint_exposure_prob}, we define the joint probability that unit $i$ has exposure value $c_\ell$ under \emph{all} $\bA \in \mathcal{A}$ by
    $
    p_i^{(\mathcal{A})}(c_\ell) =
    \mathbb{E}_{\bZ}\left[
    I_i^{(\mathcal{A})}(\bZ,c_\ell) 
    \right]
    $.
Our proposed modified HT estimator of $\mu(c_\ell)$ that simultaneously utilizes the $M$ different networks is
\begin{equation}
   \label{eq:NMR_HT}
    \hat{\mu}_{\mathcal{A}}(c_\ell) =   
    \frac{1}{n}\sum_{i=1}^{n} 
    \frac{I_i^{(\mathcal{A})}(\bZ,c_{\ell})}{p_i^{(\mathcal{A})}(c_\ell)}Y_i.
\end{equation}
That is, $\hat{\mu}_{\mathcal{A}}(c_\ell)$ selects only units that has exposure value $c_\ell$ under \textit{all the networks in} $\mathcal{A}$ and weights them with the inverse of the joint probability $p_i^{(\mathcal{A})}(c_\ell)$. 
The estimator of $\tau(c_k, c_\ell)$ is the plug-in estimator
$\hat{\tau}_{\mathcal{A}}(c_\ell,c_k) =
\hat{\mu}_{\mathcal{A}}(c_\ell) - \hat{\mu}_{\mathcal{A}}(c_k)$.
The following theorem establishes the network misspecification robustness of the proposed estimator $\hat{\mu}_{\mathcal{A}}$.
\begin{theorem}
\label{thm:unbiased_of_multiple_robust_HT}
    Let $\mathcal{A}$ be a collection of $M$ networks such that each of the networks satisfies Definition~\ref{defi:positivity}. Under Assumption~\ref{ass:consistency}, if $\mathcal{A} \cap \mathscr{A}^\ast \ne \emptyset$, then
    $
    \mathbb{E}_{\bZ}\Big[\hat{\mu}_{\mathcal{A}}(c_\ell)\Big] = \mu(c_\ell)
    $
     for all $c_\ell \in \mathcal{C}$.
\end{theorem}
The key property of the estimator $\hat{\mu}_{\mathcal{A}}$ is that by selecting only units with the same exposure values under each of the networks in $\mathcal{A}$, we are guaranteed to observe the correct exposure value if one of the networks is correctly specified, but agnostic to which network it is. 
Accordingly, the plug-in estimator $\hat{\tau}_{\mathcal{A}}(c_\ell,c_k)$ is unbiased estimator of $\tau(c_\ell,c_k)$.
Similarly to $\hat{\mu}_{\mathcal{A}}$, we also propose the NMR Hajek estimator 
\begin{equation}
  \label{eq:NMR_hajek_esti}
    \hat{\mu}^{H}_{\mathcal{A}}(c_\ell) =   
   \frac{1}{\widetilde{n}(\mathcal{A},c_\ell)}\sum_{i=1}^{n}
    \frac{I_i^{(\mathcal{A})}(\bZ,c_{\ell}) }{p_i^{(\mathcal{A})}(c_\ell)}Y_i,
\end{equation}
where $\widetilde{n}(\mathcal{A},c_\ell) \defeq \sum_{i=1}^{n}
    \frac{I_i^{(\mathcal{A})}(\bZ,c_{\ell}) }{p_i^{(\mathcal{A})}(c_\ell)}$.
Note that $\hat{\mu}^{H}_{\mathcal{A}}$ selects the same subset of units as $ \hat{\mu}_{\mathcal{A}}$, but is biased since it is a ratio estimator. In our simulation study (Section~\ref{subsec:network_bias_simulations}), we found that both NMR estimators had a similar finite sample bias.
Building on previous work \citep{aronow2013conservative} based on Young's inequality, we derive a conservative variance estimator $\reallywidehat{Var}\big(\hat{\tau}_{\mathcal{A}}\big)$, that is, its expected value is not smaller than $Var_{\bZ}\big(\hat{\tau}_{\mathcal{A}}\big)$. The variance estimation of Hajek NMR is obtained similarly with Taylor linearization. Full details are provided in Web Appendix~C.
In Web Appendix~F, the conservativeness property is demonstrated via simulations.

The NMR estimators allow flexible combinations of multiple networks, but face a \emph{bias-variance tradeoff}. While including more networks can eliminate bias whenever at least one network is correct, it increases variance through the reduction in the number of units used in estimation and the decreased values of the joint probabilities $p_i^{(\mathcal{A})}$. This variance increase depends not only on how many networks are included, but also on how similar the networks are in terms of the induced exposure patterns -- networks with different edge sets can still yield nearly identical exposures. 
Section \ref{subsec:NMR_bias_var_sim} demonstrates this tradeoff empirically. 
We discuss practical guidelines for selecting $\mathcal{A}$ in Section~\ref{sec:disc}.


\subsection{Covariate adjustment}
\label{subsec:nmr_covar}

The NMR estimators can accommodate covariates $\bX_i$. The Hajek NMR estimator is equivalent to a weighted least squares (WLS) regression, where the outcomes are regressed on exposure indicators $I_i^{(\mathcal{A})}(\bZ,c_\ell)$ with weights $w_i = 1/p_i^{(\mathcal{A})}(c_\ell)$ \citep{Saerndal2003, aronow_estimating_2017}. This equivalence facilitates the straightforward inclusion of covariates in the WLS specification.
Moreover, a model-assisted approach using the difference estimator \citep{Saerndal2003, aronow_estimating_2017}, can be employed. This approach combines design-based estimation with model predictions, resembling the structure of doubly robust estimators in causal inference. See \citet{gao2023} for further analysis of model-based alternatives to design-based estimators and their associated variance estimation procedures.

\subsection{Asymptotic properties}
\label{subsec:nmr_asymptotic}

We establish asymptotic properties of the NMR estimators within a growing sequence of populations, building on recent research \citep{aronow_estimating_2017, Li2022, Saevje2023, ogburn2022}. Our analysis focuses on a collection $\mathcal{A}$ of $M$ networks containing at least one correctly specified network $(\mathcal{A} \cap \mathscr{A}^\ast \neq \emptyset)$. 
The asymptotic analysis comprises two key components: consistency and asymptotic normality. Consistency requires a weak dependence condition on units' pairwise exposures, mathematically expressed as the sum of exposure covariances having $o(n^2)$ convergence rate. To establish asymptotic normality, we construct a dependency graph that captures the exposure dependencies across the $M$ networks. This approach allows us to apply the Central Limit Theorem (CLT) developed by \citet{baldi1989normal} to our specific setting.
Additionally, we show that confidence intervals based on the conservative variance estimators 
    $
    \Big[\hat{\tau}_{\mathcal{A}}(c_\ell,c_k) \pm z_{1-\alpha/2} \sqrt{\reallywidehat{Var}\big(\hat{\tau}_{\mathcal{A}}(c_\ell,c_k)\big)}\Big],
    $
have coverage of at least $1-\alpha$ as $n \rightarrow \infty$. Detailed proofs are provided in Web Appendix~D.

\section{Simulations}
\label{sec:simulations}
We performed a simulation study consisting of two parts. Section~\ref{subsec:network_bias_simulations} illustrates the bias resulting from using a misspecified network.  Section~\ref{subsec:NMR_bias_var_sim} shows the bias-variance tradeoff of the NMR estimators in practice.

For all simulations, the exposure mapping was defined as follows. For network $\bA$ and binary treatment vector $\bz$, denote the proportion of treated neighbors of unit $i$ by $g(\bz,\bA_i) = \lvert \mathcal{N}_i(\bA) \rvert^{-1}\sum_{j=1}^{n}A_{ij}z_j$.
The heterogeneous thresholds exposure mapping is defined by
\begin{equation}
    f(\bz,\bA_i) = 
    \begin{cases}
        c_{11},& z_i\cdot \mathbb{I}\{g(\bz,\bA_i)>\nu_i\}=1 \\
        c_{01},& (1-z_i)\cdot \mathbb{I}\{g(\bz,\bA_i)>\nu_i\}=1 \\
        c_{10},& z_i\cdot (1-\mathbb{I}\{g(\bz,\bA_i)>\nu_i\})=1 \\
        c_{00},& (1-z_i)\cdot (1-\mathbb{I}\{g(\bz,\bA_i)>\nu_i\})=1,
    \end{cases}
    \label{eq:4_level_exposure_func}
\end{equation}
where  $\nu_i \in [0,1)$ is a known, possibly unit-specific, threshold. The exposure mapping \eqref{eq:4_level_exposure_func} implies the exposure is a result of two components: whether unit $i$ is treated, and whether the proportion of its treated neighbors surpassed the threshold $\nu_i$.  If it is further assumed that $\nu_i =0 \;\forall i,$ \eqref{eq:4_level_exposure_func} reduces to a commonly used exposure mapping \citep{aronow_estimating_2017}. 
We generated the potential outcomes by taking $\widetilde{Y}_i(c_{00}) \sim U[0.5,1.5]$ and $\widetilde{Y}_i(c_{11}) = \widetilde{Y}_i(c_{00})+1,\; \widetilde{Y}_i(c_{10}) = \widetilde{Y}_i(c_{00})+0.5,\; \widetilde{Y}_i(c_{01}) = \widetilde{Y}_i(c_{00})+0.25$. Thresholds were sampled from $\nu_i \sim U[0,1]$ and are assumed to be known. Treatments were assigned with Bernoulli allocation $\Pr(\bZ = \bz)=0.5^n$. A single network $\bA^\ast$ was sampled from a preferential attachment random network \citep{Barabasi1999} with $n=3000$ nodes. All simulations were repeated for $1000$ iterations in each setup. 
We present and discuss our main findings here. Additional details, specifications, and results are provided in Web Appendix~F.

\subsection{Illustrations of the estimation bias}
\label{subsec:network_bias_simulations}
We considered two scenarios of network misspecification
\noindent 
\paragraph{Scenario (I) (Incorrect reporting of social connections)} 
We created several misspecified networks $\widetilde{\bA}$  by independently adding and removing edges from $\bA^\ast$ with probability $\eta_{1-t,t} = \Pr(\widetilde{A}_{ij}=1-t\vert A^\ast_{ij}=t),\; t=0,1$, for $i\neq j$. We took $\eta := \eta_{0,1}$, fixed $\eta_{1,0} = \eta / 100$.
\noindent
\paragraph{Scenario (II) (Censoring)} 
Censoring of edges in $\bA^\ast$ was created by randomly removing edges of units with more than $K$ edges to obtain a maximum degree of $K \in \{1,\ldots,7\}$. 
\begin{figure}[!ht]
\centering
\includegraphics[scale=0.04]{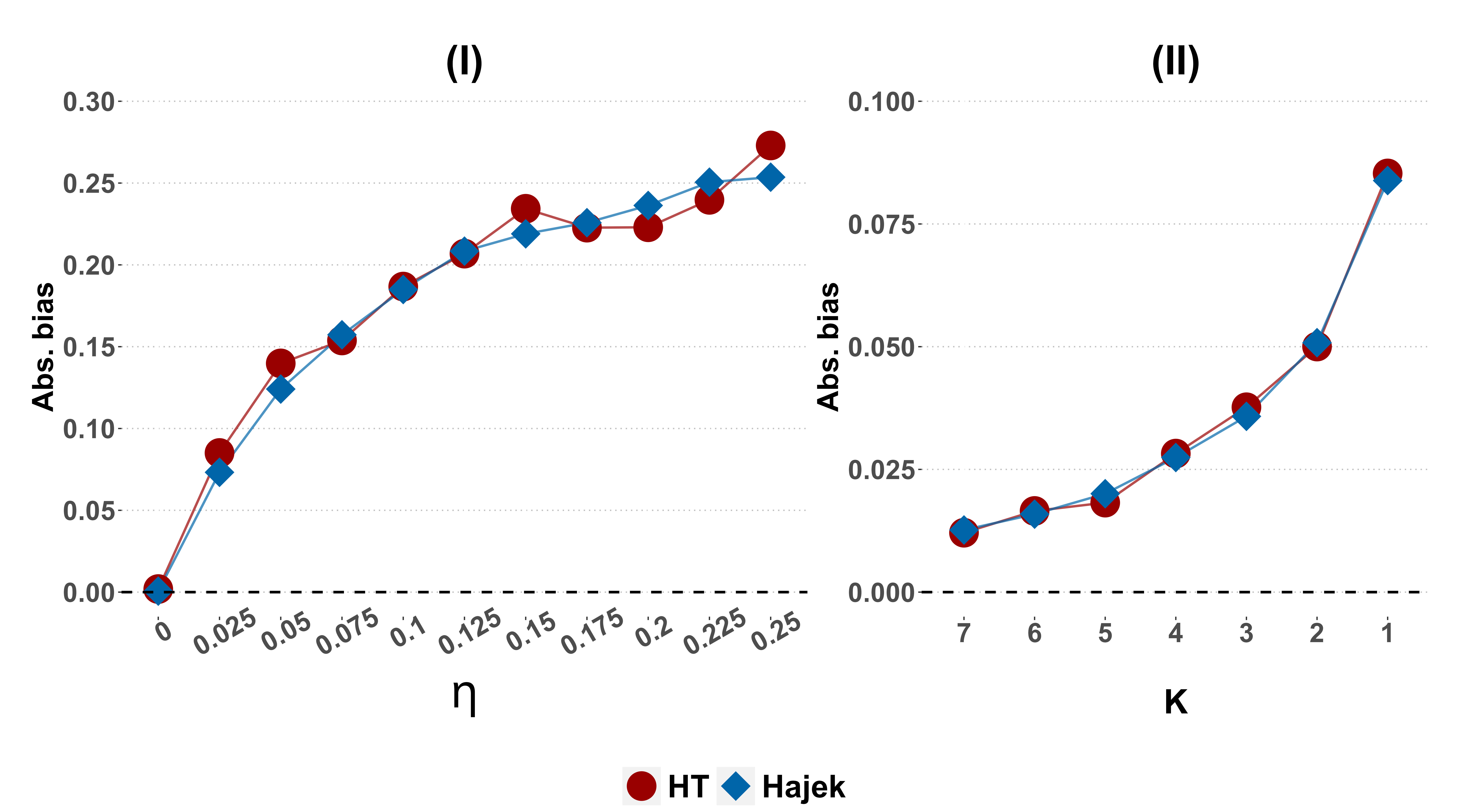}
\caption{Absolute bias $(\lvert Ave(\hat{\tau})-\tau \rvert)$ due to misspecified network. In Scenarios (I) and (II), $\tau(c_{11},c_{00})$ and $\tau(c_{10},c_{00})$, respectively, were estimated with both HT (red circles) and Hajek (blue diamonds) estimators. In Scenario (I), $\eta$  controls the misspecification level. In Scenario (II), $K$ is the censoring threshold. True causal effects are $\tau(c_{10},c_{00})=0.5,\tau(c_{11},c_{00})=1$.}
    \label{fig:simulation_bias}
\end{figure}
Figure~\ref{fig:simulation_bias} displays the absolute bias. We report the results for the  HT \eqref{eq:HT_estimator} and Hajek \eqref{eq:hajek_estimator} estimators of the overall $\tau(c_{11},c_{00})$ and direct $\tau(c_{10},c_{00})$ effects, respectively. In Scenario (I), the magnitude of misspecification was controlled by $\eta$. When $\eta=0$, the true network was used, and, as expected from Corollary~\ref{corollary:unbiased_HT}, the bias was practically zero. The absolute bias increased with $\eta$. In Scenario (II), as the censoring threshold $K$ decreased, the censoring increased, and accordingly so was the bias. 
In both Scenarios (I) and (II), the absolute bias of the indirect effects (e.g., $\tau(c_{01},c_{00})$) was larger than that of the direct effects (e.g., $\tau(c_{10},c_{00})$) (Web Appendix~F). These results can be intuitively explained by recognizing that, under the exposure mapping \eqref{eq:4_level_exposure_func}, network misspecification may lead us to classify a person with true exposure level $c_{j0}$ to exposure level $c_{j1}$ (and vice versa), but will not affect $j$ (for either $j=0$ or $j=1$). 
The estimated Monte-Carlo bias shown here was found to be almost identical to the analytic bias (Web Appendix~F).

\subsection{Bias-variance tradeoff of the NMR estimators}
\label{subsec:NMR_bias_var_sim}
The second simulation study illustrates the bias-variance tradeoff of the NMR estimators.
We generated five misspecified networks $\bA^a,\ldots,\bA^e$ from $\bA^{\ast}$ by independently adding and removing edges using $\eta_{0,1}=0.25$ and $\eta_{1,0}=\eta_{0,1}/100$ with $\eta_{1-t,t}$ as defined in Section~\ref{subsec:network_bias_simulations}.  In total, there were six available networks. The NMR estimators were computed under each of the $\binom{6}{M}$ possible combinations of $\mathcal{A}$ specifications for each $M=1,\ldots,6$. For example, if $M=2$, these possible $\mathcal{A}$ combinations are $\Big\{\{\bA^\ast, \bA^a\},\{\bA^\ast, \bA^b\},\ldots, \{\bA^d, \bA^e\}\Big\}$.

Figure~\ref{fig:NMR_bias_var_plot} shows the absolute bias, standard deviation (SD), and root mean squared error (RMSE) of the Hajek NMR estimator for the indirect effect $\tau(c_{11},c_{10})$. The bias was practically zero whenever $\bA^\ast \in \mathcal{A}$, and larger than zero otherwise. The SD increased with $M$, regardless if $\bA^\ast$ was included, due to the smaller effective sample size.
Interestingly, when $\bA^\ast$ was not included, the bias and RMSE decreased with the number of networks used in the NMR estimator. This phenomenon was stable in all setups and estimands. Additional results, networks' similarity, and empirical coverage are in Web Appendix~F.
We conducted additional simulations in semi-experimental settings by taking the four networks from \citet{Paluck2016} study (see Section \ref{sec:data} for more details on the networks) as $\mathcal{A}$, and simulating treatments and outcomes with the same DGP. 
The results are qualitatively the same (Web Appendix~F).


\begin{figure}[!ht]
    \centering
    \includegraphics[scale=0.03]{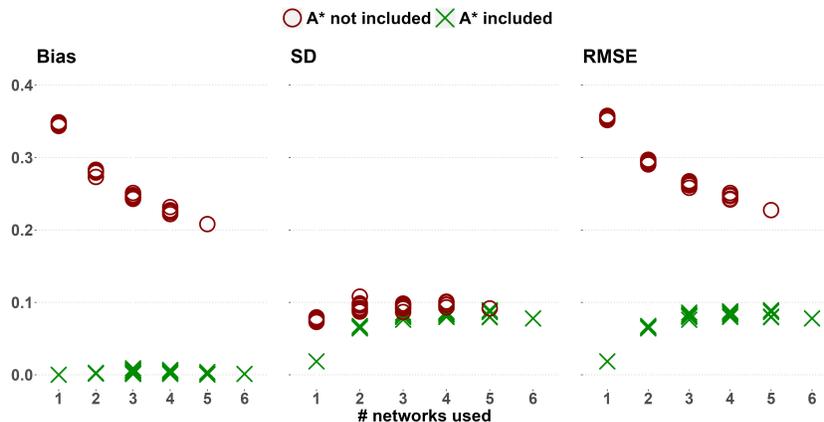}
    \caption{Bias-variance tradeoff of the NMR Hajek estimator for $\tau(c_{11},c_{10})$ as captured by absolute bias, SD, and RMSE. $X$'s indicate that the true network $\bA^{\ast}$ is included in $\mathcal{A}$, and $O$'s otherwise. True causal effect is $\tau(c_{11},c_{10})=0.5$.}
    \label{fig:NMR_bias_var_plot}
\end{figure}

\section{Data analysis} 
\label{sec:data}
We analyzed a field experiment that tested how anti-conflict norms spread in middle school social networks. Key information is provided below; full details are given in \citet{Paluck2016}. Following previous analyses \citep{aronow_estimating_2017}, we analyzed a subset of $n=2983$ eligible students from $56$ schools. Half of the schools were randomly assigned to the intervention arm, and within each selected school, half of the eligible students were given a year-long anti-conflict educational intervention. The social networks were derived from questionnaires. Students were asked to list ten students they spent time (ST) with and two best friends (BF). The questionnaires were given twice: pre- and post-intervention. This resulted in four potential network specifications: ST and BF networks measured before and after the intervention.
A network measured in the post-intervention period is a post-treatment variable, thus using it in the estimation of causal effects implies the assumption that the intervention did not affect the network structure (see Example \ref{exmp:repeated_measures}).

We estimated the effect of the intervention on a behavior outcome (an indicator of wearing a wristband endorsing the program). Following \citet{aronow_estimating_2017}, we use the exposure mapping defined below, which is similar to 
\eqref{eq:4_level_exposure_func}, but also indicates whether the school was assigned to the intervention arm. Let $s_i$ be an indicator of whether the school of unit $i$ was included in the intervention arm. Let $g(\bz,\bA_i)$ denote the proportion of treated neighbors of unit $i$ (as defined before \eqref{eq:4_level_exposure_func}). The exposure mapping is
\begin{equation*}
    f(\bz,\bA_i) = 
    \begin{cases}
        c_{111},& z_i\mathbb{I}\{g(\bz,\bA_i)>0\} s_i=1 \\
        c_{011},& (1-z_i)\mathbb{I}\{g(\bz,\bA_i)>0\}s_i=1 \\
        c_{101},& z_i(1-\mathbb{I}\{g(\bz,\bA_i)>0\})s_i=1 \\
        c_{001},& (1-z_i) (1-\mathbb{I}\{g(\bz,\bA_i)>0\})s_i=1 \\
        c_{000},& (1-s_i) = 1
    \end{cases}
\end{equation*}
We estimated causal effects using two pre-intervention networks individually, both pre-intervention networks and all four networks simultaneously using NMR estimators.
Figure~\ref{fig:palluck} displays the Hajek estimates and $95\%$ confidence intervals of the indirect effect $\tau(c_{011},c_{000})$ and the overall effect $\tau(c_{111},c_{000})$.
Point estimates were consistent across network specifications and combinations in the overall effect estimation. Analysis with all four networks (``ALL") resulted in lower point estimates for the indirect effect. Notably, both indirect and overall effects across all network combinations were statistically nonsignificant, suggesting the intervention may not have substantially altered student behaviors. 
These results reveal the robustness of estimated effects to network specifications and highlight the applicability of the NMR estimators. Additional findings are given in Web Appendix~F.

\begin{figure}[!ht]
    \centering
    \includegraphics[width=0.9\linewidth]{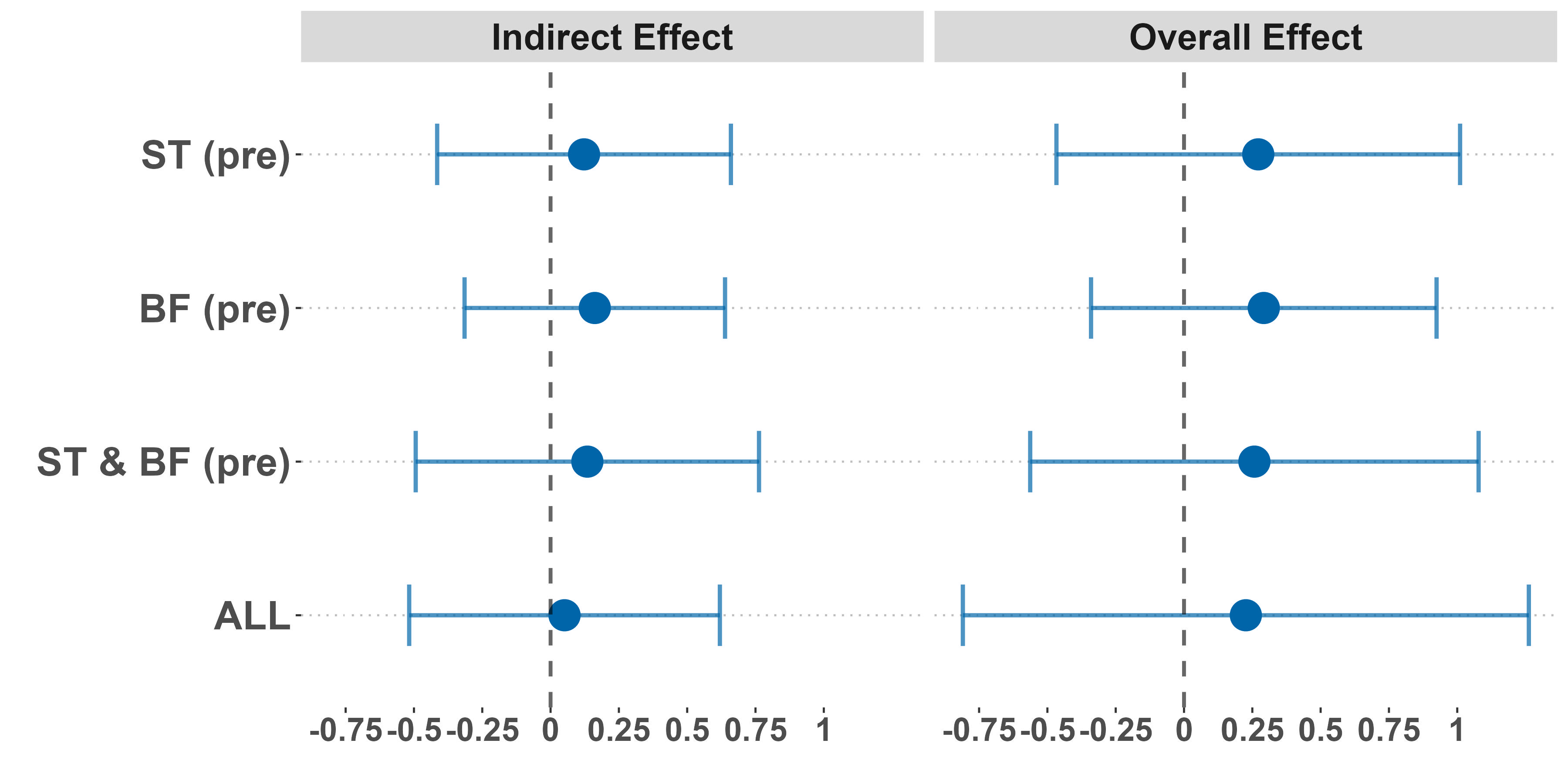}
    \caption{Estimated causal effects in the social network field experiment.
    Indirect and overall effects refer to $\tau(c_{011},c_{000})$ and $\tau(c_{111},c_{000})$, respectively.
    Point estimates and $95\%$ confidence intervals are based on Hajek estimates. ``ST \& BF (pre)" represents the combined pre-intervention networks, while ``ALL" is the four networks combined, both estimated using NMR estimators.}
    \label{fig:palluck}
\end{figure}

\section{Discussion}
\label{sec:disc}

Constructing an interference network from social information requires making additional assumptions and choices. When collecting data through surveys or questionnaires, researchers must consider multiple options, including reciprocity (Example \ref{exmp:reciprocity}), question selection (Example \ref{exmp:network_noise}), and timing (Example \ref{exmp:repeated_measures}). These choices can yield multiple networks, each capturing different aspects of social interactions. Beyond surveys, social networks can be obtained from geospatial data or online interactions. These options result in multi-layer networks measured on the same units but with different edge sets. 
Traditionally, methods have relied on specifying a single network. Our NMR estimators enable researchers to leverage multiple network data sources simultaneously. 

However, this flexibility comes with a bias-variance tradeoff. Each added network may lower the number of units with shared exposures across all networks in $\mathcal{A}$, which can be quantified by the \emph{number of effective units}, defined as $\text{NEU}(\mathcal{A},c_k) = \sum_{i=1}^{n} I^{(\mathcal{A})}_i(\bZ,c_k)$, representing the number of units used in the NMR estimators. $\text{NEU}(\mathcal{A},c_k)$ is decreasing in the number of networks used, regardless of whether $\bA^\ast \in \mathcal{A}$.
In our simulations (Section~\ref{subsec:NMR_bias_var_sim}), the bias and RMSE decreased when using more incorrect networks ($\bA^\ast \not\in \mathcal{A}$), while RMSE increased slightly when combining $\bA^\ast$ with incorrect networks. 


Another limitation of the NMR estimators arises when researchers observe a single network $\bA^{sp}$ but are unsure whether it is correctly specified. Ideally, researchers could augment $\bA^{sp}$ to create multiple candidate networks, for instance by considering sets of networks derived through all possible additions or removals of edges, and subsequently apply the NMR estimator to this augmented set. However, because the number of possible augmentations grows on the order of $O\left(2^{n^2}\right)$, explicitly enumerating all compatible networks quickly becomes computationally infeasible. While heuristic or sampling-based approaches might mitigate this computational barrier, the bias-variance tradeoff still restricts their practical applicability. 

When researchers suspect that both the network and the mapping are misspecified, the NMR estimator can still be used to estimate an \textit{expected exposure effect} \citep{Saevje2023}, which does not assume the exposure mapping is correct (Web Appendix~E).
Furthermore, if, for a given network, researchers can postulate different exposure mappings with the same image space $\mathcal{C}$, but are unsure which map is correct, a modified NMR estimator that estimates causal effects only on units with the same exposure value under all mappings can be constructed. This estimator will be unbiased if one of the mappings is correct, thus providing robustness to exposure mapping specification (see Web Appendix~E for a sketch of the proof).

Our design-based approach assumes that randomness arises only from treatment assignments and takes outcomes as fixed. However, network misspecification could similarly undermine model-based approaches. Adapting NMR-style network aggregation in model-based settings constitutes a promising direction for future research.

We discern between two types of exposure mapping misspecification: incorrect mapping and wrong network. Although previous research has focused mainly on incorrect mapping \citep{aronow_estimating_2017, Saevje2023}, it is plausible that both the mapping and the network are incorrect. 
Randomization tests have been developed to test exposure mapping specification without distinguishing whether the mapping or network is misspecified
\citep{Athey2018,Basse2019_rand,Puelz2022, Hoshino2023}. 
An important avenue for future research involves adapting these tests to evaluate a joint null hypothesis of network and mapping correctness. This could be achieved by testing the intersection of multiple null hypotheses of exposure mapping specifications, potentially by modifying the ``exclusion restriction" condition proposed by \citet{Puelz2022}. However, computational and statistical power limitations present significant challenges in implementing such tests. 


\vspace*{-8pt}

\section*{Acknowledgments}

DN gratefully acknowledges support from the Israel Science Foundation (ISF grant No. 827/21). BW is supported by the Israeli Council For Higher Education Data Science Fellowship.

\vspace*{-8pt}

\section*{Supplementary Materials}

Web Appendices, Tables, and Figures referenced in Sections~\ref{sec:notation.ass.estimand}-\ref{sec:disc}, the R code used in the simulation study, and the R package for implementing the proposed method are available below. R code to reproduce the simulation study is also available at \url{https://github.com/barwein/CI-MIS}. In addition, the R package implementing the network-misspecification-robust estimators is also available at \url{https://github.com/barwein/Misspecified_Interference}. 

\vspace*{-8pt}

\section*{Data Availability}
The data underlying this article are available in ICPSR at \url{https://www.icpsr.umich.edu/web/ICPSR/studies/37070}.

\bibliographystyle{biom} 
\bibliography{network}

\newpage

\appendix
\begin{center}
    \LARGE
    \textbf{Web Appendix}
\end{center}

\setcounter{proposition}{0} 
\renewcommand{\theproposition}{A.\arabic{proposition}}

\renewcommand{\theequation}{\thesection.\arabic{equation}}
\numberwithin{equation}{section}

\renewcommand\thefigure{\thesection.\arabic{figure}}    
\counterwithin{figure}{section}

\renewcommand\thetable{\thesection.\arabic{table}}    
\counterwithin{table}{section}

\setcounter{assumption}{2} 

\section{Proofs}
\label{apdx.sec:proof}
\subsection{Uniqueness of the interference network}

We now formalize that the uniqueness property holds under the exposure mapping framework under further assumptions on the exposure mapping and the potential outcomes. The following two assumptions are required only to illustrate the uniqueness of the correct network and are not needed for the theoretical guarantees we provide in subsequent sections.
\begin{assumption}
\label{ass:exposure_function_differences}
    For all $\bA,\bA' \in \mathscr{A}$ that satisfy  Definition~1 (positivity), if $\bA \neq \bA'$, there exists $\bz \in \mathcal{Z}$ such that  for some $i$, 
    $f(\bz,\bA_i)\neq f(\bz,\bA'_i)$.
\end{assumption}
 Assumption~\ref{ass:exposure_function_differences} states that for any two different networks, there is a treatment vector that results in two different exposure values for at least one unit. In the next subsection, we show that an extended version of a commonly assumed exposure mapping \citep{aronow_estimating_2017}, which we also utilize in this paper (Eq. (10)), satisfies Assumption~\ref{ass:exposure_function_differences}. The following assumption states that the sharp null hypothesis does not hold. 
\begin{assumption}
    \label{ass:alt_to_strong_null}
    $\widetilde{Y}_i(c_\ell)\neq \widetilde{Y}_i(c_k)$ , for all $c_\ell \neq c_k \in \mathcal{C},\; i=1,\dots,n$.
\end{assumption}
 Assumption~\ref{ass:alt_to_strong_null} is strong and is only needed for the following lemma.
\begin{proposition}
    \label{prop:network_unique}
    Assume there exists a network $\bA^\ast \in \mathscr{A}$ that satisfies Definition~2 (correctly specified interference structure). Then,
    under Assumptions~\ref{ass:exposure_function_differences}-\ref{ass:alt_to_strong_null}, $\bA^\ast$ is unique. 
\end{proposition}
In the contrapositive, when $\bA^\ast$ is not unique, at least one of Assumptions \ref{ass:exposure_function_differences} and \ref{ass:alt_to_strong_null} does not hold.
If Assumption~\ref{ass:exposure_function_differences} does not hold, there exist at least two different networks under which $f$ maps to identical values for all treatment vectors, making the networks indistinguishable in terms of the exposure values. If Assumption~\ref{ass:alt_to_strong_null} does not hold, then two different networks that yield two different exposure values $c_\ell, c_k$, for some $\bz$, will result in the same potential outcomes $\widetilde{Y}_i(c_\ell)=\widetilde{Y}_i(c_k)$ for at least one unit.

\begin{proof}
Assume in contradiction there exists another network $\bA \in \mathscr{A}$ that satisfies Definition 2, which is not $\bA^\ast$ (i.e., $\bA^\ast \neq \bA$).  Assumption \ref{ass:exposure_function_differences} implies there exists $\bz \in \mathcal{Z}$ such that $f(\bz,\bA^{\ast}_i)=c_\ell$ and  $f(\bz,\bA_i) = c_k$, for some $i$ and some $\ell \neq k$.
By Definition 2, we have that $Y_i(\bz)=\widetilde{Y}_i(c_\ell)$ and $Y_i(\bz) = \widetilde{Y}_i(c_k)$, i.e., $\widetilde{Y}_i(c_\ell) = \widetilde{Y}_i(c_k)$. However, this is in contradiction to Assumption \ref{ass:alt_to_strong_null}, thus it must be that $\bA^\ast = \bA$.
\end{proof}

 Given the (non-empty) class $\mathscr{A}^\ast$ of correctly specified networks (all networks that satisfies Definition~2), we can define the \emph{minimal} class of correctly specified networks by
    \begin{equation*}
        \mathscr{A}^\ast_{min} = \big\{\bA \in \mathscr{A}^\ast : \lvert E(\bA) \rvert = \min_{\bA' \in \mathscr{A}^\ast} \lvert E(\bA') \rvert \big\},
    \end{equation*}
    where $E(\bA)$ is the edge set of network $\bA \in \mathscr{A}$, and $\lvert E(\bA) \rvert$ is its size. That is, $\mathscr{A}^\ast_{min}$ is the class of correctly specified networks with the least number of edges.
    However, $\mathscr{A}^\ast_{min}$ is not necessarily a singleton, and there may be more than one minimal correctly specified network. To see that, we can follow a similar derivation for the proof of network uniqueness (Proposition~\ref{prop:network_unique}). 
    
    Assume there exist two networks $\bA^1,\bA^2 \in \mathscr{A}^\ast_{min}$ with $\bA^1 \neq \bA^2$. Assume that the exposure mapping satisfies Assumption~\ref{ass:exposure_function_differences} of exposure mapping distinguishability (at least for the networks in $\mathscr{A}^\ast$). That is, assume there exists a treatment assignment $\bz \in \mathcal{Z}$ such that for some unit $i$ 
    \begin{align*}
        f(\bz,\bA^1_i) &= c_\ell \\
        f(\bz,\bA^2_i) &= c_k,   
    \end{align*}
    but $c_\ell \neq c_k$. Since both $\bA^1$ and $\bA^2$ correctly specify the interference structure (satisfy Definition~2), we have
    \begin{align*}
        Y_i(\bz) &= \widetilde{Y}_i(c_\ell) \\
        Y_i(\bz) &= \widetilde{Y}_i(c_k),
    \end{align*}
    therefore, $\widetilde{Y}_i(c_\ell) = \widetilde{Y}_i(c_k)$ for some unit $i$. Thus, if we want $\mathscr{A}^\ast_{min}$ to be a singleton we have to:
    \begin{enumerate}[(i)]
        \item  Constrain the exposure mapping to have distinguishability in exposure values between networks in $\mathscr{A}^\ast_{min}$ such that two distinct networks will not yield the same exposures for all treatments and units. Otherwise, two networks with the same number of edges could still have the same effective exposures, and $\mathscr{A}^\ast_{min}$ will not be unique.
        \item Assume that the null hypothesis does not hold for some units.
    \end{enumerate}
    We show in Web Appendix~A.2 that the common four-level exposure mapping (Equation (10) in the main text), has distinguishability (i.e., satisfies Assumption~\ref{ass:exposure_function_differences}). In that case we have to assume that the sharp null does not hold to achieve uniqueness of the minimal class $\mathscr{A}^\ast_{min}$, which can be problematic.

\subsection{The heterogeneous thresholds exposure mapping satisfies 
Assumption \ref{ass:exposure_function_differences}}
\begin{proposition}
Assumption \ref{ass:exposure_function_differences} holds for the exposure mapping $(10)$.
\end{proposition}
\begin{proof}
Let $\bA\neq \bA' \in \mathscr{A}$. Since $\bA\neq \bA'$, there exists some unit $i$ with $\bA_i \neq \bA'_i$. The difference between $\bA_i$ and $\bA'_i$ can be due to the addition or removal of at least one edge. Let $d_i(\bA) = \vert \mathcal{N}_i(\bA) \vert$ be the degree of unit $i$ in network $\bA$. Assume that $d_i(\bA)=a$ and $d_i(\bA')=a'$, for some scalars $a,a' \in \mathbb{N}$. Assume WLOG that $a \geq a'$.

Denote the set of joint edges of $i$ in the two networks by $\mathcal{M}_i(\bA,\bA') = \mathcal{N}_i(\bA) \cap \mathcal{N}_i(\bA')$. Denote the complementary set of $\mathcal{N}_i(\bA)$, excluding $i$ itself, by $\mathcal{N}_i(\bA)^c = \{j\neq i : A_{ij}=0\}$, and similarly for $\mathcal{N}_i(\bA')^c$. Denote the edges difference set by $\mathcal{N}_i(\bA)\setminus \mathcal{N}_i(\bA') = \mathcal{N}_i(\bA) \cap \mathcal{N}_i(\bA')^c$. We may write $\mathcal{N}_i(\bA)$ as
    \begin{equation*}
        \begin{split}
            \mathcal{N}_i(\bA) &= 
            \left[\mathcal{N}_i(\bA) \cap \mathcal{N}_i(\bA')^c\right] \cup \left[\mathcal{N}_i(\bA) \cap \mathcal{N}_i(\bA')\right]
            \\ &=
            \left[\mathcal{N}_i(\bA) \cap \mathcal{N}_i(\bA')^c\right] \cup \mathcal{M}_i(\bA,\bA')
        \end{split}
    \end{equation*}
    Since $\mathcal{N}_i(\bA) \cap \mathcal{N}_i(\bA')^c$ and $\mathcal{M}_i(\bA,\bA')$ are disjoint, we can write $g(\bz,\bA_i)$ as
    \begin{equation*}
        g(\bz,\bA_i) = \frac{1}{a}\bigg( \sum_{j\in \mathcal{N}_i(\bA) \cap \mathcal{N}_i(\bA')^c}z_j  + \sum_{j\in\mathcal{M}_i(\bA,\bA')}z_j\bigg),
    \end{equation*}
    and similarly for $g(\bz,\bA'_i)$,
     \begin{equation*}
        g(\bz,\bA'_i) = \frac{1}{a'} \bigg( \sum_{j\in \mathcal{N}_i(\bA') \cap \mathcal{N}_i(\bA)^c}z_j  + \sum_{j\in\mathcal{M}_i(\bA,\bA')}z_j\bigg).
    \end{equation*}
    Since $a\geq a'$, the set $\mathcal{N}_i(\bA) \cap \mathcal{N}_i(\bA')^c$ is not empty.
    Now, taking $\bz$ with $z_i=1$, the possible exposure values are only $c_{10}$ and $c_{11}$. We separate the proof for the two possible cases and further separate as needed. We show that in each of these (sub) cases, one can choose a treatments vector $\bz$  such that $f(\bz, \bA_i)\ne f(\bz, \bA'_i)$ (e.g., under one network we obtain exposure level $c_{11}$ and under the other one $c_{10}$). Turning to the different cases, we start with separating the cases $\nu_i=0$ and $\nu_i>0$
    \begin{enumerate}
        \item Case 1: $\nu_i=0$. Here we can take $z_j = 0$ for all $j\in \left[\mathcal{N}_i(\bA') \cap \mathcal{N}_i(\bA)^c\right] \cup \mathcal{M}_i(\bA,\bA')$, and $z_j =1$ for at least one $j \in \mathcal{N}_i(\bA) \cap \mathcal{N}_i(\bA')^c$, to obtain a specific treatment vector $\bz$ that results with $g(\bz,\bA_i) > 0$ and $g(\bz,\bA'_i) =0$, and therefore $f(\bz,\bA_i)=c_{11}$, while $f(\bz,\bA'_i)=c_{10}$, as required.
        \item Case 2: $0<\nu_i<1$. Denote the number of edges in each of the aforementioned sets by $n_{i,a} = \vert \mathcal{N}_i(\bA) \cap \mathcal{N}_i(\bA')^c\vert, \; n_{i,a'} = \vert \mathcal{N}_i(\bA') \cap \mathcal{N}_i(\bA)^c \vert, \; n_{i,a\cap a'} = \vert \mathcal{M}_i(\bA,\bA') \vert$. We obtain that $n_{i,a} + n_{i,a\cap a'} = a$, $n_{i,a'} + n_{i,a\cap a'} = a'$, and $n_{i,a} \geq 1$. We differentiate between two cases.
        \begin{enumerate}[i.]
            \item If $\frac{n_{i,a}}{a} > \nu_i$ then for $z_j=1$ for all $j \in \mathcal{N}_i(\bA) \cap \mathcal{N}_i(\bA')^c$, and $z_j = 0$ for the rest, we obtain $g(\bz,\bA_i)> \nu_i$ while  $g(\bz,\bA'_i)=0 < \nu_i$, as required.
            \item If $\frac{n_{i,a}}{a} \leq \nu_i$, from positivity of all exposure values under both $\bA$ and $\bA'$, there must exist a set of units in $\mathcal{N}_i(\bA)$ such that $g(\bz,\bA_i) > 0$. Since $\frac{n_{i,a}}{a} \leq \nu_i$, we have to add treated units from $\mathcal{M}_i(\bA,\bA')$ for $g(\bz,\bA_i)$ to be larger than $\nu_i$, thus $\mathcal{M}_i(\bA,\bA')$ is not an empty set. Define the minimal number of such units by
            \begin{equation}
                \label{apdx.eq:minimal_n}
                \widetilde{n}_{i,a \cap a'} = \min_{\widetilde{n} \in \{1,\dots,n_{i,a\cap a'}\}} \widetilde{n},\; \text{s.t.} \; \frac{n_{i,a} + \widetilde{n}}{a} > \nu_i
            \end{equation}
            Here we also have two options.
            \begin{itemize}
                \item If $\frac{\widetilde{n}_{i,a \cap a'}}{a'} \leq \nu_i$, we can take,  $z_j=1$ for all $j \in \mathcal{N}_i(\bA) \cap \mathcal{N}_i(\bA')^c$ and for $\widetilde{n}_{i,a \cap a'}$ units from $\mathcal{M}_i(\bA,\bA')$ to obtain $g(\bz,\bA_i) > \nu_i$ and $g(\bz,\bA'_i) \leq \nu_i$, as required.
                \item If $\frac{\widetilde{n}_{i,a \cap a'}}{a'} > \nu_i$, now the previous treatments selection yields $g(\bz,\bA'_i) > \nu_i$. However, notice that $\widetilde{n}_{i,a \cap a'}$ as defined in \eqref{apdx.eq:minimal_n}, is \emph{minimal}, i.e., $\frac{n_{i,a} + \widetilde{n}_{i,a \cap a'}}{a} > \nu_i$ and $\frac{n_{i,a} + \widetilde{n}_{i,a \cap a'}-1}{a} \leq \nu_i$. Therefore,
                \begin{equation}
                \label{AppEq:finalmove}
                        \frac{\widetilde{n}_{i,a \cap a'}}{a} \ \leq \ \nu_i - \frac{n_{i,a}-1}{a} \ {\leq} \ \nu_i,
                \end{equation}
                where the last inequality in \eqref{AppEq:finalmove} holds since $n_{i,a} \geq 1$. Thus, if we take $z_j = 1$ for $\widetilde{n}_{i,a \cap a'}$ units from $\mathcal{M}_i(\bA,\bA')$, and $z_j =0$ for the rest, we obtain $g(\bz,\bA_i) \leq \nu_i$ and $g(\bz,\bA'_i) > \nu_i$, as required. 
            \end{itemize}
        \end{enumerate}
    \end{enumerate}
\end{proof}

\subsection{Proof of Theorem 1}

\begin{proof}
    Let $\bA^{sp}$ be the specified network. Let $\bA^\ast \in \mathscr{A}^\ast$ be some correctly specified network. By consistency,
    \begin{equation*}
        \begin{split}\mathbb{E}_{\bZ}\left[\hat{\mu}_{\bA^{sp}}(c_k)\right] 
        &=      \mathbb{E}_{\bZ}\left[
        \frac{1}{n}
        \sum_{i=1}^{n}
    \mathbb{I}\{f(\bZ,\bA^{sp}_i)=c_k\}\frac{1}{p_i^{(\bA^{sp})}(c_k)}
        \sum_{j=1}^{L}\mathbb{I}\{f(\bZ,\bA^\ast_i)=c_j\}
        \widetilde{Y}_i(c_j)
         \right]
     \\ &=
        \frac{1}{n}
        \mathbb{E}_{\bZ}\left[
        \sum_{i=1}^{n}
        \sum_{j=1}^{L}
    \frac{1}{p_i^{(\bA^{sp})}(c_k)}
    \mathbb{I}\{f(\bZ,\bA^{sp}_i)=c_k\}
       \mathbb{I}\{f(\bZ,\bA^\ast_i)=c_j\}
         \widetilde{Y}_i(c_j)
         \right]
     \\ &=
     \frac{1}{n}
        \sum_{i=1}^{n}
        \sum_{j=1}^{L}
    \frac{1}{p_i^{(\bA^{sp})}(c_k)}
    \mathbb{E}_{\bZ}\bigg[
       \mathbb{I}\{f(\bZ,\bA^{sp}_i)=c_k\}
       \mathbb{I}\{f(\bZ,\bA^\ast_i)=c_j\}
         \bigg]
     \widetilde{Y}_i(c_j)
    \\ &=
        \frac{1}{n}
        \sum_{i=1}^{n}
        \sum_{j=1}^{L}
    \frac{p_i^{(\bA^\ast, \bA^{sp})}(c_j, c_k)}{p_i^{(\bA^{sp})}(c_k)}
        \widetilde{Y}_i(c_j).
    \\ &=
      \frac{1}{n}
        \sum_{i=1}^{n}
        \sum_{j=1}^{L}
    p_i(c_j;\bA^\ast \mid c_k; \bA^{sp})
      \widetilde{Y}_i(c_j)
        \end{split}
    \end{equation*}
   By adding and subtracting $\mu(c_k)$ we obtain,
   \begin{equation*}
\mathbb{E}_{\bZ}\left[\hat{\mu}_{\bA^{sp}}(c_k)\right] = \mu(c_k) + 
\frac{1}{n}
        \sum_{i=1}^{n}
        \Big[\big\{p_i(c_k;\bA^\ast \mid c_k; \bA^{sp})-1\big\}\widetilde{Y}_i(c_k) +
        \sum_{j=1, j\neq k}^{L}
    p_i(c_j;\bA^\ast \mid c_k; \bA^{sp})
        \widetilde{Y}_i(c_j)\Big]
   \end{equation*} 
   Rearranging and taking absolute values on both sides yields,
   \begin{equation*}
       \begin{split}
     \Big \lvert \mathbb{E}_{\bZ}\left[\hat{\mu}_{\bA^{sp}}(c_k)\right] - \mu(c_k) \Big \rvert 
     &= \Big \lvert \frac{1}{n}
        \sum_{i=1}^{n}
        \Big[\big\{p_i(c_k;\bA^\ast \mid c_k; \bA^{sp})-1\big\}\widetilde{Y}_i(c_k) +
        \sum_{j=1, j\neq k}^{L}
    p_i(c_j;\bA^\ast \mid c_k; \bA^{sp})
        \widetilde{Y}_i(c_j)\Big] \Big \rvert
    \\ &\leq  \frac{1}{n}
        \sum_{i=1}^{n}
        \Big[\big \vert p_i(c_k;\bA^\ast \mid c_k; \bA^{sp})-1 \big \vert \cdot \vert \widetilde{Y}_i(c_k) \vert +
        \sum_{j=1, j\neq k}^{L}
    p_i(c_j;\bA^\ast \mid c_k; \bA^{sp})
        \cdot \vert \widetilde{Y}_i(c_j)\vert \Big]
        \\ &\leq
        \frac{\kappa}{n}
        \sum_{i=1}^{n}
        \Big[
        \big\vert p_i(c_k;\bA^\ast \mid c_k; \bA^{sp})-1 \big \vert  + 
        \sum_{j=1, j\neq k}^{L}
    p_i(c_j;\bA^\ast \mid c_k; \bA^{sp}) \Big]
    \\ &=
        \frac{\kappa}{n}
        \sum_{i=1}^{n}
        \Big[
         \big\vert p_i(c_k;\bA^\ast \mid c_k; \bA^{sp})-1 \big \vert + 
        1 - p_i(c_k;\bA^\ast \mid c_k; \bA^{sp}) \Big]
    \\ &=
        \frac{2\kappa}{n}
        \sum_{i=1}^{n}
        [1 - p_i(c_k;\bA^\ast \mid c_k; \bA^{sp})],  
       \end{split}
   \end{equation*}
   where the second line follows from Minkowski's inequality, the third line from Assumption 2 of bounded potential outcomes $\vert Y_i(c_j)\vert \leq \kappa,\; \forall i,j$, the fourth line since there $L$ possible exposures and their probabilities sum to one, and the fifth line since $\big\vert p_i(c_k;\bA^\ast \mid c_k; \bA^{sp})-1 \big \vert =  1 - p_i(c_k;\bA^\ast \mid c_k; \bA^{sp})$.

   Additionally, the bound is sharp. To demonstrate this, we construct a specific data-generating process that attains the bound. Assume that for a chosen exposure $c_k$, the potential outcomes are  $\widetilde{Y}_i(c_k) = -\kappa$ for all units $i$, and for all other exposure  values $\widetilde{Y}_i(c_j)=\kappa$ for all units $i$ and for all $j \ne k$. Under this construction, Assumption~2 (bounded potential outcomes) holds. We obtain,
   \begin{equation*}
       \begin{split}
           \Big \lvert \mathbb{E}_{\bZ}\left[\hat{\mu}_{\bA^{sp}}(c_k)\right] - \mu(c_k) \Big \rvert 
     &= \Big \lvert \frac{1}{n}
        \sum_{i=1}^{n}
        \Big[\big\{p_i(c_k;\bA^\ast \mid c_k; \bA^{sp})-1\big\}\widetilde{Y}_i(c_k) +
        \sum_{j=1, j\neq k}^{L}
    p_i(c_j;\bA^\ast \mid c_k; \bA^{sp})
        \widetilde{Y}_i(c_j)\Big] \Big \rvert
        \\ &=
        \Big \lvert \frac{1}{n}
        \sum_{i=1}^{n}
        \Big[\kappa\big\{1-p_i(c_k;\bA^\ast \mid c_k; \bA^{sp})\big\} +
        \kappa\sum_{j=1, j\neq k}^{L}
    p_i(c_j;\bA^\ast \mid c_k; \bA^{sp})
        \Big] \Big \rvert
        \\ &=
           \Big \lvert \frac{1}{n}
        \sum_{i=1}^{n}
        \Big[\kappa\big\{1-p_i(c_k;\bA^\ast \mid c_k; \bA^{sp})\big\} +
        \kappa \big\{1-p_i(c_k;\bA^\ast \mid c_k; \bA^{sp})\big\}
        \Big] \Big \rvert
        \\ &=
        \Big \lvert \frac{2\kappa}{n}
        \sum_{i=1}^{n}
        1 - p_i(c_k;\bA^\ast \mid c_k; \bA^{sp}) \Big \rvert
        \\&=
        \frac{2\kappa}{n}
        \sum_{i=1}^{n}
        1 - p_i(c_k;\bA^\ast \mid c_k; \bA^{sp}),
       \end{split}
   \end{equation*}
   The second equality substitutes the assumed potential outcome values. The third equality uses the fact that $\sum_{j=1, j\neq k}^{L} p_i(c_j;\bA^\ast \mid c_k; \bA^{sp}) = 1 - p_i(c_k;\bA^\ast \mid c_k; \bA^{sp})$. The final equality holds because each term $2\kappa\left(1-p_i(c_k;\bA^\ast \mid c_k; \bA^{sp})\right)$ is non-negative (since $\kappa>0$ and $p_i(c_k;\bA^\ast \mid c_k; \bA^{sp}) \le 1$), so the sum is non-negative, and the absolute value can be removed. This matches the bound, thus demonstrating its sharpness under this specific DGP.
\end{proof}
Moreover, recall that the HT estimator of causal effects is
$\hat{\tau}_{\bA^{sp}}(c_\ell,c_k) = \hat{\mu}_{\bA^{sp}}(c_\ell) - \hat{\mu}_{\bA^{sp}}(c_k)$, and that causal effects are defined as $\tau(c_\ell, c_k) = \mu(c_\ell) - \mu(c_k)$. Therefore,

\begin{equation*}
    \begin{split}
        \Big \lvert \mathbb{E}_{\bZ}\left[\hat{\tau}_{\bA^{sp}}(c_\ell,c_k)\right] - \tau(c_\ell, c_k) \Big \rvert
        &= \Big \lvert \mathbb{E}_{\bZ}\left[\hat{\mu}_{\bA^{sp}}(c_\ell) - \hat{\mu}_{\bA^{sp}}(c_k)\right] -  \{\mu(c_\ell) - \mu(c_k)\} \Big \rvert
        \\ &=
        \Big \lvert \big\{\mathbb{E}_{\bZ}\left[\hat{\mu}_{\bA^{sp}}(c_\ell)\right] -\mu(c_\ell)\big\} + \{\mu(c_k) - \mathbb{E}_{\bZ}\left[\hat{\mu}_{\bA^{sp}}(c_k)\right] \big\}  \Big \rvert
        \\ &\leq
        \Big \lvert \mathbb{E}_{\bZ}\left[\hat{\mu}_{\bA^{sp}}(c_\ell)\right] -\mu(c_\ell) \Big \rvert + \Big \lvert \mu(c_k) - \mathbb{E}_{\bZ}\left[\hat{\mu}_{\bA^{sp}}(c_k)\right] \Big \rvert
        \\ &=
        \Big \lvert \mathbb{E}_{\bZ}\left[\hat{\mu}_{\bA^{sp}}(c_\ell)\right] -\mu(c_\ell) \Big \rvert + \Big \lvert \mathbb{E}_{\bZ}\left[\hat{\mu}_{\bA^{sp}}(c_k)\right] - \mu(c_k) \Big \rvert
    \end{split}
\end{equation*}
Consequently, by Theorem 1,
\begin{equation*}
     \Big \lvert \mathbb{E}_{\bZ}\left[\hat{\tau}_{\bA^{sp}}(c_\ell,c_k)\right] - \tau(c_\ell, c_k) \Big \rvert \leq 
      \frac{2\kappa}{n}
        \sum_{i=1}^{n}
        \{1 - p_i(c_\ell;\bA^\ast \mid c_\ell; \bA^{sp})\} + 
        \{1 - p_i(c_k;\bA^\ast \mid c_k; \bA^{sp})\}
\end{equation*}

\subsection{Exact bias of the Horvitz-Thompson estimator}
In this subsection we derive the exact bias of $\hat{\tau}_{\bA^{sp}}(c_\ell,c_k)$. To that end, we can relax Assumption 2 of bounded potential outcomes. From the proof of Theorem 1 shown in the previous subsection
   \begin{equation*}
        \mathbb{E}_{\bZ}\left[\hat{\mu}_{\bA^{sp}}(c_k)\right] 
        = 
      \frac{1}{n}
        \sum_{i=1}^{n}
        \sum_{j=1}^{L}
    p_i(c_j;\bA^\ast \mid c_k; \bA^{sp})
      \widetilde{Y}_i(c_j),
    \end{equation*}
    therefore,
  \begin{equation*}
    \begin{split}
    \mathbb{E}_{\bZ}\left[\hat{\tau}_{\bA^{sp}}(c_\ell,c_k)\right] 
       &=
    \mathbb{E}_{\bZ}\left[\hat{\mu}_{\bA^{sp}}(c_\ell)-\hat{\mu}_{\bA^{sp}}(c_k)\right] 
    \\ &=
     \frac{1}{n}
        \sum_{i=1}^{n}
        \sum_{j=1}^{L}
        \left[
    p_i(c_j;\bA^\ast \mid c_\ell; \bA^{sp}) - 
     p_i(c_j;\bA^\ast \mid c_k; \bA^{sp})
    \right]
        \widetilde{Y}_i(c_j)
    \\ &=
     \frac{1}{n}
        \sum_{i=1}^{n}
        \sum_{j=1}^{L}
        \left[
    p_i(c_j;\bA^\ast \mid c_\ell; \bA^{sp}) - 
     p_i(c_j;\bA^\ast \mid c_k; \bA^{sp})
    \right]
        \widetilde{Y}_i(c_j) 
        \\ &\quad + \tau(c_\ell,c_k) - \tau(c_\ell,c_k)
    \\ &=
    \tau(c_\ell,c_k) + 
  \frac{1}{n}\sum_{i=1}^{n}\sum_{j=1}^{L}\left[q_i(c_j; \bA^{\ast} \mid c_\ell; \bA^{sp}) - q_i(c_j; \bA^{\ast} \mid c_k; \bA^{sp})\right] \widetilde{Y}_i(c_j)
  \\ &=
  \tau(c_\ell,c_k) + B(c_\ell, c_k; \bA^{sp}),
   \end{split}
   \end{equation*}
 with
    $$B(c_\ell, c_k; \bA^{sp}) = \frac{1}{n}\sum_{i=1}^{n}\sum_{j=1}^{L}\Big[q_i(c_j; \bA^{\ast} \mid c_\ell; \bA^{sp}) - q_i(c_j; \bA^{\ast} \mid c_k; \bA^{sp})\Big] \widetilde{Y}_i(c_j),$$
    and where $q_i$ are defined by
 \begin{equation*}
 q_i(c_j; \bA^{\ast} \mid c_k; \bA^{sp}) =
    \begin{cases}          
    p_i(c_j;\bA^\ast \mid c_k; \bA^{sp}) ,& j\neq k \\
    p_i(c_j;\bA^\ast \mid c_k; \bA^{sp})-1,& j= k
   \end{cases}
    \end{equation*}
That is, that bias of $\hat{\tau}_{\bA^{sp}}$ is a weighted sum of all $L$ potential outcomes $\widetilde{Y}$ with weights that relates to the conditional probabilities $p_i(c_j;\bA^\ast \mid c_k; \bA^{sp})$.

Moreover, as shown in Section 3, the sum $\sum_{j=1}^{L}\mathbb{I}\{f(\bZ,\bA^\ast_i)=c_j\}
\widetilde{Y}_i(c_j)$ is equal for all $\bA^\ast \in \mathscr{A}^\ast$. Thus, the term
$\mathbb{I}\{f(\bZ,\bA^{sp}_i)=c_k\}\sum_{j=1}^{L}\mathbb{I}\{f(\bZ,\bA^\ast_i)=c_j\}
        \widetilde{Y}_i(c_j)$ is also equal for all $\bA^\ast \in \mathscr{A}^\ast$, and by taking expectation w.r.t. $\bZ$ we obtain that $B(c_\ell, c_k; \bA^{sp})$ is equal for all $\bA^{\ast}\in\mathscr{A}^\ast$.
\subsection{Proof of Theorem 2}
\label{apdx:prof_of_NMR_unbiasedness}
\begin{proof}
Let $\mathcal{A}=\{\bA^1,\dots,\bA^M\}$ be the collection of $M$ networks. Note that $\mathcal{A} \cap \mathscr{A}^\ast \ne \emptyset$ means that for some $j,\; \bA^j \in \mathscr{A}^\ast$. Assume without loss of generality that $\bA^1 \in \mathscr{A}^\ast$ and write $\bA^1 = \bA^\ast$. We obtain

\begin{equation*}
    \begin{split}
        \mathbb{E}_{\bZ}\left[\hat{\mu}_{\mathcal{A}}(c_\ell)\right] 
        &=
        \mathbb{E}_{\bZ}\left[\frac{1}{n}
        \sum_{i=1}^{n}
    \left(\prod_{j=1}^{M}
    \mathbb{I}
    \{f(\bZ,\bA^j_i)=c_\ell\}\right)
    \frac{1}{p_i^{(\mathcal{A})}(c_\ell)}Y_i\right] 
    \\ (\text{Consistency}) &=
        \mathbb{E}_{\bZ}\left[\frac{1}{n}
        \sum_{i=1}^{n}
    \left(\prod_{j=1}^{M}
    \mathbb{I}
    \{f(\bZ,\bA^j_i)=c_\ell\}\right)
    \frac{1}{p_i^{(\mathcal{A})}(c_\ell)}\sum_{k = 1}^{L}\mathbb{I}\{f(\bZ, \bA^{\ast}_i) = c_k\}\widetilde{Y}_i(c_k)\right]
    \\ 
    &=
        \mathbb{E}_{\bZ}\Biggl[\frac{1}{n}
        \sum_{i=1}^{n}
    \left(\prod_{j=2}^{M}
    \mathbb{I}
    \{f(\bZ,\bA^j_i)=c_\ell\}\right)
    \frac{1}{p_i^{(\mathcal{A})}(c_\ell)}\cdot
    \\ &\qquad \qquad \qquad \qquad
    \mathbb{I}\{f(\bZ, \bA^1_i)=c_\ell\}\sum_{k = 1}^{L} \mathbb{I}\{f(\bZ, \bA^{\ast}_i) = c_k\}\widetilde{Y}_i(c_k)\Biggr] 
    \\
    (\bA^1 = \bA^\ast) &=
        \mathbb{E}_{\bZ}\Biggl[\frac{1}{n}
        \sum_{i=1}^{n}
    \left(\prod_{j=2}^{M}
    \mathbb{I}
    \{f(\bZ,\bA^j_i)=c_\ell\}\right)
    \frac{1}{p_i^{(\mathcal{A})}(c_\ell)}\cdot
    \\ &\qquad \qquad \qquad \qquad
    \sum_{k = 1}^{L}\mathbb{I}\{f(\bZ, \bA^{\ast}_i) = c_\ell\}\mathbb{I}\{f(\bZ, \bA^{\ast}_i) = c_k\}\widetilde{Y}_i(c_k)\Biggr] 
    \\ &\stackrel{\dagger}{=}
    \mathbb{E}_{\bZ}\left[\frac{1}{n}
        \sum_{i=1}^{n}
    \left(\prod_{j=1}^{M}
    \mathbb{I}
    \{f(\bZ,\bA^j_i)=c_\ell\}\right)
    \frac{1}{p_i^{(\mathcal{A})}(c_\ell)}\widetilde{Y}_i(c_\ell)\right] 
    \\ &=
    \frac{1}{n}
        \sum_{i=1}^{n}
    \mathbb{E}_{\bZ}\left[\prod_{j=1}^{M}
    \mathbb{I}
    \{f(\bZ,\bA^j_i)=c_\ell\}\right]
    \frac{1}{p_i^{(\mathcal{A})}(c_\ell)}\widetilde{Y}_i(c_\ell) 
    \\ &=
    \frac{1}{n}
        \sum_{i=1}^{n}
    p_i^{(\mathcal{A})}(c_\ell)
    \frac{1}{p_i^{(\mathcal{A})}(c_\ell)}\widetilde{Y}_i(c_\ell)
    \\ &=
    \frac{1}{n}
        \sum_{i=1}^{n} \widetilde{Y}_i(c_\ell)
    \\ &=
    \mu(c_\ell)
    \end{split}
\end{equation*}
Where $\dagger$ follows from the fact that $\sum_{k = 1}^{L}\mathbb{I}\{f(\bZ, \bA^{\ast}_i) = c_\ell\}\mathbb{I}\{f(\bZ, \bA^{\ast}_i) = c_k\}\widetilde{Y}_i(c_k)=\mathbb{I}\{f(\bZ, \bA^{\ast}_i) = c_\ell\}\widetilde{Y}_i(c_\ell)$. Moreover, if $\bA^\ast$ is not unique (i.e., $\mathscr{A}^\ast$ is not a singleton), the sum $\sum_{k = 1}^{L}\mathbb{I}\{f(\bZ, \bA^{\ast}_i) = c_k\}\widetilde{Y}_i(c_k)$ will be equal for any $\bA^\ast \in \mathscr{A}^\ast$, as already been established in the main text (Section 3), and thus the proof will follow using similar derivations. 
The additivity of expectation yields
\begin{equation*}
    \mathbb{E}_{\bZ}\left[\hat{\tau}_{\mathcal{A}}(c_\ell,c_k)\right] 
    =
   \mathbb{E}_{\bZ}\left[\hat{\mu}_{\mathcal{A}}(c_\ell)\right]  -
   \mathbb{E}_{\bZ}\left[\hat{\mu}_{\mathcal{A}}(c_k)\right] 
   = 
   \mu(c_\ell) - \mu(c_k) 
   =
   \tau(c_\ell,c_k).
\end{equation*}
\end{proof}
\section{Bounds on Hajek estimator bias}
\label{apdx.sec:Hajek_bound}
We consider here the NMR Hajek estimator ($(8)$ in the main text) since it is a generalization of the common Hajek estimator $(4)$. As in the proof of Theorem 2,
let $\mathcal{A}=\{\bA^1,\dots,\bA^M\}$ be the collection of $M$ networks. Assume that $\bA^j \in \mathscr{A}^\ast$ for some $j$. The Hajek estimator is given by
\begin{equation*}
    \hat{\mu}^{H}_{\mathcal{A}}(c_\ell) =   
   \frac{\sum_{i=1}^{n}
    I_i^{(\mathcal{A})}(\bZ,c_\ell) 
    \frac{1}{p_i^{(\mathcal{A})}(c_\ell)}Y_i}
    {\sum_{i=1}^{n}
    I_i^{(\mathcal{A})}(\bZ,c_\ell) 
    \frac{1}{p_i^{(\mathcal{A})}(c_\ell)}} =
    \frac{V_1}{V_2}
\end{equation*}
with $V_1$ being the numerator and $V_2$ the denominator.
As already been established in Web Appendix~\ref{apdx.sec:proof}, 
\begin{equation*}
    \begin{split}
    \mathbb{E}_{\bZ}[V_1] &=
    \mathbb{E}_{\bZ}\left[\sum_{i=1}^{n}
    I_i^{(\mathcal{A})}(\bZ,c_\ell) 
    \frac{1}{p_i^{(\mathcal{A})}(c_\ell)}Y_i\right] = \sum_{i=1}^{n}\widetilde{Y}_i(c_\ell)
    \\
    \mathbb{E}_{\bZ}[V_2] 
    &=
    \mathbb{E}_{\bZ}\left[\sum_{i=1}^{n}
    I_i^{(\mathcal{A})}(\bZ,c_\ell) 
    \frac{1}{p_i^{(\mathcal{A})}(c_\ell)}\right] = n        
    \end{split}
\end{equation*}
Thus, $\frac{\mathbb{E}_{\bZ}[V_1]}{\mathbb{E}_{\bZ}[V_2]} = \mu(c_\ell)$, i.e., the Hajek estimator is the ratio of two unbiased estimators. However, such a ratio is not unbiased in itself. The bias bound of the Hajek ratio estimator is proportional to the variance of $V_1$ and $V_2$ \citep{HARTLEY1954, Saerndal2003}
\begin{equation}
    \label{eq:hajek_bounds}
    \bigg\vert \hat{\mu}^{H}_{\mathcal{A}}(c_\ell) - \mu(c_\ell) \bigg\vert \leq 
    \sqrt{Var_{\bZ}(V_1)Var_{\bZ}(V_2)}.
\end{equation}
Under some limitation on the asymptotic network structure, it can be shown that the bias bound \eqref{eq:hajek_bounds} converges to zero \citep{ugander2013,aronow_estimating_2017,Saevje2023,Li2021}.

\section{Variance of the NMR estimators}
\label{apdx.sec:NMR_var}
In this section, we derive the variance of the NMR estimators, and, following \cite{aronow2013conservative}, suggest a conservative variance estimator.

As in the proof of Theorem 2,
let $\mathcal{A}=\{\bA^1,\dots,\bA^M\}$ be the collection of $M$ networks. 
Assume throughout that $\bA^j \in \mathscr{A}^\ast$ for some $j$, i.e., $\mathcal{A}$ contains a correctly specified network.
Define $p_{ij}^{(\mathcal{A})}(c_\ell,c_k) = \mathbb{E}_{\bZ}\left[
    I_i^{(\mathcal{A})}(\bZ,c_\ell)I_j^{(\mathcal{A})}(\bZ,c_k) 
    \right]$
    as the joint probability that units $i,j$ have exposure values $c_\ell, c_k$, respectively, under all the networks in $\mathcal{A}$, and for brevity denote $p_{ij}^{(\mathcal{A})}(c_\ell, c_\ell) = p_{ij}^{(\mathcal{A})}(c_\ell)$.
The variance of the HT NMR estimator $\hat{\tau}_{\mathcal{A}}$ $(7)$ is given by \citep{Saerndal2003}
\begin{equation}
    \label{eq:ht_ce_var}
    \begin{split}
        Var_{\bZ}\Big[\hat{\tau}_{\mathcal{A}}(c_k, c_\ell)\Big]
        = 
        Var_{\bZ}\Big[\hat{\mu}_{\mathcal{A}}(c_k)\Big] + 
        Var_{\bZ}\Big[\hat{\mu}_{\mathcal{A}}(c_\ell)\Big] - 
        2Cov_{\bZ}\bigg[\hat{\mu}_{\mathcal{A}}(c_k),\;
      \hat{\mu}_{\mathcal{A}}(c_\ell)\bigg],
    \end{split}
\end{equation}
with
\begin{equation}
    \label{eq:var_NMR_HT}
    \begin{split}
    Var_{\bZ}\Big[\hat{\mu}_{\mathcal{A}}(c_\ell)\Big]
        &=
        n^{-2}\sum_{i=1}^{n} p_i^{(\mathcal{A})}(c_\ell)
        \left(1-p_i^{(\mathcal{A})}(c_\ell)\right)
        \left( \frac{\widetilde{Y}_i(c_\ell)}{p_i^{(\mathcal{A})}(c_\ell)}\right)^2 
        \\ &+  
        n^{-2}\sum_{i=1}^{n}\sum_{j \in 
        \{j \;\vert\; j\neq i, \;  p_{ij}^{(\mathcal{A})}(c_\ell)>0\}} 
        \left(p_{ij}^{(\mathcal{A})}(c_\ell) - p_i^{(\mathcal{A})}(c_\ell)p_j^{(\mathcal{A})}(c_\ell) \right)
        \frac{\widetilde{Y}_i(c_\ell)\widetilde{Y}_j(c_\ell)}{p_i^{(\mathcal{A})}(c_\ell)p_j^{(\mathcal{A})}(c_\ell)}
        \\ &-
        n^{-2}\sum_{i=1}^{n}\sum_{j \in 
        \{j \;\vert\; j\neq i, \; p_{ij}^{(\mathcal{A})}(c_\ell)=0\}}
        \widetilde{Y}_i(c_\ell)\widetilde{Y}_j(c_\ell),
    \end{split}
\end{equation}
and,
\begin{equation}
    \label{eq:cov_NMR_ht}
    \begin{split}
    Cov_{\bZ}\bigg[\hat{\mu}_{\mathcal{A}}(c_k),\;
      \hat{\mu}_{\mathcal{A}}(c_\ell)\bigg]
    &= 
    n^{-2}\sum_{i=1}^{n}
    \sum_{j \in 
        \{j \;\vert\; j\neq i, \; p_{ij}^{(\mathcal{A})}(c_k,c_\ell)>0\}}
      \left(p_{ij}^{(\mathcal{A})}(c_k, c_\ell) - p_i^{(\mathcal{A})}(c_k)p_j^{(\mathcal{A})}(c_\ell) \right)
        \frac{\widetilde{Y}_i(c_k)\widetilde{Y}_j(c_\ell)}{p_i^{(\mathcal{A})}(c_k)p_j^{(\mathcal{A})}(c_\ell)}
        \\ &-
        n^{-2}\sum_{i=1}^{n}
         \sum_{j \in 
        \{j \;\vert\;  p_{ij}^{(\mathcal{A})}(c_k,c_\ell)=0\}}
         \widetilde{Y}_i(c_k)\widetilde{Y}_j(c_\ell).
    \end{split}
\end{equation}
The first two terms in the variance \eqref{eq:var_NMR_HT} and the first term in the covariance \eqref{eq:cov_NMR_ht} can be estimated in an unbiased manner using an unbiased Horvitz-Thompson estimator \citep{aronow2013conservative}. However, the third term in \eqref{eq:var_NMR_HT} and the  second term in \eqref{eq:cov_NMR_ht} involve potential outcomes that have zero probabilities to be jointly observed ($p_{ij}^{(\mathcal{A})} =0$), and thus, these terms are not directly estimable from the observed data. We follow \citet{aronow2013conservative} and use a conservative estimator that utilizes Young's inequality.
The inequality states that
\begin{equation*}
    \frac{a^{r}}{r} + \frac{b^q}{q} \geq ab,\quad \text{for}\; a,b>0,\; \text{and}\; \frac{1}{r} + \frac{1}{q} = 1, r,q > 0.
\end{equation*}
Thus, for $r=q=2$
\begin{equation*}
    \frac{\widetilde{Y}_i(c_k)^2}{2} + \frac{\widetilde{Y}_j(c_\ell)^2}{2} =
    \frac{\vert \widetilde{Y}_i(c_k) \vert ^2}{2} + \frac{\vert \widetilde{Y}_j(c_\ell) \vert ^2}{2}
    \geq 
    \vert \widetilde{Y}_i(c_k) \vert \cdot 
    \vert \widetilde{Y}_j(c_\ell) \vert
\end{equation*}
Since any two numbers $x,y$ satisfies $\vert x \vert \vert y\vert \geq xy$ and $\vert x \vert \vert y\vert \geq -xy$, we obtain the bounds
\begin{equation}
\label{eq:first.bound}
    -\sum_{i=1}^{n}
     \sum_{j=1}^{n}
     \widetilde{Y}_i(c_\ell)\widetilde{Y}_j(c_\ell) 
     \leq 
     \sum_{i=1}^{n}
     \sum_{j=1}^{n} \frac{\widetilde{Y}_i(c_\ell)^2}{2} +
     \frac{\widetilde{Y}_j(c_\ell)^2}{2},
\end{equation}
\begin{equation}
    \label{eq:second.bound}
     -\sum_{i=1}^{n}
     \sum_{j=1}^{n}
     \widetilde{Y}_i(c_k)\widetilde{Y}_j(c_\ell) 
     \geq 
     -\sum_{i=1}^{n}\
     \sum_{j=1}^{n} 
     \frac{\widetilde{Y}_i(c_k)^2}{2}
     +
     \frac{\widetilde{Y}_j(c_\ell)^2}{2},
\end{equation}
and the RHS in both \eqref{eq:first.bound} and \eqref{eq:second.bound} can be estimated by an Horvitz-Thompson estimator. 
We can thus use the Horvitz-Thompson variance and covariance estimators
\begin{equation*}
    \begin{split}
          \reallywidehat{Var}\Big[\hat{\mu}_{\mathcal{A}}(c_\ell)\Big] 
        &=
        n^{-2}\sum_{i=1}^{n} 
         I_i^{(\mathcal{A})}(\bZ,c_\ell)
        \left(1-p_i^{(\mathcal{A})}(c_\ell)\right)
        \left( \frac{Y_i}{p_i^{(\mathcal{A})}(c_\ell)}\right)^2
        \\ &+ 
        n^{-2}\sum_{i=1}^{n}\sum_{j \in \{j \;\vert\; j\neq i, \; p_{ij}^{(\mathcal{A})}(c_k,c_\ell)>0\}}
        \Bigg(
         I_i^{(\mathcal{A})}(\bZ,c_\ell) I_j^{(\mathcal{A})}(\bZ,c_\ell)
         \cdot
        \frac{p_{ij}^{(\mathcal{A})}(c_\ell) - p_i^{(\mathcal{A})}(c_\ell)p_j^{(\mathcal{A})}(c_\ell)}{p_{ij}^{(\mathcal{A})}(c_\ell)}
        \cdot 
        \\ &\qquad \qquad 
        \qquad \qquad \qquad \qquad \qquad \qquad \qquad 
        \frac{Y_i}{p_i^{(\mathcal{A})}(c_\ell)}
        \cdot
        \frac{Y_j}{p_j^{(\mathcal{A})}(c_\ell)} \Bigg)
        \\ &+ 
        n^{-2}\sum_{i=1}^{n}\sum_{j \in 
        \{j \;\vert\; j\neq i, \; p_{ij}^{(\mathcal{A})}(c_\ell)=0\}}
        \left(
         \frac{I_i^{(\mathcal{A})}(\bZ,c_\ell)\cdot Y_i^2}{2\cdot p_i^{(\mathcal{A})}(c_\ell)} +
         \frac{I_j^{(\mathcal{A})}(\bZ,c_\ell)\cdot Y_j^2}{2\cdot p_j^{(\mathcal{A})}(c_\ell)}
         \right)
 \end{split}
    \end{equation*}
\begin{equation*}
    \begin{split}
        \reallywidehat{Cov}\Big[\hat{\mu}_{\mathcal{A}}(c_k),\;
      \hat{\mu}_{\mathcal{A}}(c_\ell)\Big]
        &= 
        n^{-2}\sum_{i }
        \sum_{j \in 
        \{j \;\vert\; j\neq i, \; p_{ij}^{(\mathcal{A})}(c_k,c_\ell)>0\}}
         \Biggl(
         I_i^{(\mathcal{A})}(\bZ,c_k) I_j^{(\mathcal{A})}(\bZ,c_\ell)
         \cdot
        \frac{p_{ij}^{(\mathcal{A})}(c_k, c_\ell) - p_i^{(\mathcal{A})}(c_k)p_j^{(\mathcal{A})}(c_\ell)}{p_{ij}^{(\mathcal{A})}(c_k,c_\ell)}
        \\ &\qquad \qquad \qquad \qquad \qquad \qquad \qquad 
        \cdot
        \frac{Y_i}{p_i^{(\mathcal{A})}(c_k)}
        \cdot
        \frac{Y_j}{p_j^{(\mathcal{A})}(c_\ell)}
        \Biggr)
        \\ &- 
        n^{-2}
        \sum_{i }
        \sum_{j \in 
        \{j \;\vert\; p_{ij}^{(\mathcal{A})}(c_k,c_\ell)=0\}}
        \left(
         \frac{I_i^{(\mathcal{A})}(\bZ,c_k)\cdot Y_i^2}{2\cdot p_i^{(\mathcal{A})}(c_k)} +\
         \frac{I_j^{(\mathcal{A})}(\bZ,c_\ell)\cdot Y_j^2}{2\cdot p_j^{(\mathcal{A})}(c_\ell)}
         \right),
    \end{split}
\end{equation*}
to obtain a plug-in estimator of \eqref{eq:ht_ce_var}
\begin{equation}
    \label{eq:ht_ce_var_esti}
    \begin{split}
        \reallywidehat{Var}\Big[\hat{\tau}_{\mathcal{A}}(c_k, c_\ell)\Big]
        = 
        \reallywidehat{Var}\Big[\hat{\mu}_{\mathcal{A}}(c_k)\Big] + 
        \reallywidehat{Var}\Big[\hat{\mu}_{\mathcal{A}}(c_\ell)\Big] - 
        2\cdot \reallywidehat{Cov}\Big[\hat{\mu}_{\mathcal{A}}(c_k),\;
      \hat{\mu}_{\mathcal{A}}(c_\ell)\Big].
    \end{split}
\end{equation}
As formally presented below, the variance estimator \eqref{eq:ht_ce_var_esti} is a conservative estimator.
\begin{proposition}
    \label{prop:conserv_ce_estimator}
    If $\bA^j \in \mathscr{A}^\ast$ for some $j$, then
     \begin{equation*}   \mathbb{E}_{\bZ}\left[\reallywidehat{Var}(\hat{\tau}_{\mathcal{A}}(c_k, c_\ell))\right]
         \geq 
         Var_{\bZ}\Big[\hat{\tau}_{\mathcal{A}}(c_k, c_\ell)\Big] 
         , \; k,\ell = 1,\dots,L.
     \end{equation*}
\end{proposition}
\begin{proof}
    The proof stems directly from \citet{aronow2013conservative} derivations using the fact that $\mathbb{E}_{\bZ}\left[\frac{I_i^{(\mathcal{A})}(\bZ,c_k)}{p_i^{(\mathcal{A})}(c_k)}\right] = 1$
    and that if $\bA^j \in \mathscr{A}^\ast$ for some $j$ then $I_i^{(\mathcal{A})}(\bZ,c_k)Y_i= I_i^{(\mathcal{A})}(\bZ,c_k) \widetilde{Y}_i(c_k)$.
\end{proof}
Variance estimation of the Hajek NMR estimator $(8)$ is done with first order Taylor linear approximation \citep{Saerndal2003} by replacing $Y_i$ in \eqref{eq:ht_ce_var_esti} with the residuals $U_i = Y_i - \hat{\mu}_{\mathcal{A}}^{H}(c_k)$ where $c_k$ is the observed exposure value for unit $i$. 

A numerical illustration of the conservativeness property via a simulation study is Web Appendix~\ref{apdx.sec:simulations_and_data_analyses}.

\section{Asymptotic properties of NMR estimators}
\label{apdx.sec:asymptotic}

We establish asymptotic properties of the NMR estimators in a growing sequence of populations. We assume throughout that $\mathcal{A}=\{\bA_1,\ldots,\bA_M\}$ is a collection of fixed size $M$ containing at least one correctly specified network, that is, $\mathcal{A} \cap \mathscr{A}^\ast \neq \emptyset$. Specifically, each $\bA_m \in \mathcal{A}$ is a function of $n$, i.e., $\bA_m = \bA_m(n)$. 

As in previous works \citep{aronow_estimating_2017, Leung2020, Saevje2023}, for both consistency and CLT, we have to limit the growth of the pairwise exposure's covariance. 
Define $p_{ij}^{(\mathcal{A})}(c_\ell,c_k) = \mathbb{E}_{\bZ}\left[
I_i^{(\mathcal{A})}(\bZ,c_\ell)I_j^{(\mathcal{A})}(\bZ,c_k) 
\right]$
as the joint probability that units $i$ and $j$ have exposures $c_\ell$ and $c_k$, respectively, under all the networks in $\mathcal{A}$. The exposures covariance for two units is  
$$ Cov\big(I_i^{(\mathcal{A})}(\bZ,c_\ell), I_j^{(\mathcal{A})}(\bZ,c_k)
    \big) 
    = p_{ij}^{(\mathcal{A})}(c_\ell,c_k) -  p_i^{(\mathcal{A})}(c_\ell) p_j^{(\mathcal{A})}(c_k)
$$
Note that all the above terms may change with $n$. We assume that the sum of all pairwise covariance terms satisfies the following assumption.
\begin{assumption}[Pairwise dependence]
    \label{ass:exposure_dependence}
    $
    \sum_{i=1}^{n}\sum_{j\neq i}
    \big(p_{ij}^{(\mathcal{A})}(c_\ell,c_k) -  p_i^{(\mathcal{A})}(c_\ell) p_j^{(\mathcal{A})}(c_k)\big) = o(n^2)
    $
    for all $c_\ell,c_k \in \mathcal{C}$.
\end{assumption}
We begin by showing that the NMR estimators $\hat{\tau}_{\mathcal{A}}(c_\ell,c_k)$ are consistent and then show the CLT and resulting confidence intervals based on the conservative variance estimator.

\subsection{Consistency}

\setcounter{theorem}{2}
\begin{theorem}[Consistency]
    \label{thm:consistency}
    Assume that each network in $\mathcal{A}$ satisfies Definition 1 (positivity). Under Assumptions 1,2,5, if $\mathcal{A}\cap \mathscr{A}^\ast \neq \emptyset$ then for all $c_\ell,c_k \in \mathcal{C}$, $\hat{\tau}_{\mathcal{A}}(c_\ell,c_k) - \tau(c_\ell, c_k) \xrightarrow{p} 0
     $ as $n \rightarrow \infty$, where $p$ denotes convergence in probability.
\end{theorem}

\begin{proof}
 As all networks in $\mathcal{A}$ satisfy positivity (Definition 1), there exists a constant $\kappa_2 > 0$ such that $\lvert 1/p_i^{(\mathcal{A})}(c_\ell) \rvert \leq \kappa_2$ for all $i$ and $c_\ell$.

The variance of $\hat{\mu}_{\mathcal{A}}(c_\ell)$ is (Section \ref{apdx.sec:NMR_var})

\begin{equation*}
    \begin{split}
    Var_{\bZ}\Big[\hat{\mu}_{\mathcal{A}}(c_\ell)\Big]
        &=
        n^{-2}\sum_{i=1}^{n} p_i^{(\mathcal{A})}(c_\ell)
        \left(1-p_i^{(\mathcal{A})}(c_\ell)\right)
        \left( \frac{\widetilde{Y}_i(c_\ell)}{p_i^{(\mathcal{A})}(c_\ell)}\right)^2 
        \\ &+  
        n^{-2}\sum_{i=1}^{n}\sum_{j \in 
        \{j \;\vert\; j\neq i, \;  p_{ij}^{(\mathcal{A})}(c_\ell)>0\}} 
        \left(p_{ij}^{(\mathcal{A})}(c_\ell) - p_i^{(\mathcal{A})}(c_\ell)p_j^{(\mathcal{A})}(c_\ell) \right)
        \frac{\widetilde{Y}_i(c_\ell)\widetilde{Y}_j(c_\ell)}{p_i^{(\mathcal{A})}(c_\ell)p_j^{(\mathcal{A})}(c_\ell)}
        \\ &-
        n^{-2}\sum_{i=1}^{n}\sum_{j \in 
        \{j \;\vert\; j\neq i, \; p_{ij}^{(\mathcal{A})}(c_\ell)=0\}}
        \widetilde{Y}_i(c_\ell)\widetilde{Y}_j(c_\ell)
        \\ &=
        n^{-2}\sum_{i=1}^{n} p_i^{(\mathcal{A})}(c_\ell)
        \left(1-p_i^{(\mathcal{A})}(c_\ell)\right)
        \left( \frac{\widetilde{Y}_i(c_\ell)}{p_i^{(\mathcal{A})}(c_\ell)}\right)^2 
        \\ &+  
        n^{-2}\sum_{i=1}^{n}\sum_{j \neq i} 
        \left(p_{ij}^{(\mathcal{A})}(c_\ell) - p_i^{(\mathcal{A})}(c_\ell)p_j^{(\mathcal{A})}(c_\ell) \right)
        \frac{\widetilde{Y}_i(c_\ell)\widetilde{Y}_j(c_\ell)}{p_i^{(\mathcal{A})}(c_\ell)p_j^{(\mathcal{A})}(c_\ell)}
        \\ &\leq 
         n^{-2} \Big(\frac{\kappa}{\kappa_2}\Big)^2\sum_{i=1}^{n} p_i^{(\mathcal{A})}(c_\ell)
        \left(1-p_i^{(\mathcal{A})}(c_\ell)\right)
        \\ &+ 
        n^{-2} \Big(\frac{\kappa}{\kappa_2}\Big)^2 
        \sum_{i=1}^{n}\sum_{j \neq i} 
        \left(p_{ij}^{(\mathcal{A})}(c_\ell) - p_i^{(\mathcal{A})}(c_\ell)p_j^{(\mathcal{A})}(c_\ell) \right)
        \\ &\leq 
        n^{-1} \frac{1}{2} \Big(\frac{\kappa}{\kappa_2}\Big)^2 + o(1),
    \end{split}
\end{equation*}
where the first inequality follows from Assumption 2 (bounded potential outcomes) and positivity, and the second inequality since $p_i(1-p_i)\leq 1/2$ and Assumption 5. Therefore, $Var_{\bZ}\Big[\hat{\mu}_{\mathcal{A}}(c_\ell)\Big] \rightarrow 0$ as $n \rightarrow \infty$, and from Chebyshev's inequality, $\hat{\mu}_{\mathcal{A}}(c_\ell) - \mu(c_\ell) \xrightarrow{p} 0$ as $n \rightarrow \infty$ for all $c_\ell$. From the Continuous Mapping Theorem, we obtain that 
$
\hat{\tau}_{\mathcal{A}}(c_\ell,c_k) - \tau(c_\ell, c_k) \xrightarrow{p} 0
     \; \text{as} \; n \rightarrow \infty.
$
\end{proof}

\subsection{CLT and confidence intervals}
The CLT argument is based on dependency graphs CLT derived by \citet{baldi1989normal}.

The dependency graph $G_n = (V_n,E_n)$ is an undirected graph with vertices $\vert V_n \rvert = n$ and edge set $E_n$ that describes the dependencies between exposures indicators $I_i^{(\mathcal{A})}(\bZ,c_{\ell})$. Formally, $(i,j) \in E_n$ if there exists $c_\ell,c_k \in \mathcal{C}$ such that  $I_i^{(\mathcal{A})}(\bZ,c_{\ell})$ and $I_j^{(\mathcal{A})}(\bZ,c_k)$ are dependent for $i \neq j$. Define the degrees in $G_n$ by $D_{n,i} = \lvert j: (i,j) \in E_n \rvert$ and denote the maximal degree by $D_{n,max} = \max_i D_{n,i}$. The degrees $D_{n,i}$ 
represent the number of units that have dependent exposures with $i$ in $\mathcal{A}$. The maximal degree $D_{n,max}$ correspond to the unit with the highest number of dependent exposures. We assume that this degree is bounded for each $G_n$. 

\begin{assumption}[Bounded degree]
    \label{ass:bound_deg}
    There exists a finite constant $\kappa_3$ such that $D_{n,max} \leq \kappa_3$ for all $n>1$.
\end{assumption}
Assumption~\ref{ass:bound_deg} can also be relaxed for $\kappa_3$ that grows at some sub-linear rate.
Assumption~\ref{ass:bound_deg} implies that in the limit there is a constraint on the number of units with dependent exposures. 
Under neighborhood interference and Bernoulli experimental design, Assumption~\ref{ass:bound_deg} is directly related to the degrees of the networks in $\mathcal{A}$ as it precludes a unit that is connected to all other units. 

\begin{theorem}[CLT]
    \label{thm:clt}
    Assume that each network in $\mathcal{A}$ satisfies Definition 1 (positivity). Under Assumptions 1,2,5,6, if $\mathcal{A}\cap \mathscr{A}^\ast \neq \emptyset$ then for all $c_\ell,c_k \in \mathcal{C}$, 
    
    $$\frac{\hat{\tau}_{\mathcal{A}}(c_\ell,c_k) - \tau(c_\ell,c_k)}{\sqrt{ Var_{\bZ}\Big[\hat{\tau}_{\mathcal{A}}(c_\ell,c_k)\Big]}} \xrightarrow{d}
     N(0,1),\; \text{as} \; n \rightarrow \infty,$$
     where $d$ denotes convergence in distribution.
\end{theorem}

\begin{proof}
     By Theorem 2, $\mathbb{E}_{\bZ}\Big[\hat{\tau}_{\mathcal{A}}(c_\ell,c_k)\Big] = \tau(c_\ell,c_k)$ for all $n$. From a similar derivation to the one provided in the proof of Theorem \ref{thm:consistency}, we obtain that $Var_{\bZ}\Big[\hat{\tau}_{\mathcal{A}}(c_\ell,c_k)\Big] = O(n)$. By Assumption 2 (bounded potential outcomes) and Definition 1 (positivity) all items in the estimator $\hat{\tau}_{\mathcal{A}}(c_\ell,c_k)$ are bounded. By Assumption \ref{ass:bound_deg}, $D^2_{n,max} \leq \kappa_3^2 < \infty$. Since $\lvert V_n \rvert = n$ in the dependency graph $G_n$ we obtain that $\frac{\lvert V_n \rvert}{Var_{\bZ}\Big[\hat{\tau}_{\mathcal{A}}(c_\ell,c_k)\Big]^{3/2}} \rightarrow 0$ as $n \rightarrow \infty$. The CLT thus follows from \citet[Corollary 2]{baldi1989normal}.
\end{proof}

Finally, the following theorem shows that constructing confidence intervals with the conservative variance estimator \eqref{eq:ht_ce_var_esti} are asymptotically valid.

\begin{theorem}[Confidence intervals]
    \label{thm:valid_ci}
    Define confidence interval with $1-\alpha$ confidence level by
    $$
    \reallywidehat{CI}_\alpha = \Big[\hat{\tau}_{\mathcal{A}}(c_\ell,c_k) \pm z_{1-\alpha/2} \sqrt{\reallywidehat{Var}\big(\hat{\tau}_{\mathcal{A}}(c_\ell,c_k)\big)}\Big].
    $$
    Under the conditions stated in Theorem \ref{thm:clt}, $\Pr\big(\tau(c_\ell,c_k) \in \reallywidehat{CI}_\alpha \big) \rightarrow c \geq 1-\alpha$ as $n \rightarrow \infty$.
\end{theorem}

\begin{proof}
    By Proposition \ref{prop:conserv_ce_estimator}, 
     \begin{equation*}   
     \frac{Var_{\bZ}\Big[\hat{\tau}_{\mathcal{A}}(c_k, c_\ell)\Big]}{\mathbb{E}_{\bZ}\left[\reallywidehat{Var}(\hat{\tau}_{\mathcal{A}}(c_k, c_\ell))\right]}
     \in [0,1],
 \end{equation*}
     assuming finite expectation. We can write 
     \begin{equation*}
     \reallywidehat{Var}(\hat{\tau}_{\mathcal{A}}(c_k, c_\ell)) = n^{-2}\sum_i \sum_j \phi_{ij}(\bZ),    
     \end{equation*}
    for some random variables $\phi_{ij}(\bZ)$ that depends on the indicator of exposures and other constants such as the potential outcomes and probabilities of exposures which we can bound. By positivity and bounded potential outcomes, each $\phi_{ij}$ is bounded with probability one. We have
    \begin{equation*}
        Var_{\bZ}\Big[\reallywidehat{Var}(\hat{\tau}_{\mathcal{A}}(c_k, c_\ell)) \Big] =
        n^{-4} \sum_i\sum_j\sum_{i'}\sum_{j'} Cov \big(\phi_{ij}(\bZ), \phi_{i',j'}(\bZ) \big).
    \end{equation*}
    But $Cov \big(\phi_{ij}(\bZ), \phi_{i',j'}(\bZ) \big)$ is non-zero only if $(i,j)=(i',j')$ or $i,i'$ or $j,j'$ are connected in the dependency graph $G_n$. Since the covariance can be bounded for each $i,j,i',j'$ we obtain that the entire quadruple sum is $O(n^2D_{n,max}^2)$. Thus, $Var_{\bZ}\Big[\reallywidehat{Var}(\hat{\tau}_{\mathcal{A}}(c_k, c_\ell)) \Big] = O(n^{-2} D_{n,max}^2)$ which by Assumption \ref{ass:bound_deg} will converges to zero as $n \rightarrow \infty$. 
    Consequently, 
    \begin{equation*}
        \frac{Var_{\bZ}\Big[\hat{\tau}_{\mathcal{A}}(c_k, c_\ell)\Big]}{\reallywidehat{Var}(\hat{\tau}_{\mathcal{A}}(c_k, c_\ell))} \rightarrow
      c \in [0,1], \; \text{as} \; n \rightarrow \infty.
    \end{equation*}
    From Theorem \ref{thm:clt}, the statistic 
    \begin{equation*}
        \frac{\hat{\tau}_{\mathcal{A}}(c_\ell,c_k) - \tau(c_\ell,c_k)}{\sqrt{ Var_{\bZ}\Big[\hat{\tau}_{\mathcal{A}}(c_\ell,c_k)\Big]}}
    \end{equation*}
    converges to standard normal distribution. 
    Therefore, the confidence interval 
    \begin{equation*}
        CI_\alpha = 
        \Big[\hat{\tau}_{\mathcal{A}}(c_\ell,c_k) \pm z_{1-\alpha/2} \sqrt{Var\big(\hat{\tau}_{\mathcal{A}}(c_\ell,c_k)\big)}\Big],
    \end{equation*}
    achieve nominal $1-\alpha$ coverage as $n \rightarrow \infty$. But since asymptotically $Var_{\bZ}\Big[\hat{\tau}_{\mathcal{A}}(c_k, c_\ell)\Big] \leq \reallywidehat{Var}(\hat{\tau}_{\mathcal{A}}(c_k, c_\ell))$, constructing CI with the variance estimator $\reallywidehat{CI}_\alpha$ yields
    \begin{equation*}
        1-\alpha \leq \Pr \big(\tau(c_\ell,c_k) \in CI_\alpha \big) \leq \Pr\big(\tau(c_\ell,c_k) \in \reallywidehat{CI}_\alpha \big),
    \end{equation*}
    as $n \rightarrow \infty$.
\end{proof}

\section{Exposure mapping misspecification}
\label{apdx.sec:exposure_MR}

\subsection{Expected exposure effects}

Assume that researchers estimate causal effects using the NMR estimator with a set $\mathcal{A}$ of $M$ networks. It is possible that all the networks in $\mathcal{A}$ and the exposure mapping $f$ are misspecified. However, we can use the HT (or Hajek) NMR estimators to unbiasedly and consistently estimate a variant of the \emph{expected exposure effects} defined by \citet{Saevje2023}.

Let $C^{\mathcal{A}}_i =\sum_{j=1}^{L}c_jI_i^{(\mathcal{A})}(\bZ,c_j) $ be the observed exposure for unit $i$ when all networks in $\mathcal{A}$ have the same exposure value. That is, $C^{\mathcal{A}}_i=c_\ell$ if and only if $f(\bZ,\bA_i)=c_\ell$ for all $\bA \in \mathcal{A}$. Recall that given a correct exposure mapping $f$, we defined a correctly specified interference network (Definition~2) as the network that will enable us to connect the potential outcomes $Y_i(\bZ)$ to the modified potential outcomes $\widetilde{Y}_i(c_\ell)$ expressed in terms of exposure values. If both the network and the mapping are misspecified, we cannot connect $Y(\cdot)$ to $\widetilde{Y}(\cdot)$. 
Let $\overline{Y}_i(c_\ell)=\mathbb{E}_{\bZ}\left[ Y_i(\bZ) \mid C_i^{\mathcal{A}}=c_\ell\right]$ be the expected potential outcome of unit $i$ when exposures under all the networks in $\mathcal{A}$ are $c_\ell$. Define the expected exposure effect for exposures $c_\ell,c_k \in \mathcal{C}$ as
\begin{equation}
    \label{eq.apdx:eee}
    \overline{\tau}(c_\ell,c_k)=\frac{1}{n}\sum_{i=1}^{n}\left(\overline{Y}_i(c_\ell)-\overline{Y}_i(c_k) \right).
\end{equation}
Eq.~\eqref{eq.apdx:eee} is a variant of the estimand proposed by \citet{Saevje2023} as it conditions on the exposures under multiple networks instead of a single network.
Now, for any $c_\ell \in \mathcal{C}$ we can write
\begin{align*}
    \mathbb{E}_{\bZ}\left[\frac{I_i^{(\mathcal{A})}(\bZ,c_\ell)Y_i}{ p_i^{(\mathcal{A})}(c_\ell)} \right] 
    &= \frac{\mathbb{E}_{\bZ}\left[I_i^{(\mathcal{A})}(\bZ,c_\ell)Y_i\right]}{ p_i^{(\mathcal{A})}(c_\ell)}
    \\ &= 
    \frac{p_i^{(\mathcal{A})}(c_\ell) \mathbb{E}_{\bZ}\left[Y_i \mid C_i^{\mathcal{A}} = c_\ell \right]}{ p_i^{(\mathcal{A})}(c_\ell)}
    \\ &= 
    \mathbb{E}_{\bZ}\left[Y_i(\bZ) \mid C_i^{\mathcal{A}} = c_\ell \right]
    \\ &= 
     \overline{Y}_i(c_\ell),
\end{align*}
where the second equality results from the law of total expectation, and the third equality from consistency in its general form $Y_i = Y_i(\bZ)$ (without exposure mappings). 
Thus, the HT NMR estimator $\hat{\tau}_{\mathcal{A}}(c_\ell, c_k)$ is unbiased estimator of \eqref{eq.apdx:eee}. 
Under bounded potential outcomes (Assumption~2 in the main text), positivity of all networks in $\mathcal{A}$, Assumption~\ref{ass:exposure_dependence} (which is equivalent to Condition~3 of \citet{Saevje2023} for the case of joint probabilities of exposures under multiple networks), and additional limitations on the amount of specification error dependence, the results of \citet{Saevje2023} can be adapted to show that the NMR estimator is consistent estimator of the expected exposure effect \eqref{eq.apdx:eee}.

\subsection{Exposure misspecification robust estimator} 

We sketch how the NMR estimator can be modified to settings where the exposure mapping $f$, not the interference network, might be misspecified. In this scenario, researchers have a collection of possible mappings but are not sure which one is correct. We show how to construct a robust estimator that is unbiased if one of the mapping is correct. We modify the assumptions and definitions in the paper accordingly to this setup.

We assume that $\bA^\ast$ is the interference network.  
Now, the mapping $f$ is unknown but a part of a larger space of possible mappings. Denote the set of all exposure mappings with the image set $\mathcal{C}$ by $\mathscr{F} = \left\{f : Im(f) = \mathcal{C} = \{c_1, \ldots,c_L\} \right\}$. 
For example, under the four-level exposure mapping with thresholds (Eq.~(6)), $\mathscr{F}$ is the infinite set of all mappings with different threshold values.

Write the exposure probabilities under mapping $f$ as $p_i^{(f)}(c_\ell) = \mathbb{E}_{\bZ}\left[I(f(\bZ,A^\ast_i)=c_\ell) \right]$.
Positivity (Definition~1) is modified to 

\paragraph{Definition~1(M) (Positivity; modified).} We say that $f \in \mathscr{F}$ satisfies positivity if $p_i^{(f)}(c_\ell) > 0$ for all $i=1,\ldots,n$ and $c_\ell \in \mathcal{C}$.
\\
\\
The definition of a correctly specified interference structure (Definition 2) also needs to be modified to the specification of the exposure mapping instead of the network.

\paragraph{Definition~2(M) (Correctly specified interference structure; modified)} For an interference network $\bA^\ast$, we say that $f \in \mathscr{F}$ correctly specifies the exposure mapping, if $f$ satisfies Definition 1 (positivity; modified), and for all $\bz \in \mathcal{Z}$
\begin{equation*}
    \text{if} \; f(\bz,\bA^\ast_i)=c_\ell,\; \text{then}\; Y_i(\bz) = \widetilde{Y}_i(c_\ell), \; i=1,\ldots,n.
\end{equation*}
\\
If some $f \in \mathscr{F}$ satisfies Definition~2(M), then for any $\bz,\bz'$, if $f(\bz,\bA^\ast_i) = f(\bz',\bA^\ast_i)$ then $Y_i(\bz)=Y_i(\bz')$.
Similarly to the class $\mathscr{A}^\ast$ of correctly specified networks, we can define $\mathscr{F}^\ast$ as the class of all mappings $f \in \mathscr{F}$ that satisfy Definition~2(M), given an interference network $\bA^\ast$.
As with $\mathscr{A}^\ast$, the class $\mathscr{F}^\ast$ does not necessarily contain a singleton, i.e., $f^\ast \in \mathscr{F}^\ast$ is not necessarily unique.
Since all mappings have the same image space, we can define causal estimands as before, that is, as contrasts $\tau(c_\ell,c_k)=\mu(c_\ell)-\mu(c_k)$.
The consistency assumption (Assumption~1) is modified to 

\paragraph{Assumption~1(M) (Consistency; modified).} The observed outcomes are generated from one of the potential outcomes by 
$$Y_i = \sum_{j=1}^{L} \mathbb{I}\{f^\ast(\bZ,\bA^\ast_i)=c_j\}\widetilde{Y}_i(c_j), \; i=1,\ldots,n, \; f^\ast \in \mathscr{F}^\ast.$$
\\
\\
Even if $\mathscr{F}^\ast$ is not a singleton, all mappings in it will result in the same observed outcomes. That is, the sum $\sum_{j=1}^{L} \mathbb{I}\{f^\ast(\bZ,\bA^\ast_i)=c_j\}\widetilde{Y}_i(c_j)$ is constant for any $f^\ast \in \mathscr{F}^\ast$. Otherwise, if two mappings in $\mathscr{F}^\ast$ will yield two different potential outcomes for a given $\bZ$, we will either have a contradiction to Definition~2(M) or the sharp null hypothesis will hold for some exposure values.

Now, assume researchers have $\widetilde{M}$ possible mappings $\mathcal{F} =\{f^1,\ldots,f^{\widetilde{M}}\}$ but are not sure which one, if any, is a correctly specified exposure mapping. Define $I_i^{(\mathcal{F})}(\bZ,c_\ell) = \prod_{f \in \mathcal{F}} \mathbb{I}\{f(\bZ,\bA^\ast_i)=c_\ell \}$ to be the indicator that equals one only if the exposure value equals $c_\ell$ under each of the mappings in $\mathcal{F}$.
Denote the joint probability that unit $i$ has exposure value $c_\ell$ under \emph{all} mappings $f \in \mathcal{F}$ by $p_i^{(\mathcal{F})}(c_\ell)=\mathbb{E}_{\bZ}\left[I_i^{(\mathcal{F})}(\bZ,c_\ell)\right]$.
Define the \emph{exposure misspecification robust} (EMR) estimator as 
\begin{equation}
    \label{eq.apdx:emr}
    \hat{\mu}_{\mathcal{F}}(c_\ell) = 
    \frac{1}{n}
    \sum_{i=1}^{n}\frac{I_i^{(\mathcal{F})}(\bZ,c_\ell)}{p_i^{(\mathcal{F})}(c_\ell)} Y_i.
\end{equation}
The following theorem asserts the EMR estimator is unbiased if $\mathcal{F}$ include a correctly specified mapping.
\paragraph{Theorem~2(M) (modified).} Let $\mathcal{F}$ be a collection of $\widetilde{M}$ exposure mapping such that each of mappings satisfies Definition~1(M) . Under Assumption~1(M), if $\mathcal{F} \cap \mathscr{F}^\ast \neq \emptyset$, then for $c_\ell \in \mathcal{C}$
\begin{equation*}
    \mathbb{E}_{\bZ}\left[ \hat{\mu}_{\mathcal{F}}(c_\ell)\right] = \mu(c_\ell).
\end{equation*}

\begin{proof}
Note that $\mathcal{F} \cap \mathscr{F}^\ast \ne \emptyset$ means that for some $j,\; f^j \in \mathscr{A}^\ast$. Assume without loss of generality that $f^1 \in \mathscr{F}^\ast$ and write $f^1 = f^\ast$. We obtain
\begin{equation*}
    \begin{split}
        \mathbb{E}_{\bZ}\left[\hat{\mu}_{\mathcal{F}}(c_\ell)\right] 
        &=
        \mathbb{E}_{\bZ}\left[\frac{1}{n}
        \sum_{i=1}^{n}
    \left(\prod_{j=1}^{\widetilde{M}}
    \mathbb{I}
    \{f^j(\bZ,\bA^\ast_i)=c_\ell\}\right)
    \frac{1}{p_i^{(\mathcal{F})}(c_\ell)}Y_i\right] 
    \\ (\text{Consistency}) &=
        \mathbb{E}_{\bZ}\left[\frac{1}{n}
        \sum_{i=1}^{n}
    \left(\prod_{j=1}^{\widetilde{M}}
    \mathbb{I}
    \{f^j(\bZ,\bA^\ast_i)=c_\ell\}\right)
    \frac{1}{p_i^{(\mathcal{F})}(c_\ell)}\sum_{k = 1}^{L}\mathbb{I}\{f^\ast(\bZ, \bA^{\ast}_i) = c_k\}\widetilde{Y}_i(c_k)\right]
    \\ 
    &=
        \mathbb{E}_{\bZ}\Biggl[\frac{1}{n}
        \sum_{i=1}^{n}
    \left(\prod_{j=2}^{\widetilde{M}}
    \mathbb{I}
    \{f^j(\bZ,\bA^\ast_i)=c_\ell\}\right)
    \frac{1}{p_i^{(\mathcal{F})}(c_\ell)}\cdot
    \\ &\qquad \qquad \qquad \qquad
    \mathbb{I}\{f^1(\bZ, \bA^\ast_i)=c_\ell\}\sum_{k = 1}^{L} \mathbb{I}\{f^\ast(\bZ, \bA^{\ast}_i) = c_k\}\widetilde{Y}_i(c_k)\Biggr] 
    \\
    (f^1 = f^\ast) &=
        \mathbb{E}_{\bZ}\Biggl[\frac{1}{n}
        \sum_{i=1}^{n}
    \left(\prod_{j=2}^{\widetilde{M}}
    \mathbb{I}
    \{f^j(\bZ,\bA^\ast_i)=c_\ell\}\right)
    \frac{1}{p_i^{(\mathcal{F})}(c_\ell)}\cdot
    \\ &\qquad \qquad \qquad \qquad
    \sum_{k = 1}^{L}\mathbb{I}\{f^\ast(\bZ, \bA^{\ast}_i) = c_\ell\}\mathbb{I}\{f^\ast(\bZ, \bA^{\ast}_i) = c_k\}\widetilde{Y}_i(c_k)\Biggr] 
    \\ &\stackrel{\dagger}{=}
    \mathbb{E}_{\bZ}\left[\frac{1}{n}
        \sum_{i=1}^{n}
    \left(\prod_{j=1}^{\widetilde{M}}
    \mathbb{I}
    \{f^j(\bZ,\bA^\ast_i)=c_\ell\}\right)
    \frac{1}{p_i^{(\mathcal{F})}(c_\ell)}\widetilde{Y}_i(c_\ell)\right] 
    \\ &=
    \frac{1}{n}
        \sum_{i=1}^{n}
    \mathbb{E}_{\bZ}\left[\prod_{j=1}^{\widetilde{M}}
    \mathbb{I}
    \{f^j(\bZ,\bA^\ast_i)=c_\ell\}\right]
    \frac{1}{p_i^{(\mathcal{F})}(c_\ell)}\widetilde{Y}_i(c_\ell) 
    \\ &=
    \frac{1}{n}
        \sum_{i=1}^{n}
    p_i^{(\mathcal{F})}(c_\ell)
    \frac{1}{p_i^{(\mathcal{F})}(c_\ell)}\widetilde{Y}_i(c_\ell)
    \\ &=
    \frac{1}{n}
        \sum_{i=1}^{n} \widetilde{Y}_i(c_\ell)
    \\ &=
    \mu(c_\ell)
    \end{split}
\end{equation*}
Where $\dagger$ follows from the fact that $\sum_{k = 1}^{L}\mathbb{I}\{f^\ast(\bZ, \bA^{\ast}_i) = c_\ell\}\mathbb{I}\{f^\ast(\bZ, \bA^{\ast}_i) = c_k\}\widetilde{Y}_i(c_k)=\mathbb{I}\{f^\ast(\bZ, \bA^{\ast}_i) = c_\ell\}\widetilde{Y}_i(c_\ell)$. Moreover, if $f^\ast$ is not unique (i.e., $\mathscr{F}^\ast$ is not a singleton), the sum $\sum_{k = 1}^{L}\mathbb{I}\{f^\ast(\bZ, \bA^{\ast}_i) = c_k\}\widetilde{Y}_i(c_k)$ will be equal for any $f^\ast \in \mathscr{F}^\ast$.
\end{proof}

\section{Simulations and data analysis}
\label{apdx.sec:simulations_and_data_analyses}
The \texttt{R} package implementing our methodology is available at \url{https://github.com/barwein/Misspecified_Interference}. Simulations and data analysis reproducibility materials of the results are available at \url{https://github.com/barwein/CI-MIS}.

Throughout all the simulations and data analyses performed, the exposure probabilities $p_i$ (in each form) were estimated with $R=10^4$ re-sampling from the relevant $\Pr(\bZ = \bz)$. Formally, let $\bz_1,\dots,\bz_R$ denote the sampled treatments from $\Pr(\bZ = \bz)$. Define the indicator matrix $I(c_\ell) \in \mathbb{R}^{n\times R},\; \ell =1,\dots,L$ by 
$I_{ij}(c_\ell) = \mathbb{I} \{f(\bz_j,\bA_i)=c_\ell\},\; i=1,\dots,n, j=1,\dots,R$. The estimation of the exposures probabilities is performed via additive smoothing 
\citep{aronow_estimating_2017} 
\begin{equation*}
\reallywidehat{P}(c_\ell) = 
\frac{I(c_\ell)I(c_\ell)^T + I_n}{R+1},
\end{equation*}
where $I_n$ is the $n\times n$ identity matrix, and $\reallywidehat{P}(c_\ell)$ is the estimator of $P(c_\ell)$ defined by 
\begin{equation*}
    P_{ij}(c_\ell) =
    \begin{cases}
    p_i^{(\bA)}(c_\ell),& i=j \\
    p_{ij}^{(\bA)}(c_\ell),& i\neq j
    \end{cases}
\end{equation*}

To express network similarity we utilized the Jaccard index. Let $\mathcal{E}(\bA)$ be the edges set of network $\bA$. For two networks $\bA,\bA'$, the Jaccard index is defined by
\begin{equation*}
    J_{\bA,\bA'} = \frac{\big\vert \mathcal{E}(\bA) \cap \mathcal{E}(\bA') \big\vert}{\big\vert \mathcal{E}(\bA) \cup \mathcal{E}(\bA')\big\vert},
\end{equation*}
that is, $J_{\bA,\bA'}$ is the proportion of shared edges between $\bA$ and $\bA'$ to the total number of edges in $\bA$ or $\bA'$. Thus, $0 \leq J_{\bA,\bA'} \leq 1$, where values close to $1$ indicates that the networks are similar. 

\subsection{Simulations}
In the simulations, a PA network of $n=3000$ units was sampled as the baseline true network via the \texttt{igraph} package \url{https://igraph.org/r/} with power parameter set to 1 \citep{Barabasi1999}. Figure~\ref{fig:sw.deg.plot} displays the degree distribution of the sampled network. Clearly, the degrees distribution implies a heavy right tail, a property inherent in the PA algorithm which is known to generate degrees that are asymptotically Pareto distributed  \citep{Barabasi1999}.
\begin{figure}[H]
    \centering
    \includegraphics[scale=0.08]{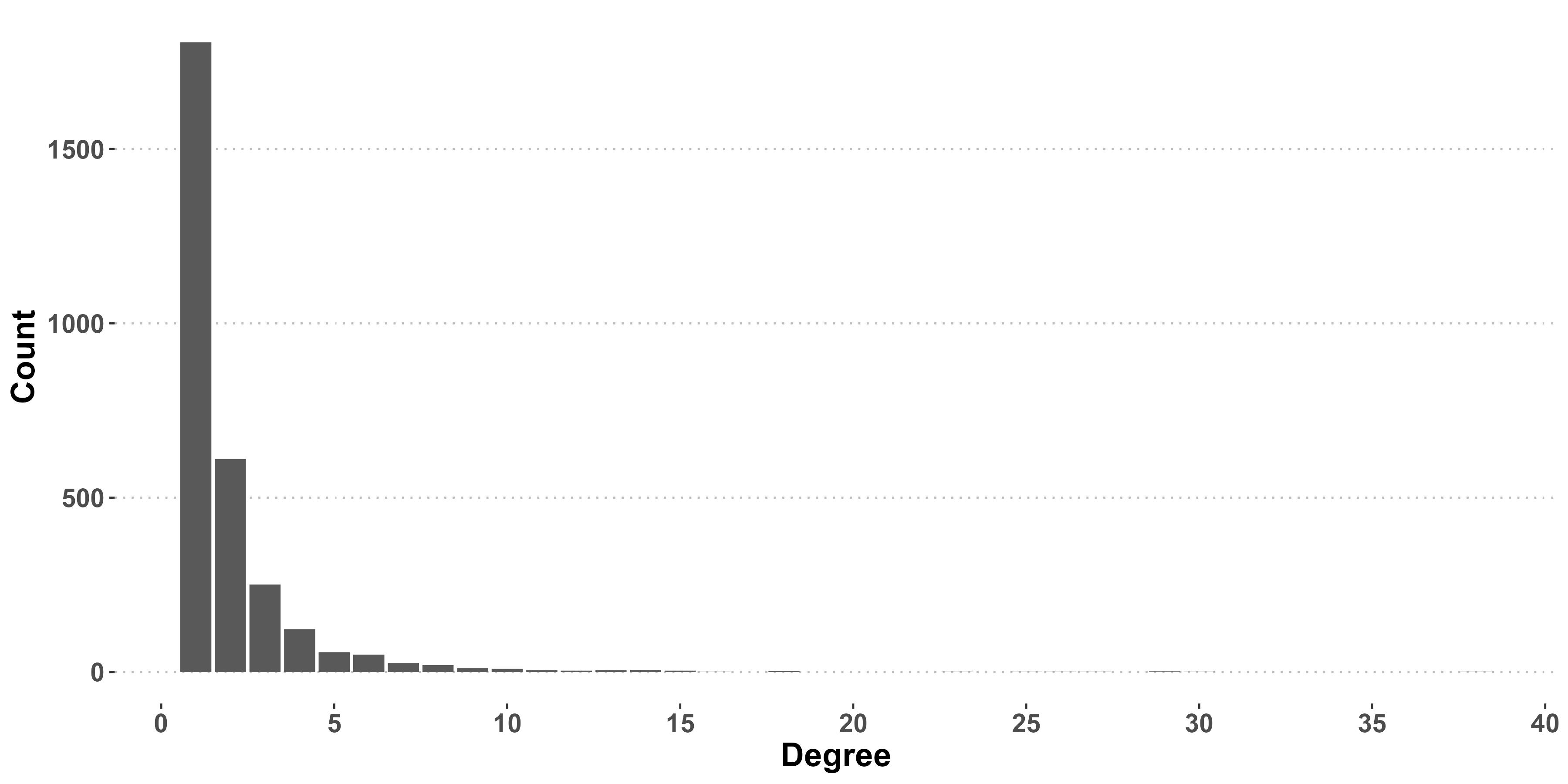}
    \caption{Histogram of the baseline preferential attachment random network degree's distribution. $n=3000$ nodes. The mean degree is $2$, and the maximal degree is $38$.}
    \label{fig:sw.deg.plot}
\end{figure}

\subsubsection{Illustration of the estimation bias}

In this subsection, we report additional results of the simulation study shown in the main text.
\noindent
\paragraph{Scenario (I) (Incorrect reporting of social connections).}
Figure~\ref{fig:sim.bias.additional} shows the absolute bias for additional estimands not displayed in the main text. The results were similar. When $\eta=0$ the bias is zero and increases with $\eta$ otherwise. Moreover, Figure~\ref{fig:sim.bias.additional} also shows the exact bias, as derived from Theorem 1, in comparison to the empirical bias of HT and Hajek estimators. The two are similar.

As discussed in Theorem 1, the bias from using a misspecified network structure results from selecting the wrong units and using invalid weights. Selecting the wrong units in our framework is equivalent to embedding units with the wrong exposure values. 
Figure~\ref{fig:n.missclass.exposure} shows the number of units with misclassified exposure values in the simulation. Clearly, the number of misclassified exposures increases with $\eta$, regardless of the exposure value.

The simulation validated Theorem 1 by illustrating that both Hajek and HT are unbiased whenever the network is correctly specified ($\eta=0$). However, HT had a larger empirical standard deviation (SD) than Hajek, possibly due to the stabilizing character of estimating $n$ when using Hajek \citep{Saerndal2003}. Figure~\ref{fig:sim_sw_SD} shows the 
empirical SD of the two estimators. We can conclude that even though both HT and Hajek had a similar bias, Hajek had a lower SD.

To quantify the similarity of $\bA^\ast$ and each of the misspecified networks, the Jaccard index was computed. Table~\ref{tab.apdx:jaccard_noise} displays the Jaccard index of $\bA^\ast$ with each sampled network (by $\eta$). In the extreme ($\eta=0.25$), there were only about $16\%$ shared edges in the networks.

In the simulation, we sampled one incorrect network for each $\eta >0$ value. To illustrate that the results are robust for replications, Figure~\ref{fig:sim_noise_spagehti} displays the results of additional $50$ replications in each we sampled different incorrect network. The bias across all replications is similar.

\begin{figure}[H]
    \centering
    \includegraphics[scale=0.055]{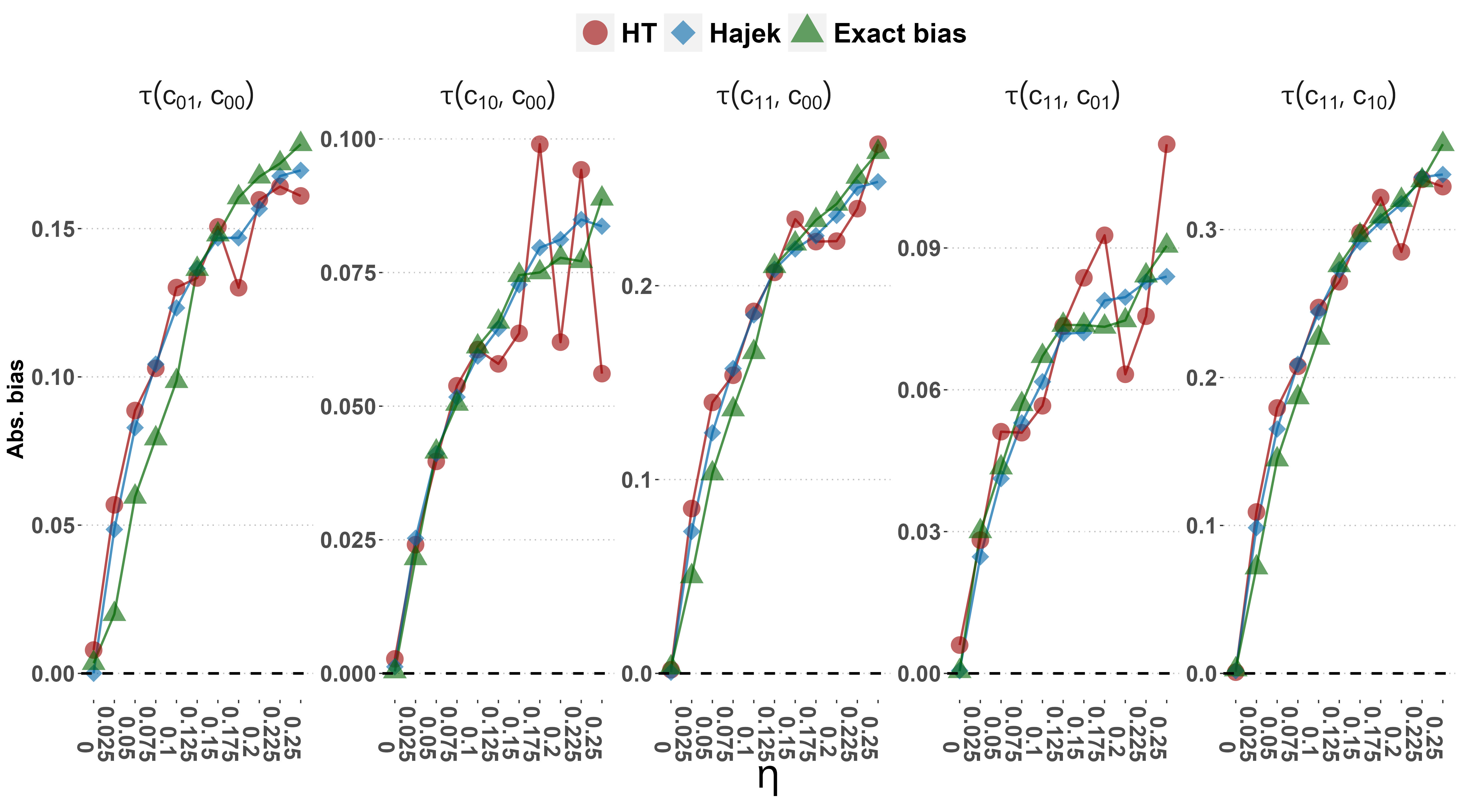}
    \caption{Scenario (I). Additional absolute bias results from estimating $\tau(c_{01},c_{00}),\; \tau(c_{11},c_{00}), \; \tau(c_{11},c_{01}),\; \tau(c_{11},c_{10})$.}
    \label{fig:sim.bias.additional}
\end{figure}

\begin{figure}[H]
    \centering
    \includegraphics[scale=0.055]{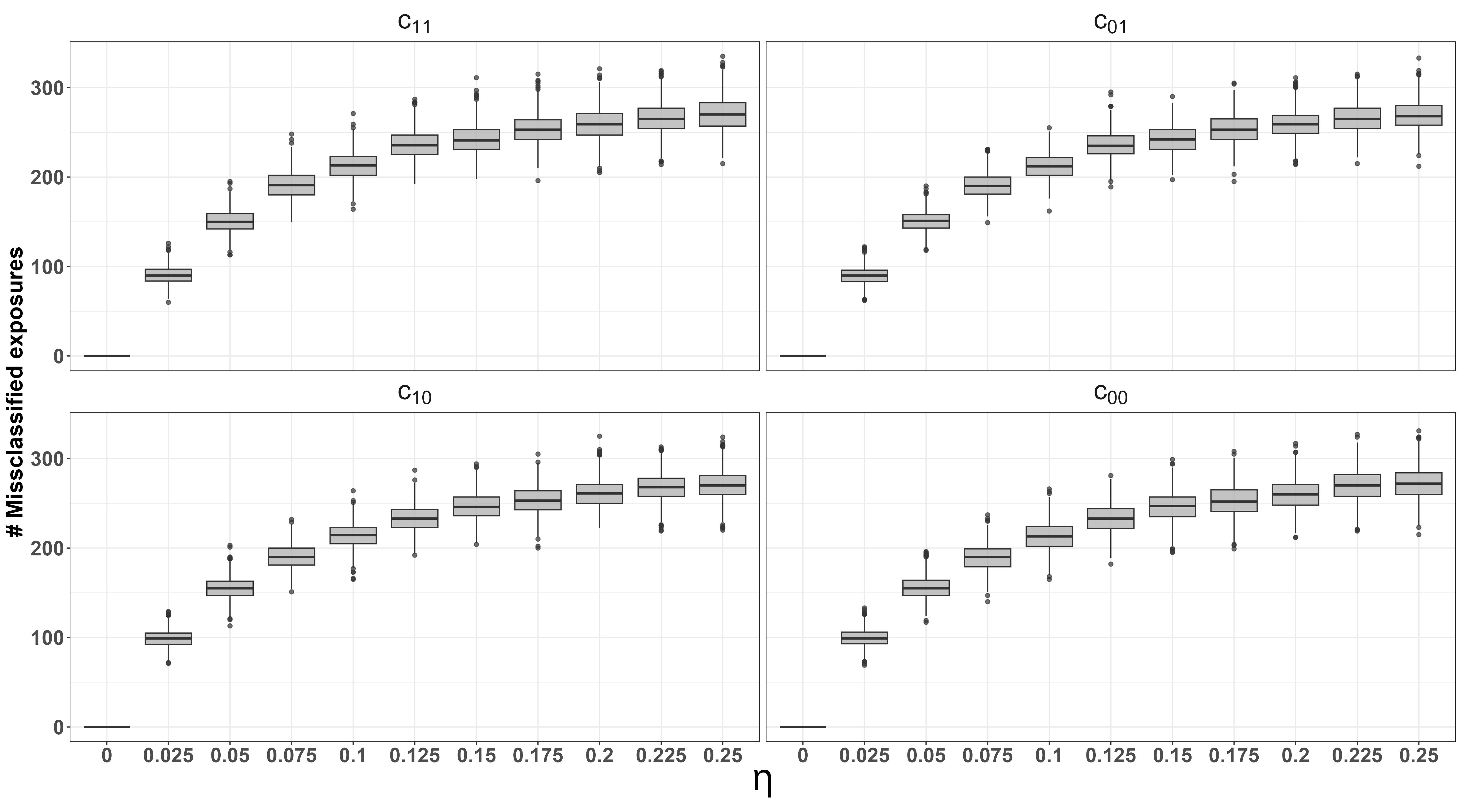}
    \caption{Number of units with misclassified exposures by exposure value in Scenario (I).}
    \label{fig:n.missclass.exposure}
\end{figure}

\begin{figure}[H]
    \centering
    \includegraphics[scale=0.055]{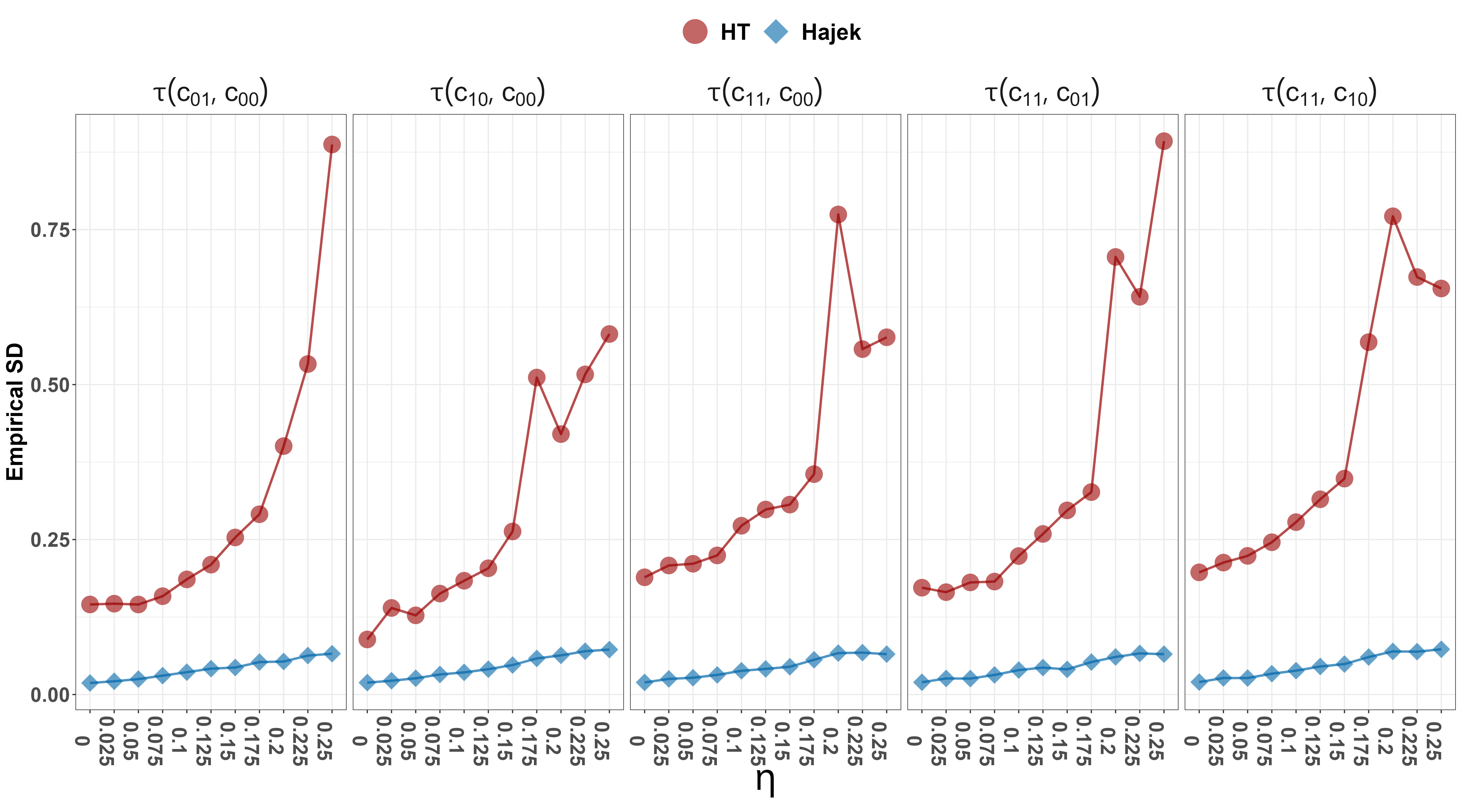}
    \caption{Empirical standard deviation (SD) of HT and Hajek estimators in Scenario (I).}
    \label{fig:sim_sw_SD}
\end{figure}

\begin{table}[H]
    \centering
    \begin{tabular}{lrrrrrrrrrrr}
    \toprule
    $\eta$ & 0 & 0.025 & 0.050 & 0.075 & 0.100 & 0.125 & 0.150 & 0.175 & 0.200 & 0.225 & 0.250\\
    $J_{\bA^\ast,\bA}$ & 1 & 0.713 & 0.545 & 0.437 & 0.365 & 0.299 & 0.261 & 0.231 & 0.203 & 0.175 & 0.163\\
    \bottomrule
    \end{tabular}
    \caption{Jaccard index of $\bA^\ast$ and the misspecified networks in Scenario (I).}
    \label{tab.apdx:jaccard_noise}
\end{table}

\begin{figure}[H]
    \centering
    \includegraphics[scale=0.04]{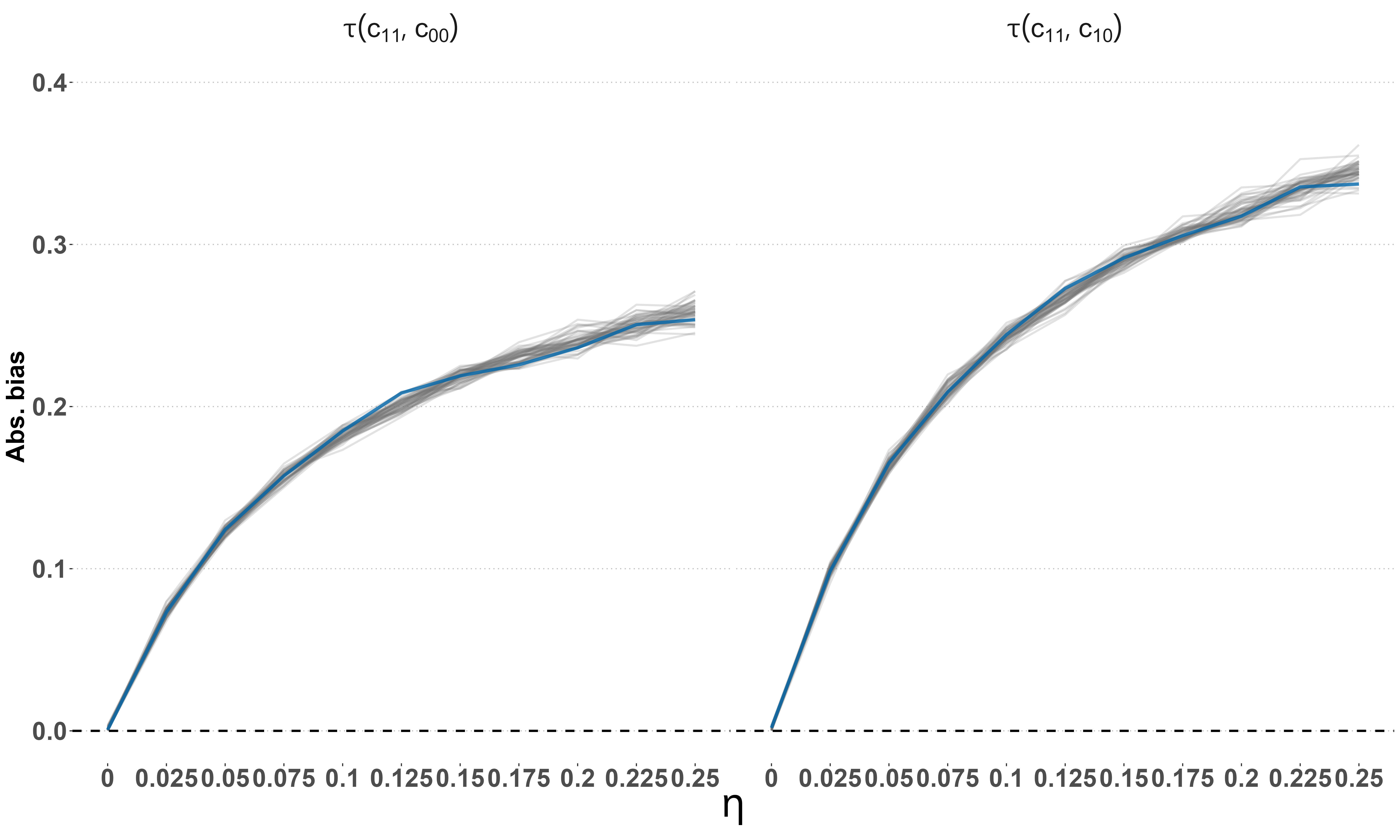}
    \caption{Multiple replications of Scenario (I). The blue line represents the absolute bias of Hajek estimates shown in the main text, whereas each grey line results from the $50$ additional replication in which different networks are sampled for each $\eta >0$.}
    \label{fig:sim_noise_spagehti}
\end{figure}

\noindent
\paragraph{Scenario (II) (Censoring).}
Here we also report additional results similar to the ones reported in the previous scenario.
Table~\ref{tab.apdx:prop.censored} shows the proportion of units with more than $K=1,\dots,7$ neighbors, i.e., the proportion of units we censored some of their edges for each of the thresholds. For example, when $K=7$ only about $2.5\%$ units had censored edges, whereas when $K=1$ almost $40\%$ of units had censored edges.
Figure~\ref{fig:sim.bias.additional.scenario2} shows absolute bias for additional estimands not shown in the main text. The same picture holds. When the censoring threshold $K$ decreases, the bias increases, and the bias is larger. Notice that HT had a larger bias than Hajek when the censoring threshold $K$ decreased, probably due to the smaller effective sample size and the weight stability of Hajek. Furthermore, the exact bias is also displayed and is similar to the empirical bias of HT and Hajek. Figure~\ref{fig:n.missclass.exposure.scenario2} displays the number of units with misclassified exposure values by censoring threshold $K$. Figure~\ref{fig:sim_sw_SD.scenario2} shows the empirical SD of HT and Hajek estimators in Scenario (II). Here also the SD of HT is uniformly higher than Hajek. However, the SD of HT decreases with $K$, i.e., when more censoring is present the variance is reduced. Table~\ref{tab.apdx:jaccard_noise.scenario2} provides the Jaccard index of $\bA^\ast$ and each of the censored networks. Similarly to Scenario (I), the index decreases with $K$. Figure~\ref{fig:sim_censor_spagehti} shows that the results from additional $50$ replications of the simulations are almost identical for those reported. 
 
\begin{table}[H]
    \centering  
    \begin{tabular}{lrrrrrrr}
    \toprule
    $K$ & 1 & 2 & 3 & 4 & 5 & 6 & 7\\
    $\Pr(d_i(\bA^\ast) > K)$ & 0.398 & 0.194 & 0.111 & 0.07 & 0.051 & 0.034 & 0.025\\
    \bottomrule
    \end{tabular}
    \caption{Edges empirical right-tail function in the PA network $\bA^\ast$. $d_i(\bA^\ast)$ is the degree of unit $i$.}
    \label{tab.apdx:prop.censored}
\end{table}

\begin{figure}[H]
    \centering
    \includegraphics[scale=0.055]{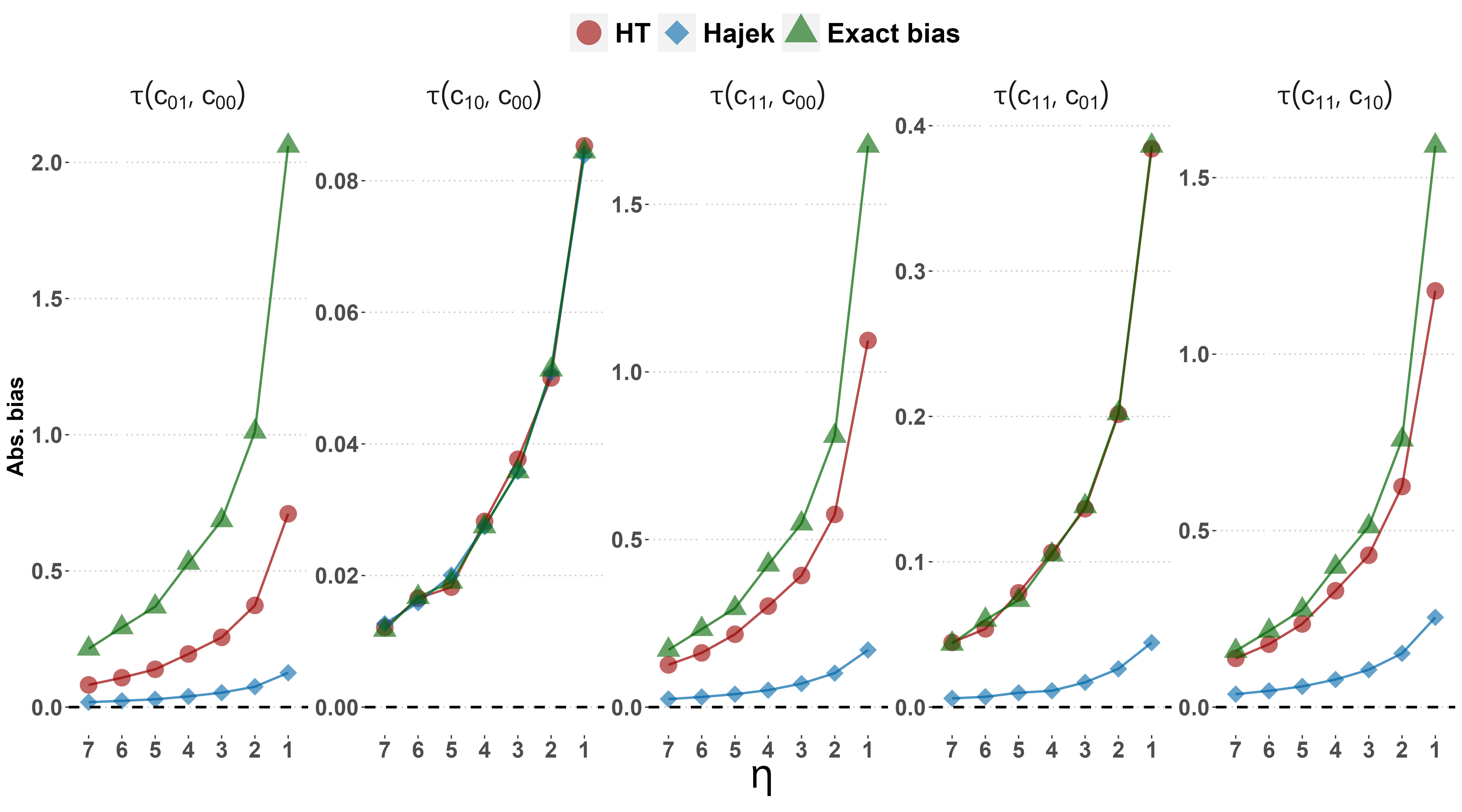}
    \caption{Scenario (II). Additional absolute bias results from estimating $\tau(c_{01},c_{00}),\; \tau(c_{11},c_{00}), \; \tau(c_{11},c_{01}),\; \tau(c_{11},c_{10})$.}
    \label{fig:sim.bias.additional.scenario2}
\end{figure}

\begin{figure}[H]
    \centering
    \includegraphics[scale=0.055]{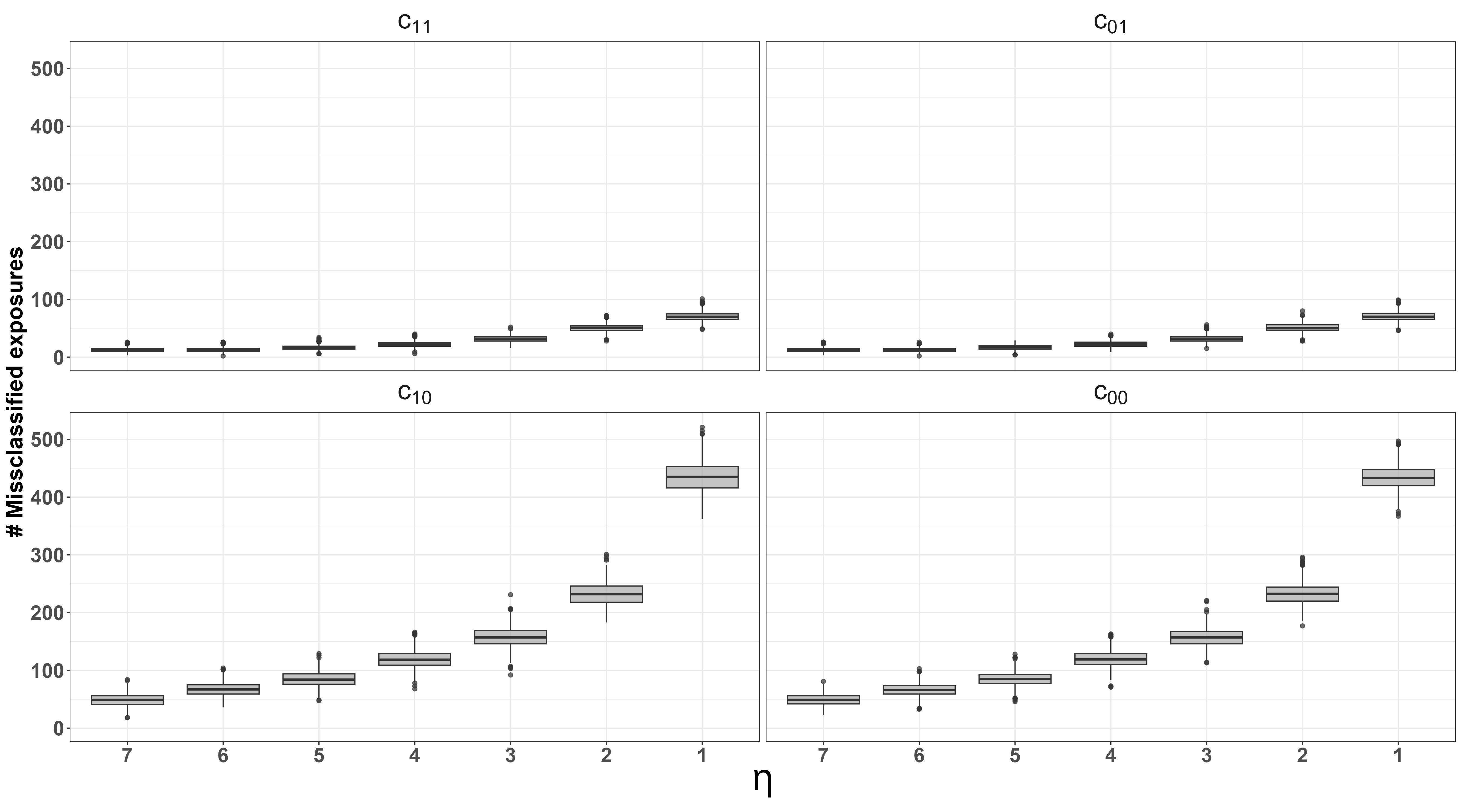}
    \caption{Number of units with misclassified exposures by exposure value in Scenario (II).}
    \label{fig:n.missclass.exposure.scenario2}
\end{figure}

\begin{figure}[H]
    \centering
    \includegraphics[scale=0.055]{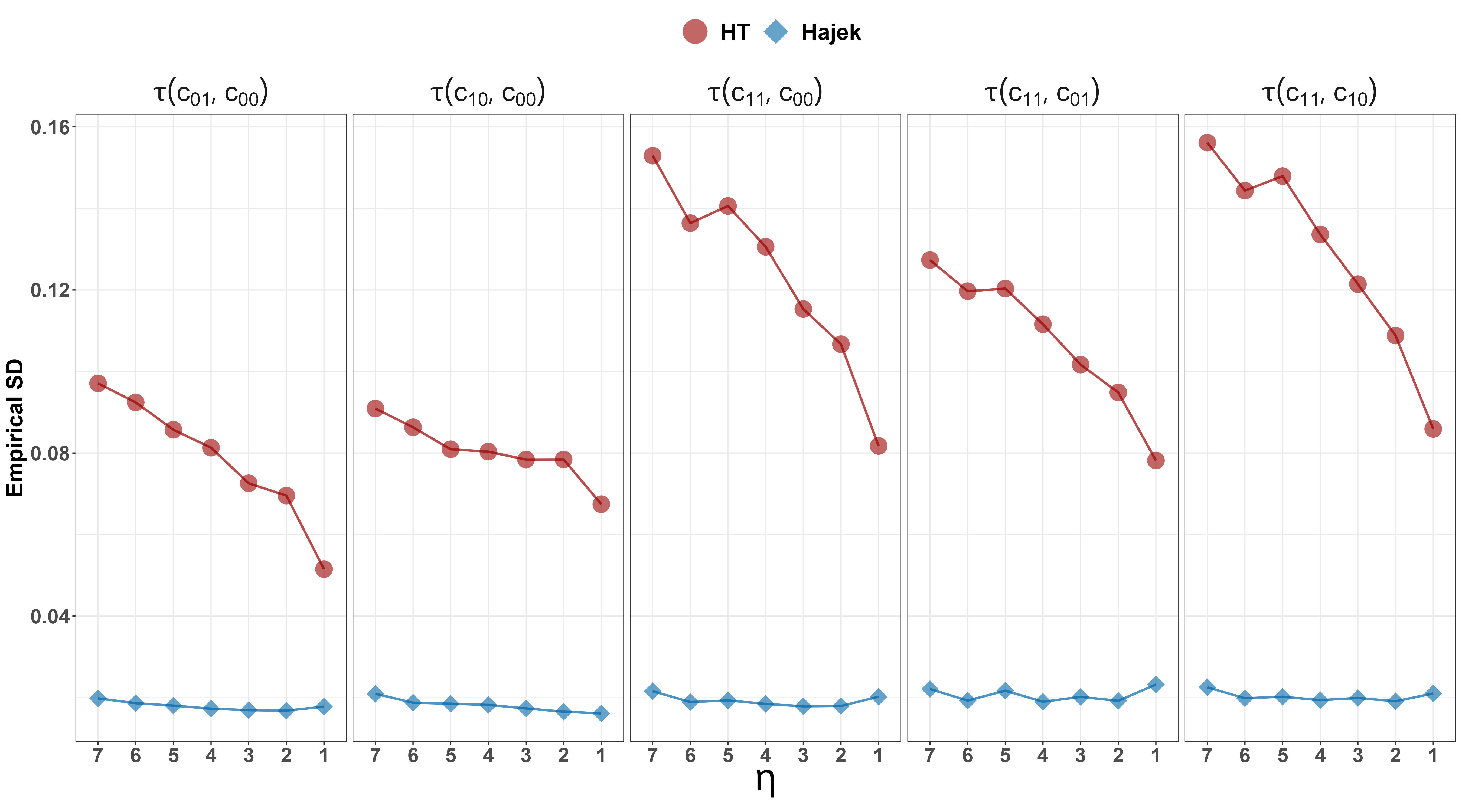}
    \caption{Empirical standard deviation (SD) of HT and Hajek estimators in Scenario (II).}
    \label{fig:sim_sw_SD.scenario2}
\end{figure}

\begin{table}[H]
    \centering  
    \begin{tabular}{lrrrrrrr}
    \toprule
    $K$ & 7 & 6 & 5 & 4 & 3 & 2 & 1\\
    $J_{\bA^\ast,\bA}$ & 0.866 & 0.835 & 0.792 & 0.730 & 0.646 & 0.509 & 0.258\\
    \bottomrule
    \end{tabular}
    \caption{Jaccard index of $\bA^\ast$ and the censored networks in Scenario (II).}
    \label{tab.apdx:jaccard_noise.scenario2}
\end{table}

\begin{figure}[H]
    \centering
    \includegraphics[scale=0.04]{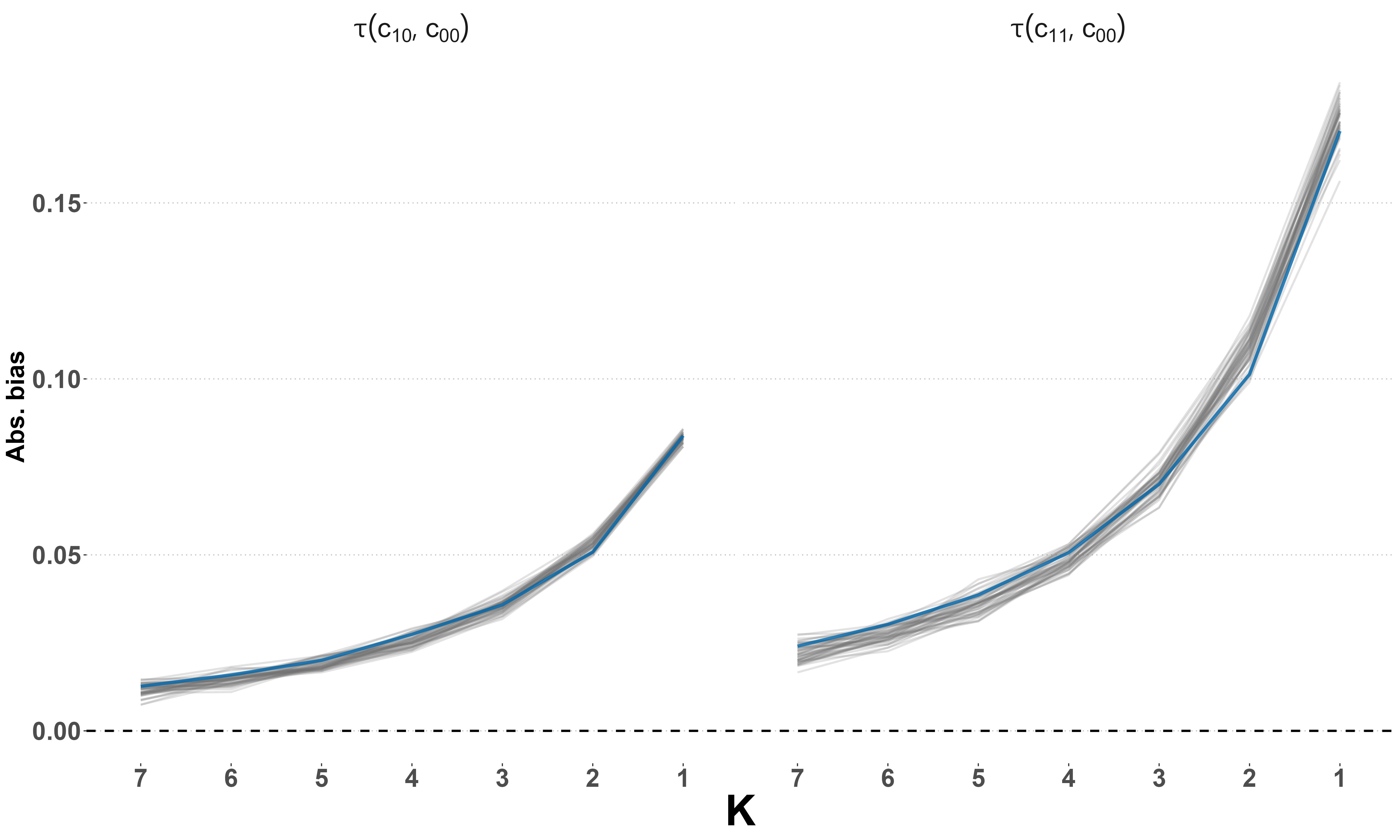}
    \caption{Multiple replications of Scenario (II). The blue line represents the absolute bias of Hajek estimates shown in the main text, whereas each grey line results from the $50$ additional replication in which different networks are sampled for each $K$. }
    \label{fig:sim_censor_spagehti}
\end{figure}

\subsubsection{Bias-variance tradeoff of the NMR estimators}
Figure~\ref{fig:mr.bias.var.apdx.plot} displays additional results of the bias-variance tradeoff simulation for $\tau(c_{01},c_{00})$ and $\tau(c_{11},c_{00})$. Similar results to those given in the main text appear there. Table~\ref{tab.apdx:jaccard_noise.nmr.bias.var} shows the pairwise Jaccard indices of all six networks used in the simulation. Figure~\ref{fig:nmr.coverage} shows the empirical $95\%$ coverage of the Hajek NMR estimator in estimating $\tau(c_{11}, c_{10})$. The confidence interval is computed with a normal approximation (Web Appendix~\ref{apdx.sec:asymptotic}) and conservative variance estimator (Web Appendix~\ref{apdx.sec:NMR_var}). NMR with the correct network achieves nominal coverage in each setup, whereas NMR with incorrect networks achieves nominal coverage only when $M\geq2$ networks are used.  Figure~\ref{fig:nmr_neu} shows the Number of Effective Units (NEU) of the NMR estimator in different combinations of networks $\mathcal{A}$. As expected, NEU decreases with the number of networks used (regardless of whether the true network is included), but the rate of decline is non-linear in the number of networks, where the slope decreases in this setup.

\begin{table}[H]
    \centering  
    \begin{tabular}{lrrrrrr}
    \toprule
      & $\bA^\ast$ & $\bA^a$ & $\bA^b$ & $\bA^c$ & $\bA^d$ & $\bA^e$
      \\
    \midrule
    $\bA^\ast$ & 1 &  &  &  &  & \\
    $\bA^a$ & 0.156 & 1 &  &  &  & \\
    $\bA^b$ & 0.155 & 0.066 & 1 &  &  & \\
    $\bA^c$ & 0.159 & 0.067 & 0.066 & 1 &  & \\
    $\bA^d$ & 0.157 & 0.067 & 0.068 & 0.068 & 1 & \\
    $\bA^e$ & 0.157 & 0.067 & 0.066 & 0.068 & 0.068 & 1\\
    \bottomrule
    \end{tabular}
    \caption{Jaccard index of the networks used in the simulations of the NMR bias-variance tradeoff.}
    \label{tab.apdx:jaccard_noise.nmr.bias.var}
\end{table}

\begin{figure}[H]
    \centering
    \includegraphics[scale=0.055]{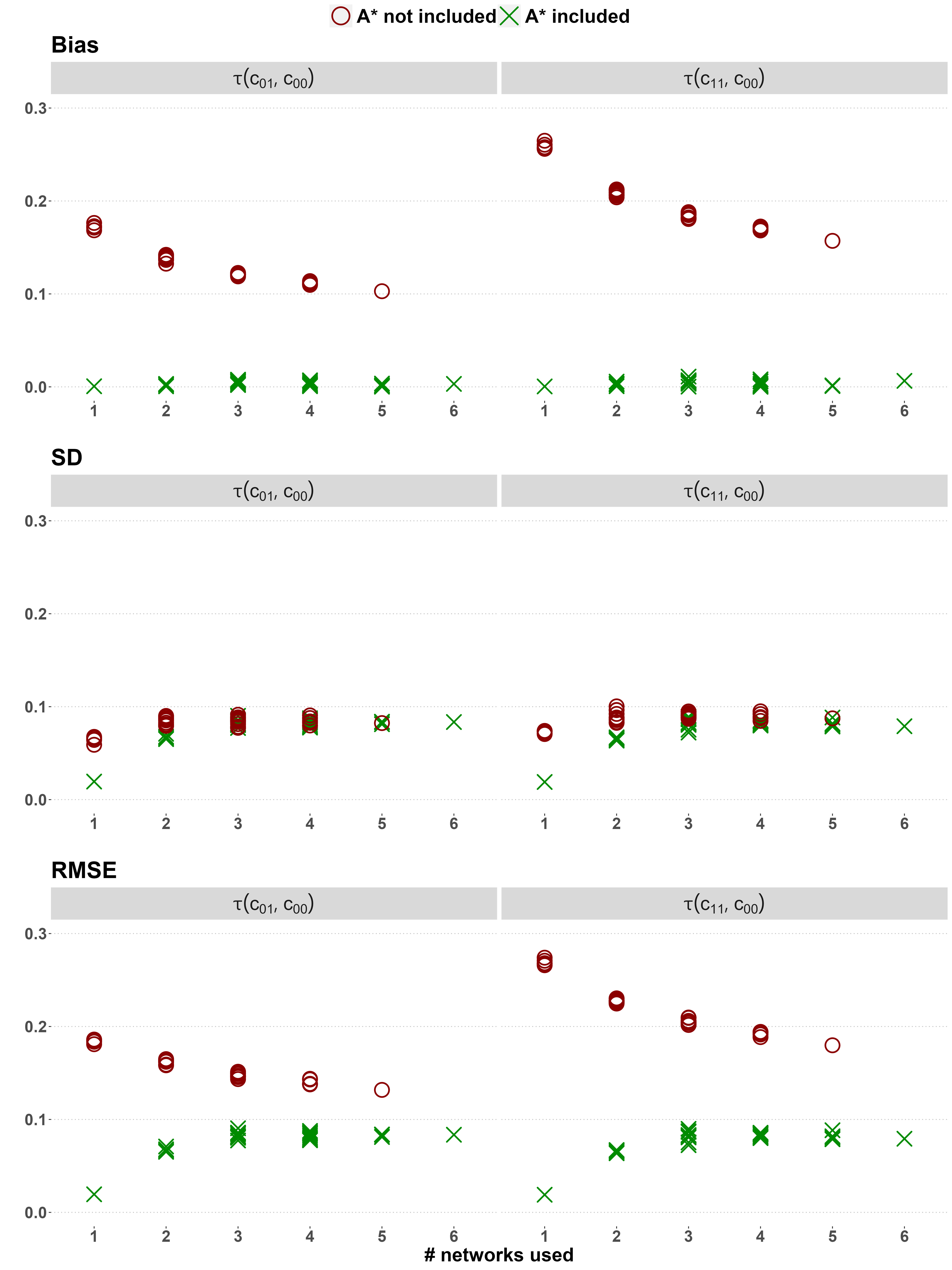}
    \caption{Bias-variance tradeoff of the NMR estimator. The results presented are the absolute bias, SD, and RMSE estimates of the Hajek NMR estimator. True causal effects are $\tau(c_{01},c_{00})=0.25$ and $\tau(c_{11},c_{00})=1$.}
    \label{fig:mr.bias.var.apdx.plot}
\end{figure}

\begin{figure}[H]
    \centering
    \includegraphics[width=0.6\linewidth]{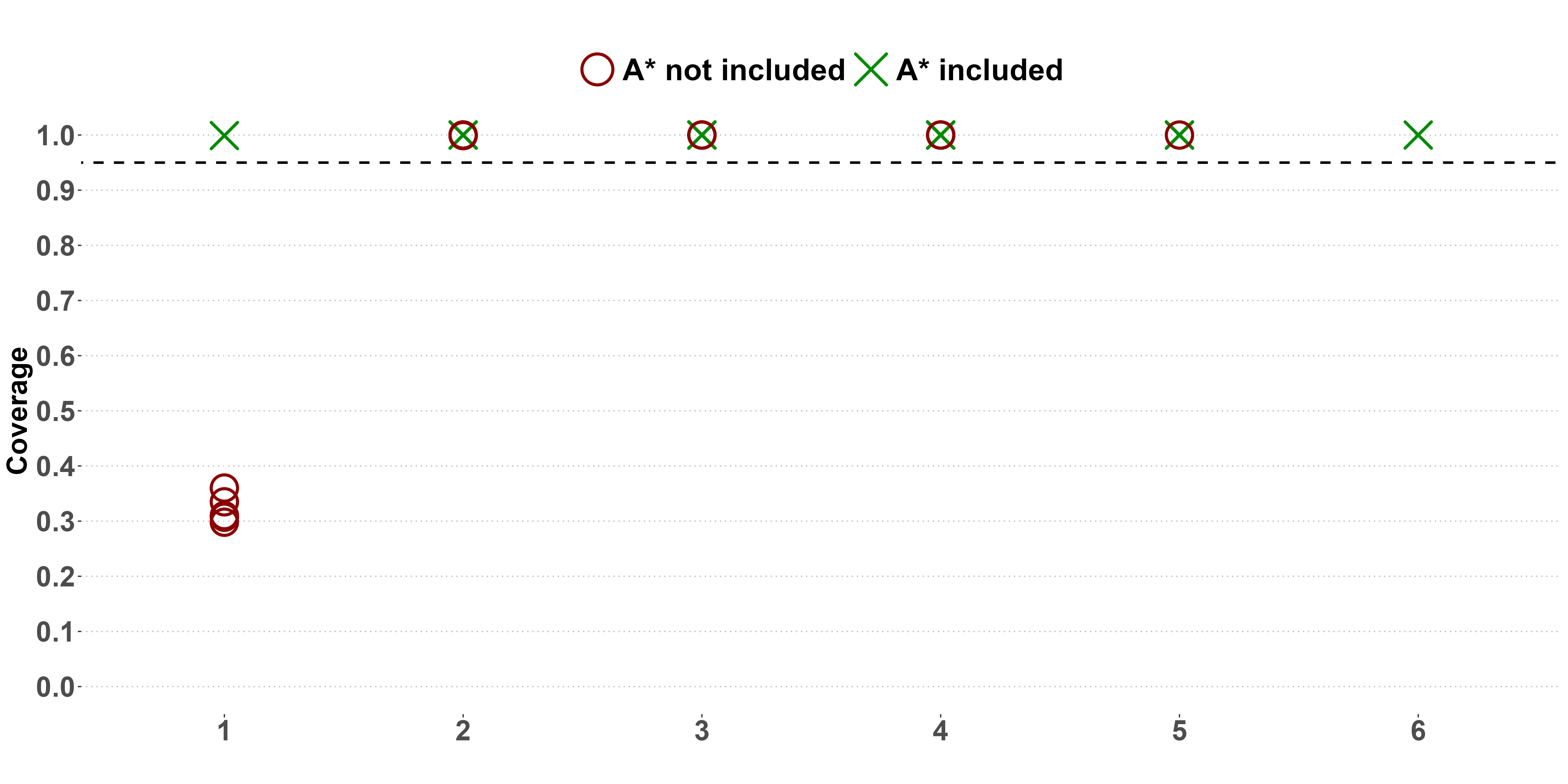}
    \caption{Empirical $95\%$ coverage of the Hajek NMR estimator with the conservative variance estimator. Coverage is defined as the proportion of iterations where the $95\%$ confidence interval contained the true estimand $\tau(c_{11},c_{10})$. The confidence interval is computed with a normal approximation (Web Appendix~\ref{apdx.sec:asymptotic}) and the conservative variance estimator (Web Appendix~\ref{apdx.sec:NMR_var}).}
    \label{fig:nmr.coverage}
\end{figure}

\begin{figure}[H]
    \centering
    \includegraphics[width=0.8\linewidth]{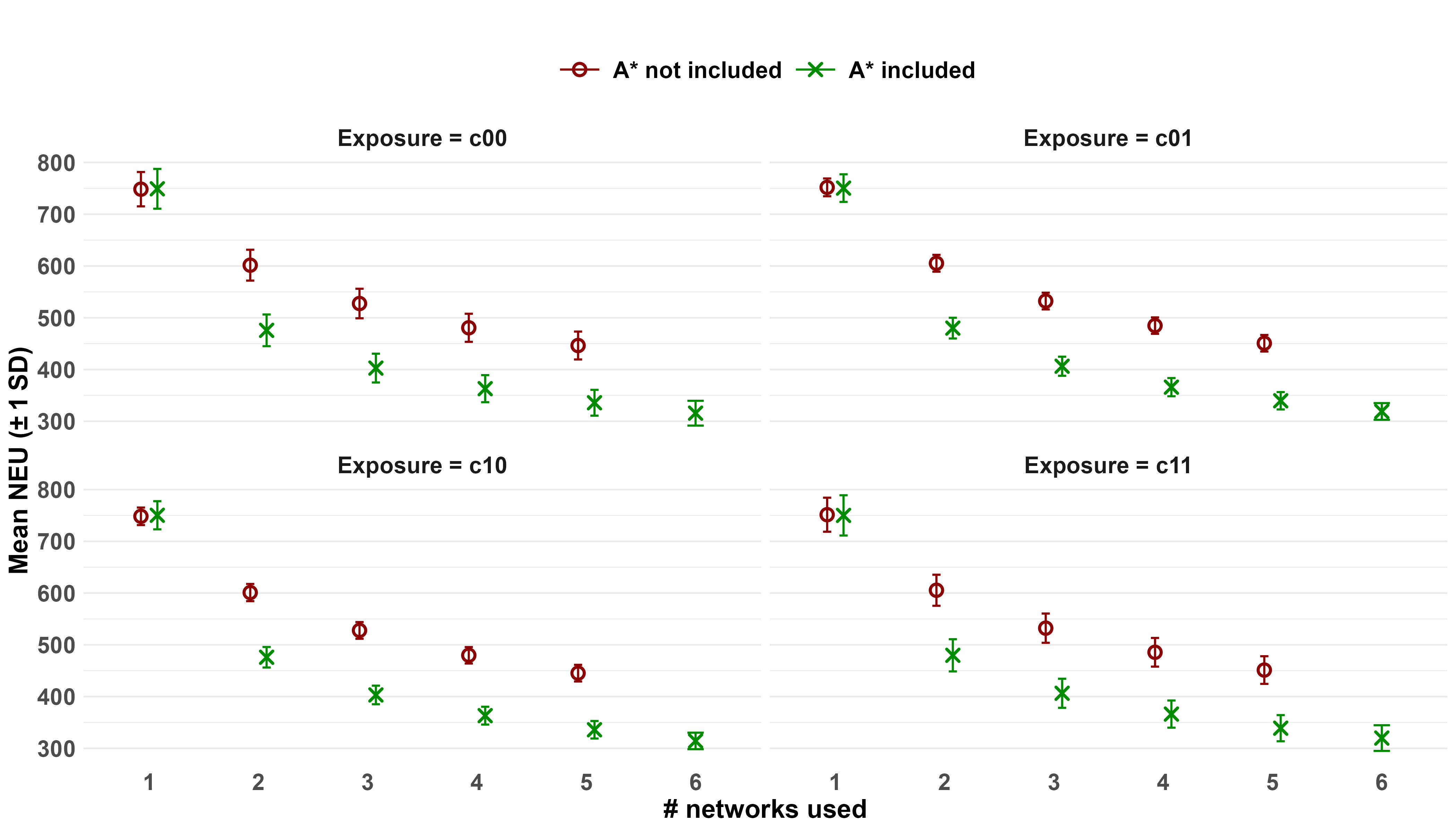}
    \caption{Mean $\pm$ SD of the Number of Effective Units (NEU) used in the NMR estimator across $1000$ iterations. NEU is defined by $\text{NEU}(\mathcal{A},c_k) = \sum_{i=1}^{n} I^{(\mathcal{A})}_i(\bZ,c_k)$ and represents the number of units used in the NMR estimator. In this figure, we aggregated combinations $\mathcal{A}$ that did not contain the true network $\bA^\ast$. As more networks are used, NEU decreases as fewer units have the same exposure value across all networks. However, the decrease is non-linear. For example, increasing from $1$ to $2$ networks yielded a steeper decline in NEU than the move from $2$ to $3$.}
    \label{fig:nmr_neu}
\end{figure}

Furthermore, we repeat the simulation in realistic quasi-experimental settings by taking $\mathcal{A}$ to consists of the four available networks from \citet{Paluck2016} study, as analyzed in the data analysis section in the main text. The correct network $\bA^\ast$ is taken to be the ST-pre network, which is the main network in \citet{Paluck2016} analysis. We used the same DGP to generate treatments and outcomes as in the previously displayed bias-variance simulations of the NMR estimators. Figure \ref{fig:mr.bias.var.apdx.palluck.plot} displays the results from $1000$ replications. The results portray the bias-variance tradeoff inherent in the NMR estimators.

\begin{figure}[H]
    \centering
    \includegraphics[scale=0.055]{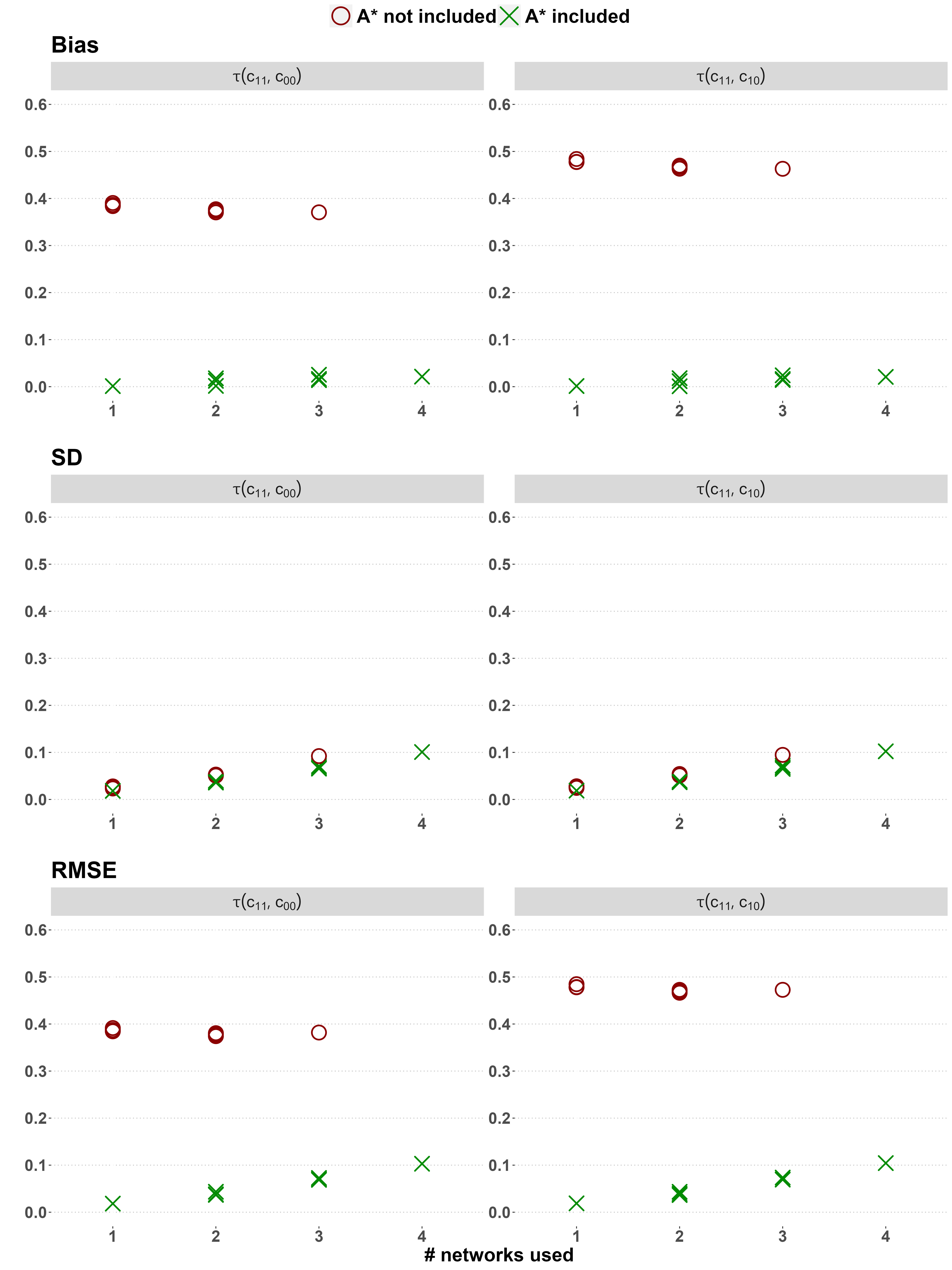}
    \caption{Bias-variance tradeoff of the NMR estimator with $\mathcal{A}$ being the four networks from \citet{Paluck2016} study. The results presented are the absolute bias, SD, and RMSE estimates of the Hajek NMR estimator. True causal effects are $\tau(c_{11},c_{00})=1$ and $\tau(c_{11},c_{10})=0.5$.}
    \label{fig:mr.bias.var.apdx.palluck.plot}
\end{figure}

\subsubsection{Conservative variance estimators}
We illustrate the conservative property of the NMR variance estimators proposed in Web Appendix~\ref{apdx.sec:NMR_var} in a small simulation study. In the same setup of the NMR bias-variance tradeoff simulation, we took all scenarios in which $\mathcal{A}$ contained the true networks $\bA^\ast$ and compared the estimated conservative SE to the empirical SD. Figure~\ref{fig:NMR_conserv_var} displays the mean SE/SD ratio of the overall effect $\tau(c_{11},c_{00})$ across the $1000$ iterations performed. Since all mean values are above one, we can surmise that the conservativeness property of the variance estimator holds. Nevertheless, it seems like the variance estimator is more conservative for Hajek than HT NMR estimators.

\begin{figure}[H]
    \centering
    \includegraphics[scale=.05]{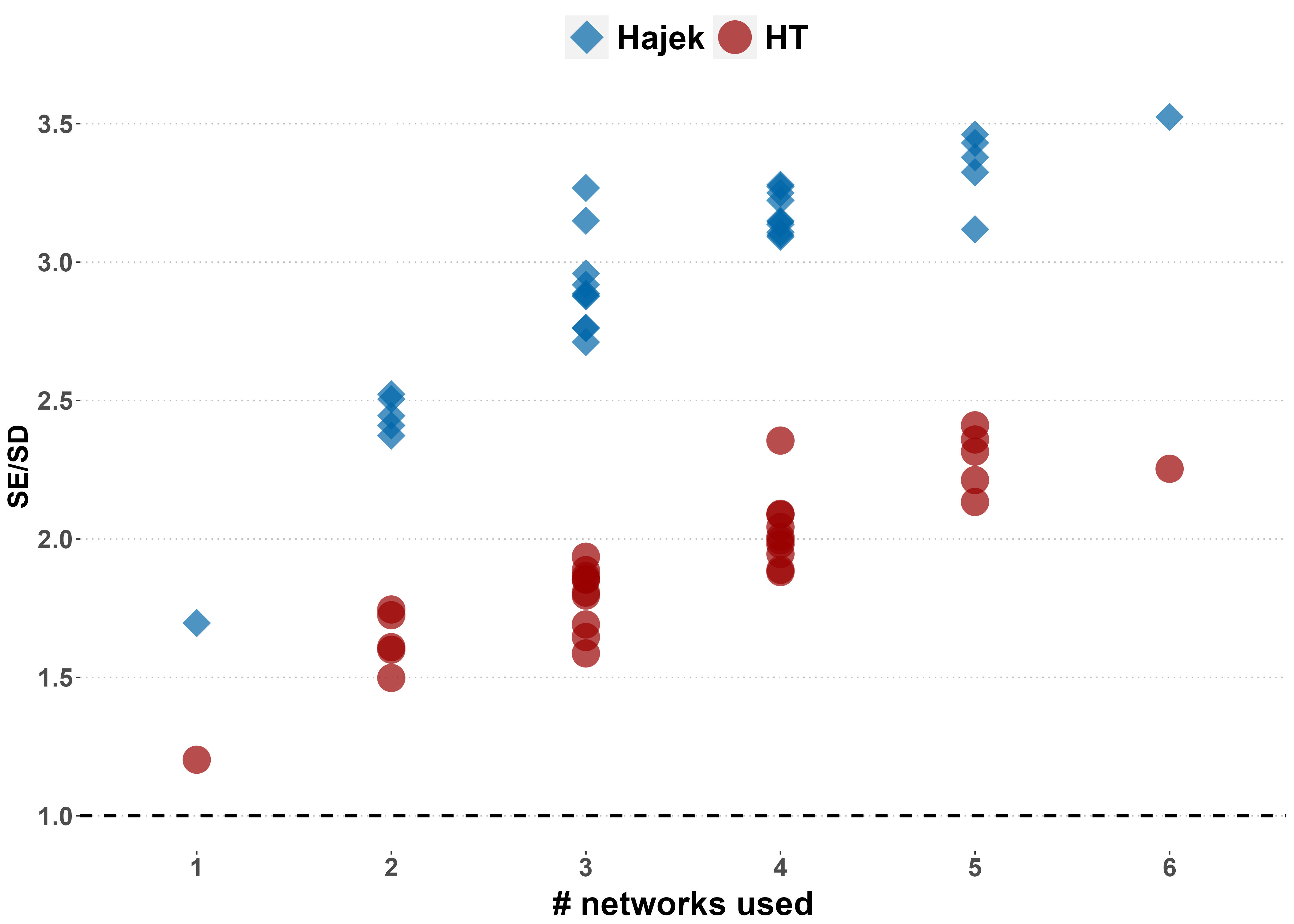}
    \caption{Conservative NMR variance estimator. Values are the mean of $\tau(c_{01},c_{00})$ estimated SE/SD.}
    \label{fig:NMR_conserv_var}
\end{figure}

\subsection{Data analysis}
In our analysis of the data, we performed the same data pre-processing conducted by \citet{Paluck2016}. The open-source replicability package provided by \cite{Paluck2016} can be found at \url{https://www.icpsr.umich.edu/web/ICPSR/studies/37070}.
Table~\ref{tab:palluck.apdx} is an extended version of the results displayed in the main text. It contains the estimation of two more estimands ($\tau(c_{011},c_{000}),\; \tau(c_{111},c_{000})$) using more networks combinations. For example, we also use the NMR with both the ST networks (measured at the two time periods) simultaneously.

Table~\ref{tab.apdx:jaccard.palluck} shows the Jaccard index of the four available networks. Clearly, networks derived from the same questions are more similar than those from different questions, e.g., the similarity of ST and ST-2 is $27.5\%$ whereas those of ST and BF is $21.1\%$. 

\begin{table}[H]
    \centering  
      \begin{tabular}{lrrrr}
    \toprule
      & ST-pre & ST-post & BF-pre & BF-post\\
    \midrule
    ST-pre & 1 &  &  & \\
    ST-post & 0.274 & 1 &  & \\
    BF-pre & 0.211 & 0.137 & 1 & \\
    BF-post & 0.137 & 0.200 & 0.244 & 1\\
    \bottomrule
    \end{tabular}
    \caption{Jaccard index of all the four available networks from \cite{Paluck2016}.}
    \label{tab.apdx:jaccard.palluck}
\end{table}

\newpage

\begin{table}[H]
    \centering
    \begin{turn}{90}
    \scalebox{0.7}{\begin{tabular}{lllllllll}
    \toprule
    \multicolumn{1}{c}{ } & \multicolumn{2}{c}{$\tau(c_{001},c_{000})$} & \multicolumn{2}{c}{$\tau(c_{011},c_{000})$} & \multicolumn{2}{c}{$\tau(c_{101},c_{000})$} & \multicolumn{2}{c}{$\tau(c_{111},c_{000})$} \\
    \cmidrule(l{3pt}r{3pt}){2-3} \cmidrule(l{3pt}r{3pt}){4-5} \cmidrule(l{3pt}r{3pt}){6-7} \cmidrule(l{3pt}r{3pt}){8-9}
    Networks & HT & Hajek & HT & Hajek & HT & Hajek & HT & Hajek\\
\midrule
    ST (pre) & 0.061 [-0.364, 0.486] & 0.146 [-0.176, 0.468] & 0.162 [-0.53, 0.854] & 0.122 [-0.415, 0.66] & 0.096 [-0.437, 0.628] & 0.271 [-0.102, 0.644] & 0.369 [-0.67, 1.409] & 0.272 [-0.467, 1.01]\\
    BF (pre) & 0.084 [-0.414, 0.581] & 0.123 [-0.26, 0.505] & 0.068 [-0.616, 0.753] & 0.162 [-0.315, 0.639] & 0.169 [-0.538, 0.877] & 0.265 [-0.233, 0.763] & 0.143 [-0.846, 1.131] & 0.292 [-0.34, 0.924]\\
    ST \& BF (pre) & 0.051 [-0.338, 0.44] & 0.134 [-0.162, 0.431] & 0.11 [-0.769, 0.99] & 0.135 [-0.494, 0.763] & 0.079 [-0.406, 0.564] & 0.261 [-0.081, 0.603] & 0.224 [-1.006, 1.453] & 0.258 [-0.563, 1.078]\\
    ST (post) & 0.06 [-0.36, 0.479] & 0.131 [-0.189, 0.452] & 0.137 [-0.607, 0.881] & 0.13 [-0.424, 0.683] & 0.116 [-0.469, 0.701] & 0.252 [-0.163, 0.668] & 0.251 [-0.755, 1.257] & 0.246 [-0.45, 0.943]\\
    BF (post) & 0.09 [-0.424, 0.604] & 0.135 [-0.258, 0.527] & 0.039 [-0.486, 0.565] & 0.09 [-0.291, 0.47] & 0.17 [-0.539, 0.88] & 0.258 [-0.243, 0.76] & 0.134 [-0.84, 1.108] & 0.297 [-0.324, 0.917]\\
    ST pre \& post & 0.037 [-0.296, 0.37] & 0.139 [-0.115, 0.393] & 0.231 [-0.716, 1.178] & 0.133 [-0.591, 0.858] & 0.071 [-0.386, 0.528] & 0.296 [-0.02, 0.613] & 0.469 [-0.828, 1.766] & 0.25 [-0.716, 1.215]\\
    BF pre \& post & 0.077 [-0.4, 0.555] & 0.124 [-0.242, 0.491] & 0.037 [-0.494, 0.569] & 0.063 [-0.33, 0.455] & 0.15 [-0.515, 0.815] & 0.258 [-0.213, 0.728] & 0.154 [-0.922, 1.23] & 0.303 [-0.382, 0.988]\\
    ALL & 0.04 [-0.306, 0.387] & 0.178 [-0.08, 0.435] & 0.067 [-0.711, 0.845] & 0.051 [-0.517, 0.62] & 0.046 [-0.326, 0.418] & 0.227 [-0.042, 0.496] & 0.266 [-1.306, 1.838] & 0.226 [-0.809, 1.262]\\
    \bottomrule
    \end{tabular}}
    \end{turn}
     \caption{Extended results of the social network field experiment analysis. Results are reported as point estimates ($95\%$ CI). Estimation is performed using the NMR HT and Hajek estimators.}
    \label{tab:palluck.apdx}
\end{table}

\end{document}